\documentclass[journal,12pt,onecolumn,draftclsnofoot,]{IEEEtran}

\usepackage[T1]{fontenc}

\usepackage{graphicx}
\usepackage{caption}
\usepackage{subcaption}
\usepackage{color}

\IEEEoverridecommandlockouts
\usepackage{mathtools}
\usepackage{flushend}
\usepackage{amsmath}
\usepackage{algorithmic,algorithm}
\usepackage{cite}

\allowdisplaybreaks

\usepackage{amsmath}
\usepackage{amsthm}

\newtheorem{theorem}{Theorem}

\newcommand{\Tau}{\mathrm{T}}

\begin{document}
%
\title{On the Benefits of Network-Level Cooperation in Millimeter-Wave Communications} 


\author{Cristian~Tatino,~\IEEEmembership{Student Member,~IEEE,}
        Nikolaos~Pappas,~\IEEEmembership{Member,~IEEE,}
        Ilaria~Malanchini,~\IEEEmembership{Member,~IEEE,}
        Lutz~Ewe,
        and~Di~Yuan,~\IEEEmembership{Senior Member,~IEEE}
\thanks{This work extends the preliminary study in~\cite{RelConfO}. This project has received funding from the European Union's Horizon 2020 research and innovation programme under the Marie Sklodowska-Curie grant agreement No. 643002.}}%

\maketitle

\begin{abstract}
Relaying techniques for millimeter-wave wireless networks represent a powerful solution for improving the transmission performance. In this work, we quantify the benefits in terms of delay and throughput for a random-access multi-user millimeter-wave wireless network, assisted by a full-duplex network cooperative relay. The relay is equipped with a queue for which we analyze the performance characteristics (e.g., arrival rate, service rate, average size, and stability condition). Moreover, we study two possible transmission schemes: fully directional and broadcast. In the former, the source nodes transmit a packet either to the relay or to the destination by using narrow beams, whereas, in the latter, the nodes transmit to both the destination and the relay in the same timeslot by using a wider beam, but with lower beamforming gain. In our analysis, we also take into account the beam alignment phase that occurs every time a transmitter node changes the destination node. We show how the beam alignment duration, as well as position and number of transmitting nodes, significantly affect the network performance. We additionally discuss the impact of beam alignment errors and imperfect self-interference cancellation technique at the relay for full-duplex communications. Moreover, we illustrate the optimal transmission scheme (i.e., broadcast or fully directional)  for several system parameters and show that a fully directional transmission is not always beneficial, but, in some scenarios, broadcasting and relaying can improve the performance in terms of throughput and delay. 

%
\end{abstract}


%
\IEEEpeerreviewmaketitle

\section{Introduction}
\label{sec:Intro}


In the recent years, millimeter-wave (mm-wave) communications have attracted the interest of many researchers, who see the abundance of spectrum resource in the mm-wave frequency range (30-300 GHz) as a possible solution to the longstanding problem of spectrum scarcity. 
For this reason, mm-wave wireless networks have been identified as one of the key enabler technologies for the next generation of mobile communications, i.e., 5G~\cite{5GEnabler}. Although mm-waves communications can reach tremendous high data rates~\cite{RappMeas}, the signal propagation is subject to higher path loss and penetration loss~\cite{Meas28,Meas73}, in comparison to lower frequency communications. Such high losses cause frequent interruptions, especially when obstacles block the signal path~\cite{SurMM}.
Directional communications and narrow beams provide high beamforming gains that contribute to mitigate the path loss issue. By using narrow beams, the transmitters focus the signal energy along only few directions and paths and, usually, the line-of-sight (LOS) path is characterized by the lowest path loss~\cite{Meas28,Meas73}. When the LOS path is blocked by an obstacle, the use of reflected transmission paths can overcome the blockage issue~\cite{Our1}. 

Other solutions for avoiding interruptions caused by blockages provide alternative transmission paths by using additional nodes, e.g., multi-connectivity \cite{MCcapacity,Our2} and relaying techniques. In the latter, a source node (user equipment, UE) transmit a packet to an intermediate node (relay) when the source-destination path is blocked. Though relaying and cooperative communications have been extensively analyzed for microwave frequencies~\cite{Coop1,Coop2,Coop3,CoopSpa,Rel1,MonCoop,CoopWC,Niko,Rel2,ITC30,CoopSurv,CoopSel}, mm-wave communications present some peculiarities, such as the use of narrow beams and the beam alignment phase, that make further analysis necessary. For instance, by using narrow beams, UEs might not be able to transmit simultaneously to both the relay and the destination node, which is usually the case with omnidirectional transmissions at lower frequencies. By using narrow beams in mm-waves, in each timeslot the UE may transmit a packet either to the destination or the relay. This is particularly true for UEs that are equipped with a single phased antenna array, which is a common solution for mm-wave communications in order to minimize the cost and energy consumption~\cite{PhaseArray}. Moreover, every time the UEs change the receiver (i.e., from the destination to the relay and vice-versa), a new beam alignment might be required~\cite{InitAcc1,InitAcc2}. This can both cause further delays and affect the throughput.

We propose a novel analysis of network-level cooperative communications in mm-wave wireless networks with a mm-wave access point (mmAP) as transmission target and one network cooperative full-duplex relay that is equipped with a queue. We analyze the impact of  directional communications by evaluating two possible transmission schemes: broadcast ($\mathrm{BR}$) and fully directional ($\mathrm{FD}$). Using the former, the UEs transmit simultaneously to both the mmAP and the relay by means of wider beams at lower beamforming gains, whereas, with the $\mathrm{FD}$ scheme, the UEs transmit either to the mmAP or to the relay by using narrow beams. Moreover, we take into account the beam alignments that occur every time the transmitters change receiver and scheme.  

\subsection{Related Work}
\label{sec:Rel}
Several works have been proposed for evaluating the benefits of relaying techniques in mm-wave communications, e.g,~\cite{CoopDiv,RelayPhy1,RelayPhy2,ConnRel,CovRel,ProbD2D,CovD2D,RelayBlo,FairRelay,FallRelay}. In~\cite{CoopDiv}, the authors propose a physical layer analysis of cooperative communications for frequencies above 10 GHz and evaluate the outage probability of several multiple access protocols, combining techniques, and relay transmission techniques. The study shows that the use of relays drastically improves the coverage probability and the correlation between the source-relay and relay-destination links can be exploited to improve the performance. The authors of~\cite{RelayPhy1,RelayPhy2} use stochastic geometry to show the improvements in the signal-to-interference-plus-noise ratio (SINR) distribution and coverage probability for a mm-wave cellular network that is assisted by a relay. The results of~\cite{RelayPhy2} show the asymptotic gain that can be achieved by using the best relay selection strategy over random relay selection.

Stochastic geometry is also used in~\cite{ConnRel,CovRel,ProbD2D,CovD2D}. In~\cite{ConnRel}, the connection probability for mm-wave wireless networks with multi-hop relaying is analyzed. The authors show that the connection probability is strictly correlated to the obstacle density and the width of the region where the relays are potentially selected. In~\cite{CovRel}, the coverage probability for a decode-and-forward relay is analyzed; the authors consider the relay that has the highest signal-to-noise ratio (SNR) to the receiver among the set of relays that can decode the source message. In~\cite{ProbD2D} and~\cite{CovD2D}, the authors focus on relaying techniques for device-to-device (D2D) scenarios and analyze, by using stochastic geometry, the coverage probability and the relay selection problem, respectively. The relay selection strategy is further evaluated in~\cite{RelayBlo,FairRelay}. The former proposes a two-hop relay selection algorithm for mm-wave communications to take into account the dependency between the source-destination and relay-destination paths in terms of line-of-sight (LOS) probability. The work in~\cite{FairRelay} considers a joint relay selection and mmAP association problem. In particular, the authors propose a distributed solution that takes into account the load balancing and fairness aspects among multiple mmAPs. 

None of the aforementioned studies considers the beam alignment phase. This aspect is taken into account in~\cite{FallRelay}, for a single source-destination pair and a single half-duplex relay. When the source-destination link is blocked, the source node can transmit either to the relay by using mm-waves or to both the relay and the destination by using lower frequencies. In the former case a beam alignment occurs. The authors compare the two approaches in terms of throughput and delay, but differently from our approach they assume continuous time and single UE scenario.

In general, analysis of relaying techniques in mm-wave wireless networks regarding network-level performance need further studies. However, it is worth mentioning works that propose similar analysis for lower frequencies, such as~\cite{MonCoop,CoopWC,Niko,Rel2,ITC30,CoopSurv,CoopSel}. In~\cite{MonCoop,CoopWC}, benefits and challenges of cooperative communications for wireless networks are extensively discussed. In~\cite{Niko}, the authors consider a multi-user scenario with a full-duplex relay and a destination that have multi-packet reception capability, whereas, the studies in~\cite{Rel2, ITC30} analyze a similar scenario, but with two relays. In~\cite{Niko, Rel2, ITC30}, the relays are equipped with infinite size queue for which the performance are analyzed as well as the per-user and network throughput, and the delay per packet. Buffer-aided relays are also considered by~\cite{CoopSurv,CoopSel}. The former illustrates and compares several buffer-aided relaying protocols, whereas, the latter analyzes presents a comprehensive study of relay selection techniques for lower frequencies wireless networks.

\subsection{Contributions}
\label{sec:Contributions}
We provide a novel analysis of delay and throughput for random access multi-user cooperative relaying mm-wave wireless networks. We show the tradeoff between using the aforementioned transmission schemes, i.e., $\mathrm{FD}$ and $\mathrm{BR}$, by taking into account the different beamforming gains and interference caused by both types of transmissions. Namely, in contrast to the $\mathrm{FD}$ scheme, $\mathrm{BR}$ transmissions use wider beams that provide a lower beamforming gain, but they can allow to transmit simultaneously both to the relay and the mmAP. Furthermore, switching transmission scheme involves a beam alignment phase between the transmitter and the receiver and, therefore, we show how the duration of this phase impacts the performance. 

In more detail, at first, we compute the analytical expression of the user transmit probability, which, as we show, is decreased by the beam alignment. Then, by using queueing theory, we study the performance characteristics of the queue at the relay. More precisely, we consider what is called network-level cooperation at the relay. This forward the successfully decoded packets that are stored in a queue, whose operations are analyzed in details. This analysis includes stability condition, as well as the service and the arrival rate. Moreover, we model the evolution of the queue as a discrete time Markov Chain in which each state denotes the number of the packets in the queue. Since we derive the transition probabilities then we can provide the probability that the queue is empty and the average queue size.

Finally, we identify the optimal transmission scheme (i.e., $\mathrm{FD}$ and $\mathrm{BR}$) with respect to several system parameters, e.g., number and positions of nodes, and beam alignment duration. In addition, we also analyze and discuss the impact of imperfect beam alignment and imperfect self-interference cancellation on the network performance. Namely, we investigate when it is more beneficial for the UEs to transmit simultaneously to both the mmAP and the relay by using wider beams, and when instead it is better to use narrow beams and transmit either to the mmAP or the relay. To the best of our knowledge, such analysis has not been investigated yet.

The rest of the paper is organized as follows. In Section \ref{sec:Ass} we describe the system model. In Section~\ref{sec:PA}, we present the queue analysis at the relay and, in Section~\ref{sec:Thr}, we evaluate the aggregate network throughput. In Section~\ref{sec:Del} we derive the delay per packet expression and, in Section~\ref{sec:Res}, we provide performance evaluation. Finally, Section~\ref{sec:Conc} concludes the paper.



\section{System Model and Assumptions}
\label{sec:Ass}

\subsection{Network Model}
\label{sec:NM}
We consider a set $\mathcal{N}$, with cardinality $N$, of symmetric\footnote{This study can be generalized to the asymmetric case; however, the analysis will be dramatically involved without providing any additional meaningful insight.} UEs, which are characterized by the same mm-wave networking characteristics such as propagation conditions and topology. We assume multiple packet reception capability both at the mmAP and the relay ($R$), which can form multiple beams at the same time for multiple packets reception~\cite{Hybrid}. Each UE however is considered to be equipped with one analog beamformer and it can form only one beam at a time. We assume slotted time and each packet transmission takes one timeslot. The relay has no packets of its own, but it stores the successfully received packets from the UEs in a queue, which has infinite size\footnote{The
analysis with infinite size is more general, and it can also provide insights on the optimal design of the queue size based on the distribution of the occupancy of the queue. Moreover, the analysis is still valid if the
queue is large enough}. The UEs have saturated queues, i.e., they never empty. We assume that acknowledgements (ACKs) are instantaneous and error free and successfully received packets are removed from the queues of the transmitting nodes, i.e., both the UEs and $R$.  

In a given timeslot, the relay transmits a packet to the mmAP with probability $q_{r}$, whereas, the UEs decide to transmit a packet with probability $q_{u}$. Then, the UEs randomly select one of the two transmission schemes, i.e., $\mathrm{BR}$ or $\mathrm{FD}$, with probability $q_{ub}$ and $q_{uf}$, respectively, with $q_{uf}+q_{ub}=1$. If the UEs use a $\mathrm{BR}$ transmission and the transmission to the destination fails, the relay stores the packets (that are correctly decoded) in its queue and is responsible to transmit it to the destination. This technique is also known as network level cooperation relaying~\cite{Coop2, Rel1,Rel2,Niko}. In contrast, if the $\mathrm{FD}$ scheme is selected, then the UEs choose to transmit either to the relay, with probability $q_{ur}$, or to the mmAP, with probability $q_{um}$, where $q_{ur}+q_{um}=1$. 

We can summarize this process by defining the set of transmission strategies, $\mathcal{S}=\{fm,fr,b\}$, where, $fm$, $fr$ and $b$ represent the cases in which a UE transmits to the mmAP, to $R$, and to both, respectively. Let $P(s=i)$ be the probability of using strategy $i$, with $i \in \mathcal{S}$. These probabilities do not depend on the particular timeslot and are given by: $P(s=fm)=q_{uf}q_{um}$, $P(s=fr)=q_{uf}q_{ur}$ and $P(s=b)=q_{ub}$. If the selected strategy is the same as in the previous transmission attempt, then the UE can directly transmit, otherwise it has to perform a beam alignment. The alignment is done by the UEs every time they decide to transmit and change strategy. We assume that the beam alignment duration is independent from the selected strategy and equals to $D_{a}$ timeslots and, while a UE is performing an alignment, it can not transmit. Thus, the probability that a UE is \textit{actually} transmitting, $q_{tx}$, is affected by the beam alignment, and its derivation is presented in the next section.
\begin{figure}[tb]
	\centering
	\includegraphics[width=6cm]{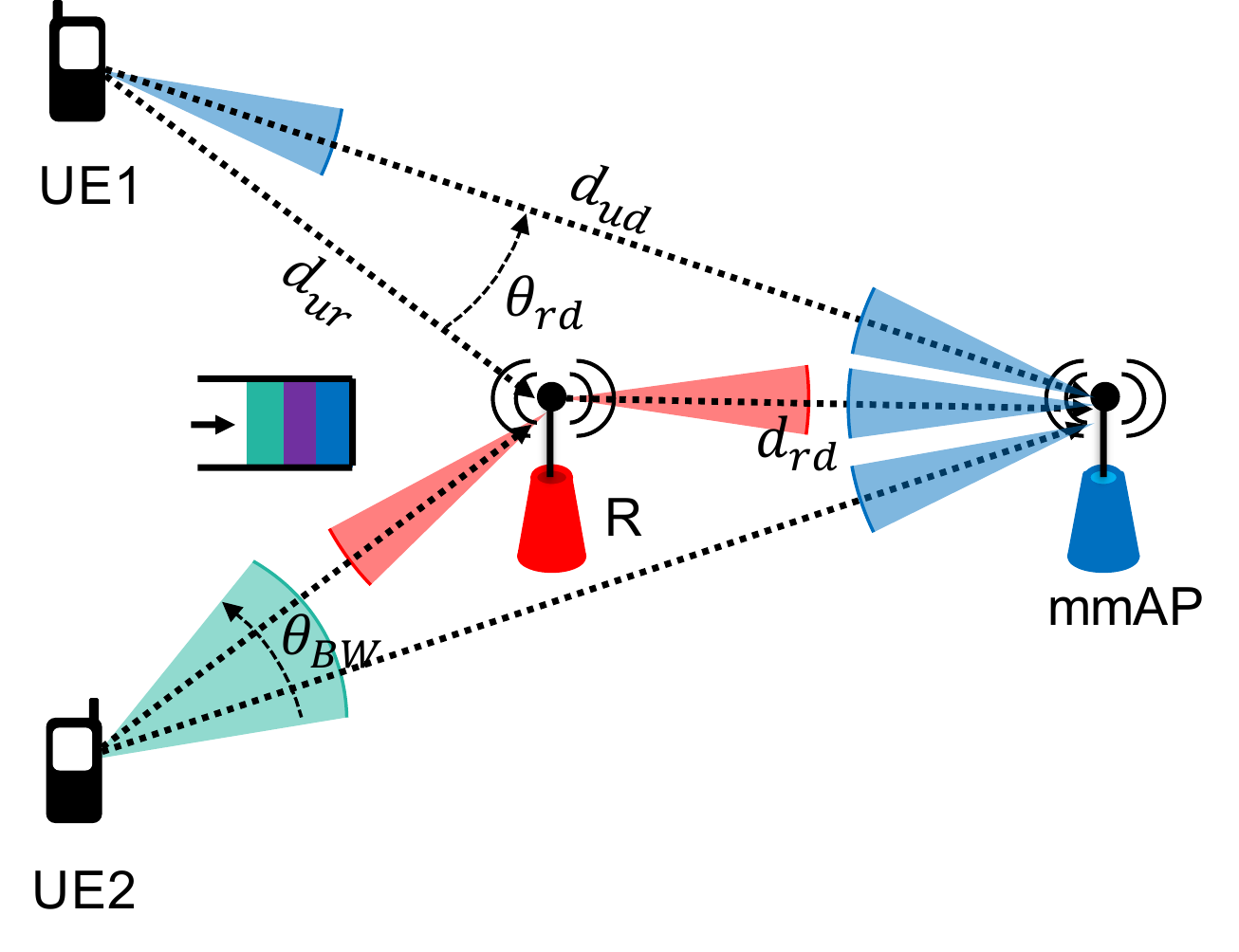}
	\caption[]{$\mathrm{FD}$ (UE1) and $\mathrm{BR}$ (UE2) transmissions for a scenario with two UEs, one relay and one mmAP. In this example, UE1 is transmitting to the mmAP.}
	 \label{fig:FB}
\end{figure}
In Fig.~\ref{fig:FB}, we illustrate an example of the $\mathrm{FD}$ and $\mathrm{BR}$ transmissions, where $d_{ur}$ and $d_{ud}$ represent the distances of the paths UE-$R$ and UE-mmAP, respectively. The parameter $\theta_{rd}$ is the angle formed by $R$ and the mmAP with a UE as vertex and $\theta_{BW}$ is the beamwidth. Hereafter, we indicate the probability of the complementary event by a bar over the term (e.g., $\overline{q}_{u}=1-q_{u}$). Moreover, we use superscripts $f$ and $b$ to indicate the $\mathrm{FD}$ and $\mathrm{BR}$ transmissions, respectively. 

\subsection{SINR Expression and Success Probability}
\label{sec:SINRexp}
A packet is considered to be successfully received if the SINR is above a certain threshold $\gamma$. Ideally, multiple transmissions at the receiver side of a node do not interfere when they are received on different beams. However, in real scenarios, interference cancellation techniques are not perfect. Therefore, we introduce a coefficient $0\le \alpha \le1$ that models the interference between received beams\footnote{Our work can be easily generalized to the case where $\alpha$ depends on the transmission strategy. However, in order to keep the clarity of the presentation we consider $\alpha$ to be constant.}. The cases $\alpha=0$ and $\alpha=1$ represent perfect interference cancellation and no interference cancellation, respectively. Moreover, given the negligible interference between transmissions of different pairs of nodes in mm-waves~\cite{Int_Fischione}, we assume that an $\mathrm{FD}$ transmission to the mmAP does not interfere with the packet transmitted to $R$ and vice-versa. On the other hand, when a UE uses a $\mathrm{BR}$ transmission, its transmission interferes with the transmissions of the other UEs for both the mmAP and $R$. 

We assume that the links between all pairs of nodes are independent and can be in two different states, line-of-sight (LOS) and non-line-of-sight (NLOS). Specifically, $\text{LOS}_{ij}$ and $\text{NLOS}_{ij}$ are the events that node $i$ is in LOS and NLOS with node $j$, respectively. The associated probabilities are denoted as $P(\text{LOS}_{ij})$ and $P(\text{NLOS}_{ij})$. Note that, hereafter, we use subscripts $i$ and $j$ to indicate generic nodes, while, $u$, $r$, and $d$ to indicate the UEs, the relay, and the mmAP, respectively.  

In order to compute the SINR for link $ij$, we first identify the sets of interferers that use $\mathrm{FD}$ and $\mathrm{BR}$ transmissions, which are $\mathcal{I}_{f}$ and $\mathcal{I}_{b}$, respectively. Then, we partition each of them into the sets of nodes that are in LOS and NLOS with node $j$. These sets are $\mathcal{I}_{fl}$ and $\mathcal{I}_{fn}$, for the nodes that use the $\mathrm{FD}$ transmission and $\mathcal{I}_{bl}$ and $\mathcal{I}_{bn}$ for the UEs that use the $\mathrm{BR}$ transmission. 
Thus, when link $ij$ is in LOS, we can derive the SINR, conditioned to $\mathcal{I}_{fl},\mathcal{I}_{fn},\mathcal{I}_{bl},\mathcal{I}_{bn}$, as follows:
\begin{align*} 
&\text{SINR}_{ij/\mathcal{I}_{fl},\mathcal{I}_{fn},\mathcal{I}_{bl},\mathcal{I}_{bn}}^{f} |\text{LOS}_{ij}\\
&=\frac{p_{t}g_{i}^{f}g_{j}^{f}h_{l}(i,j)}{p_{N}+\alpha \Bigl(\displaystyle \sum_{k \in \mathcal{I}_{fl}}p_{r/l}^{f}(k,j)+\sum_{m \in \mathcal{I}_{bl}}p_{r/l}^{b}(m,j)+\sum_{u \in \mathcal{I}_{fn}}p_{r/n}^{f}(u,j)+\sum_{v \in \mathcal{I}_{bn}}p_{r/n}^{b}(v,j)\Bigl)},\stepcounter{equation}\tag{\theequation}\label{eq:SINR}
\end{align*} 

where, $g_{i}$ and $g_{j}$ are the transmitter and receiver beamforming gains, respectively. They are computed in according to the ideal sectored antenna model \cite{Bai}, which is given by: $g_{i}=g_{j}=\frac{2\pi}{\theta_{BW}}$ in the main lobe, and $0$ otherwise. 
The transmit and noise power are $p_{t}$ and $p_{N}$, respectively and $h_{l}(i,j)$ is the path loss on link $ij$ when this is in LOS. The terms $p_{r/l}(i,j)$ and $p_{r/n}(i,j)$ represent the received power by node $j$ from node $i$, when the first is in LOS and NLOS, respectively. Similar expressions of the SINR can be derived also in case of $\mathrm{BR}$ and NLOS. 

Finally, the success probabilities for a packet sent on link $ij$ by using $\mathrm{FD}$ and $\mathrm{BR}$ transmissions are represented by the terms $P_{ij/\mathcal{I}_{f},\mathcal{I}_{b}}^{f}$ and $P_{ij/\mathcal{I}_{f},\mathcal{I}_{b}}^{b}$, respectively. Here, we consider only the conditioning on the sets $\mathcal{I}_{f}$ and $\mathcal{I}_{b}$ because we average on all the possible scenarios for the LOS and NLOS link conditions. The expression for the $\mathrm{FD}$ transmission and $N$ UEs is given in Appendix A, where we assume perfect beam alignment and relay self-interference cancellation. Imperfect beam alignment and Imperfect relay self-interference cancellation cases are discussed in Section~\ref{sec:Beam_Err} and~\ref{sec:FDself}, respectively.

\begin{table}[!t]
	\renewcommand{\arraystretch}{1.0}
	\caption{Summary of the notation.}
	\label{par}
	\centering
	\begin{tabular}{ l | l || l | l }
		\hline
		UE & user equipment &  $N$ & number of UEs\\
		mmAP & mm-wave access point (destination) &  $R$ & relay\\
		$D_{a}$ & beam alignment duration &$q_{r}$ & relay transmit probability  \\
		$q_{u}$ & UE transmit probability & $q_{tx}$ & \textit{actual} UE transmit probability\\
		$q_{ub}$ & probability to use a $\mathrm{BR}$ transmission & $q_{uf}$ & probability to use an $\mathrm{FD}$ transmission\\
		$q_{um}$ & probability to transmit to the mmAP when & $q_{ur}$ & probability to transmit to $R$ when\\
		 &using $\mathrm{FD}$ transmissions& 	 &using $\mathrm{FD}$ transmissions\\
		$\mathcal{S}$ & set of transmission strategies & $b$ & broadcast transmissions\\
		$fm$ & fully directional transmission to the mmAP & $fr$ & fully directional transmission to the relay\\ 
		$d_{dr}$ & mmAP-relay distance& $d_{ud}$ & UE-mmAP distance \\
		$d_{ur}$ & UE-relay distance & $\alpha$ & interference cancellation parameter\\ 
		$\gamma$ &SINR threshold for successful transmissions & $g_{i}^{s}$ & beamforming gain at node $i$ while using strategy $s$  \\
		$\mathcal{I}_{b}$ & set of interferers that use $\mathrm{BR}$ transmissions & $\mathcal{I}_{f}$ &  set of interferers that use $\mathrm{FD}$ transmissions\\
		$\lambda_{r}$ & arrival rate at the relay & $\mu_{r}$ & service rate at the relay\\ 
		$P_{ij/\mathcal{I}_{f},\mathcal{I}_{b}}^{b}$ & success probability of a transmission from the i-th&	$P_{ij/\mathcal{I}_{f},\mathcal{I}_{b}}^{b}$ & success probability of a transmission from the i-th \\ 
		 &to the j-th nodes by using a $\mathrm{BR}$ transmission& 	 &to the j-th nodes by using an $\mathrm{FD}$ transmission\\
		$\theta_{rd}$ & angle between the mmAP& $\theta_{BW}^{b}$ & beamwidth for $\mathrm{BR}$ transmissions \\
		& and $R$ with the UE as vertex & $\theta_{BW}^{f}$ & beamwidth for $\mathrm{FD}$ transmissions \\
		\hline
	\end{tabular}
\end{table}

\section{Performance Analysis}
\label{sec:PA}

\subsection{UE Transmit Probability}
\label{sec:TP}
In this section, we first derive the \textit{actual} transmit probability of a UE in a given timeslot, i.e., $q_{tx}$, when beam alignment is taken into account. Then, we evaluate the performance of the queue at the relay and we analyze the network throughput and the delay per packet. 

\begin{theorem}\label{theo}
\textit{For each timeslot $k$, the probability distributions of the transmission strategy selection $P(s_{k}=i)$ are identical distributed (i.d.) with $i \in \mathcal{S}$, then, the transmit probability for a UE in a timeslot $k$, with constant alignment duration $D_{a}$, is given by:}
\begin{align*} 
q_{tx}&=P\bigl(I_{k}\bigr)q_{u}=\frac{q_{u}}{1+D_{a}\bigl(1-P(s_{k}=i \cap s_{\hat{k}}=i)\bigr)}\stepcounter{equation}\tag{\theequation}\label{eq:qtxf},
\end{align*} 
\textit{where, $q_{u}$ is defined in Section~\ref{sec:NM} and $P\bigl(I_{k}\bigr)$ is the probability that the UE has not started an alignment in the previous $D_{a}$ timeslots. The term $P(s_{k}=i \cap s_{\hat{k}}=i)$ is the probability to use the i-th strategy in timeslot $k$ while using the same strategy for the previous transmission attempt, which occurs in the $\hat{k}$-th timeslot.}
\end{theorem}

\begin{proof}
The proof is given in Appendix~B.
\end{proof}

From \eqref{eq:qtxf}, one can notice that $q_{tx}$ is inversely proportional to the beam alignment duration $D_{a}$ as well as to the probability of changing strategy $1-P(s_{k}=i \cap s_{\hat{k}}=i) $. Assuming that the probabilities of the transmission strategy selection, $P(s_{k}=i)$, are independent in each timeslot $k$ and have values as reported in Section~\ref{sec:NM}, $P\bigl(I_{k}\bigr)$ can be written as: 
\begin{align*} 
P\bigl(I_{k}\bigr)&=\frac{1}{1+D_{a}q_{u}\bigl(1-(q_{uf}q_{um})^{2}-(q_{uf}q_{ur})^{2}-(q_{ub})^{2}\bigr)}\stepcounter{equation}\tag{\theequation}\label{eq:PI}.
\end{align*} 

\subsection{Queue Analysis}
\label{sec:QA}
In this section, we evaluate the arrival rate, $\lambda_{r}$, the service rate, $\mu_{r}$, and the stability condition for the queue at the relay $R$. First, we compute $\lambda_{r}$ that can be expressed as follows:
\begin{align*} 
\lambda_{r}&=P(Q=0)\lambda_{r}^{0}+P(Q\neq0)\lambda_{r}^{1}=P(Q=0)\sum_{k=1}^{N}kr_{k}^{0}+P(Q\neq0)\sum_{k=1}^{N}kr_{k}^{1},\stepcounter{equation}\tag{\theequation}\label{eq:lambdaN}
\end{align*} 
where, $\lambda_{r}^{0}$ and $r_{k}^{0}$ represent the arrival rate at $R$ and the probability that it receives $k$ packets in a timeslot when the queue is empty. Whereas, when the queue is not empty, these two terms assume different values, i.e., $\lambda_{r}^{1}$ and $r_{k}^{1}$. The probabilities that the queue is either empty or not empty, $P(Q=0)$ and $P(Q\neq0)$, respectively, are derived in appendix D. When the queue is not empty, $R$ may transmit and interfere with the other transmissions to the mmAP. This interference affects the probability to successfully transmit a packet to the mmAP and therefore the number of received packets by $R$. Thus, In order to compute $\lambda_{r}^{0}$ and $\lambda_{r}^{1}$, we first compute the success transmission probability by identifying the nodes that belong to the sets of interferers $\mathcal{I}_{f}$ and $\mathcal{I}_{b}$. Since the UEs are symmetric, it is sufficient to indicate the number of UEs that are interfering and whether $R$ is transmitting; i.e., we indicate with $\{|\mathcal{I}_{f}|,r\}^{f}$ and $\{|\mathcal{I}_{f}|\}^{f}$ the sets of interferers that use $\mathrm{FD}$ transmissions when $R$ is transmitting or not, and with $\{r\}^{f}$ the set of interferers when only the relay $R$ is transmitting. For the sake of clarity, we first present hereafter the results for two UEs and then, in Appendix D, we generalize the analysis to $N$ UEs. When $N=2$, we can have at maximum two interferers, i.e., the relay and one UE ($|\mathcal{I}_{f}|\le 1$). Moreover, $R$ can receive at maximum two packets per timeslot, i.e., when both the UEs successfully transmit to $R$. Thus, by considering all the possible transmission strategies, $s \in \mathcal{S}$, and all the possible combinations of successfully received packets, we can compute $\lambda_{r}^{0}$ and $\lambda_{r}^{1}$:
\begin{align*}
\lambda_{r}^{0}&=2q_{tx}\overline{q}_{tx}q_{uf}q_{ur}P_{ur}^{f}+2q_{tx}\overline{q}_{tx}q_{ub}P_{ur}^{b}\overline{P}_{ud}^{b}\\
&+q_{tx}^2q_{uf}^2q_{ur}^2q_{ur}^2\Bigl[2P_{ur/\{1\}^{f}}^{f}\overline{P}_{ur/\{1\}^{f}}^{f}+2\Bigl(P_{ur/\{1\}^{f}}^{f}\Bigl)^{2}\Bigl]+2q_{tx}^{2}q_{uf}^{2}q_{ur}q_{um}P_{ur}^{f}\\
&+2q_{tx}^{2}q_{1f}q_{ub}q_{ur}\Bigl[P_{ur/\{1\}^{b}}^{f}\Bigl(1-P_{ur/\{1\}^{f}}^{b}\overline{P}_{ud}^{b}\Bigl)+\overline{P}_{ur/\{1\}^{b}}^{f}P_{ur/\{1\}^{f}}^{b}\overline{P}_{ud}^{b}+2\Bigl(P_{ur/\{1\}^{f}}^{b}\overline{P}_{ud}^{b}\Bigl)^{2}\Bigl]\\
&+2q_{tx}^{2}q_{ub}q_{uf}q_{um}P_{ur}^{b}\overline{P}_{ud/\{1\}^{f}}^{b}+q_{tx}^{2}q_{ub}^{2}\Bigl[2P_{ur/\{1\}^{b}}^{b}\overline{P}_{ud/\{2\}^{b}}^{b}\Bigl(1-P_{ur/\{1\}^{b}}^{b}\overline{P}_{ud/\{1\}^{b}}^{b}\Bigl)\\
&+2\Bigl(P_{ur/\{1\}^{b}}^{b}\overline{P}_{ud/\{1\}^{b}}^{b}\Bigl)^{2}\Bigl],\stepcounter{equation}\tag{\theequation}\label{eq:lambda02}
\end{align*}
where, ${q}_{tx}$, $q_{ub}$, $q_{uf}$, $q_{ud}$, and $q_{ur}$ are introduced in Section~\ref{sec:NM} and a summary of the notation is available in Table~\ref{par}. In order to compute $\lambda_{r}^{1}$, we must consider the possible interference of $R$. Thus, $\lambda_{r}^{1}=\overline{q}_{r}\lambda_{r}^{0}+q_{r}A_{r}$ and $A_{r}$ is given by:
\begin{align*}
A_{r}&=2q_{tx}\overline{q}_{tx}q_{uf}q_{ur}P_{ur}^{f}+2q_{tx}\overline{q}_{tx}q_{ub}P_{ur}^{b}\overline{P}_{ud}^{b}+q_{tx}^2q_{uf}^2q_{ur}^2q_{ur}^2\Bigl[2P_{ur/\{1\}^{f}}^{f}\overline{P}_{ur/\{1\}^{f}}^{f}\\
&+2\Bigl(P_{ur/\{1\}^{f}}^{f}\Bigl)^{2}\Bigl]+2q_{tx}^{2}q_{uf}^{2}q_{ur}q_{um}P_{ur}^{f}+2q_{tx}^{2}q_{uf}q_{ub}q_{ur}\Bigl[P_{ur/\{1\}^{b}}^{f}\Bigl(1-P_{ur/\{1\}^{f}}^{b}\overline{P}_{ud/\{r\}^{f}}^{b}\Bigl)\\
&+\overline{P}_{ur/\{1\}^{b}}^{f}P_{ur/\{1\}^{f}}^{b}\overline{P}_{ud/\{r\}^{f}}^{b}+2\Bigl(P_{ur/\{1\}^{f}}^{b}\overline{P}_{ud/\{r\}^{f}}^{b}\Bigl)^{2}\Bigl]+2q_{tx}^{2}q_{ub}q_{uf}q_{um}P_{ur}^{b}\overline{P}_{ud/\{1,r\}^{f}}^{b}+q_{tx}^{2}q_{ub}^{2}\\
&\times\Bigl[2P_{ur/\{1\}^{b}}^{b}\overline{P}_{ud/\{r\}^{f},\{1\}^{b}}^{b}\Bigl(1-P_{ur/\{1\}^{b}}^{b}\overline{P}_{ud/\{r\}^{f},\{1\}^{b}}^{b}\Bigl)+2\Bigl(P_{ur/\{1\}^{b}}^{b}\overline{P}_{ud/\{r\}^{f},\{1\}^{b}}^{b}\Bigl)^{2}\Bigl].\stepcounter{equation}\tag{\theequation}\label{eq:AR}
\end{align*}

As introduced in Section \ref{sec:NM}, $R$ can transmit a packet to the mmAP by using the $\mathrm{FD}$ scheme. This transmission may be subject to the interference of the UEs that are transmitting to the mmAP. Therefore, we compute the service rate as $\mu_{r}=q_{r}B_{r}$, where $B_{r}$ is given by:
\begin{align*}
B_{r}&=P_{rd}^{f}\Bigl(\overline{q}_{tx}^{2}+2q_{tx}\overline{q}_{tx}q_{uf}q_{ur}+q_{tx}^{2}q_{uf}^{2}q_{2f}q_{ur}^{2}\Bigl)+P_{rd/\{1\}^{f}}^{f}\Bigl(2q_{tx}\overline{q}_{tx}q_{uf}q_{um}+2q_{tx}^{2}q_{uf}^{2}q_{um}q_{ur}\Bigl)\\
&+P_{rd/\{1\}^{b}}^{f}\Bigl(2q_{tx}\overline{q}_{tx}q_{ub}+2q_{tx}^{2}q_{ub}q_{uf}q_{ur}\Bigl)+P_{rd/\{2\}^{f}}^{f}q_{tx}^{2}q_{uf}^{2}q_{um}^{2}+P_{rd/\{1\}^{f},\{1\}^{b}}^{f}2q_{tx}q_{uf}q_{ub}q_{um}\\
&+P_{rd/\{2\}^{b}}^{f}q_{tx}^{2}q_{ub}^{2}.\stepcounter{equation}\tag{\theequation}\label{eq:muR2}
\end{align*}

By applying the Loyne's criterion~\cite{loynes}, we can now obtain the range of values of $q_{r}$ for which the queue is stable by solving the following inequality: $\lambda_{r}^{1} <\mu_{r}$. Thus, we have that the queue at $R$ is stable if and only if $q_{rmin}<q_{r}\le1$, where $q_{rmin}$ is given by:
\begin{equation}\label{eq:qmin}
\begin{aligned} 
q_{rmin}=\frac{\lambda_{r}^{0}}{\lambda_{r}^{0}+B_{r}-A_{r}}.
\end{aligned} 
\end{equation}

The evolution of the queue at the relay can be modelled as a discrete time Markov Chain (DTMC), see Fig.~\ref{fig:DMTC}. The terms $p_{k}^{0}$ and $p_{k}^{1}$, derived in Appendix~C, are the probabilities that the queue size increases by $k$ packets in a timeslot when the queue is empty or not. The expressions for the probability that the queue is empty, $P(Q=0)$, and the average relay queue size, $\overline{Q}$, are derived in Appendix~D that contains the queue performance analysis for $N$ symmetric UEs.
\begin{figure}[tb]
	\centering
	\includegraphics[width=8cm]{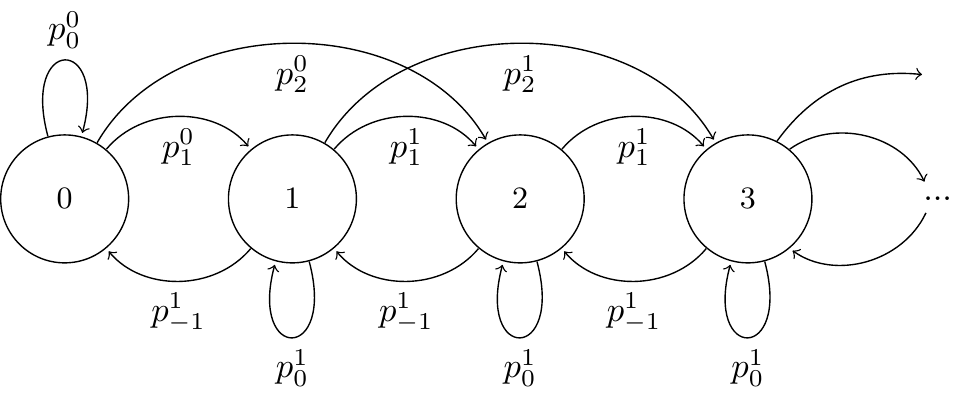}
	\caption[]{The DTMC model for the two UE case.}
	 \label{fig:DMTC}
\end{figure}

\section{Throughput and Delay Analysis}
\label{sec:ThrDel}
\subsection{Throughput Analysis}
\label{sec:Thr}
In this section, we derive the network aggregate throughput, $T$, for $N$ UEs. More specifically, $T$ represents the end-to-end throughput from the UEs to the mmAP and can be computed by considering the following cases: i) when the queue at $R$ is stable and ii) otherwise. In the former case, $T$ can be expressed as follows:
\begin{equation}\label{eq:TT2}
\begin{aligned} 
T=Nq_{tx}T_{u}&=Nq_{tx}q_{uf}\Bigl(q_{um}T_{ud}^{f}+q_{ur}T_{ur}^{f}\Bigr)+Nq_{tx}q_{ub}\Bigl(T_{ud}^{b}+T_{ur}^{b}\Bigr),
\end{aligned} 
\end{equation}
where, $T_{u}$ is the per-user throughput conditioned to the event that the UE is transmitting. On the other hand, when the queue at $R$ is unstable, the aggregate throughput becomes:
\begin{equation}\label{eq:TT}
\begin{aligned} 
T=Nq_{tx}\Bigl(q_{uf}q_{um}T_{ud}^{f}+q_{ub}T_{ud}^{b}\Bigr)+\mu_{r}.
\end{aligned} 
\end{equation}
The terms, $T_{ud}^{f}$, $T_{ud}^{b}$, $T_{ur}^{f}$, and $T_{ur}^{b}$ represent the contributions to $T_{u}$ given by the packets received directly by the mmAP or by $R$, when the $\mathrm{FD}$ and the $\mathrm{BR}$ transmissions are used, respectively and can be expressed as follows: 
\begin{equation}\label{eq:TUidNf}
\begin{aligned} 
T_{ud}^{f}&=\Bigl(1-q_{r}P(Q\neq0)\Bigr)T_{ud}^{f0}+q_{r}P(Q\neq0)T_{ud}^{f1},
\end{aligned}
\end{equation}
\begin{equation}\label{eq:TUidNb}
\begin{aligned} 
T_{ud}^{b}&=\Bigl(1-q_{r}P(Q\neq0)\Bigr)T_{ud}^{b0}+q_{r}P(Q\neq0)T_{ud}^{b1},
\end{aligned}
\end{equation}
\begin{align*} 
T_{ur}^{b}&=\Bigl(1-q_{r}P(Q\neq0)\Bigr)T_{ur}^{b0}+q_{r}P(Q\neq0)T_{ur}^{b1},\stepcounter{equation}\tag{\theequation}\label{eq:TUirNb}
\end{align*} 
where, we show the contributions to $T_{u}$ when $R$ is interfering or not. These two cases are indicated with the superscripts $0$ and $1$, respectively (e.g., $T_{ud}^{f0}$ and $T_{ud}^{f1}$). The term $P(Q=0)$ is derived in Appendix D. Note that the expression of $T_{ur}^{f}$ is not affected by the interference of $R$ because, in the following analysis, we assume perfect self-interference cancellation at $R$. This assumption is relaxed in Section~\ref{sec:FDself}, where we present results for imperfect self-interference cancellation. In order to compute the throughput components (e.g., $T_{ud}^{f0}$ and $T_{ud}^{f1}$), we follow the same reasoning that is done for $\mu_{r}$ in~\eqref{eq:muR2}. More specifically, we average the number of successfully transmitted UEs packets on all the possible interference scenarios.

Hereafter, we indicate by $m$ the number of UEs that interfere and with $i$ the number of those that use $\mathrm{FD}$ transmissions ($m-i$ UEs use the $\mathrm{BR}$ transmission). Moreover, among the interfering UEs that use $\mathrm{FD}$ transmissions, a certain number $j$ transmit to $R$ and $i-j$ to the mmAP. Thus, we obtain the following:
\begin{align*} 
T_{ud}^{f0}&=\sum_{m=0}^{N-1}\binom{N-1}{m}q_{tx}^{m}\overline{q}_{tx}^{\ N-1-m} \sum_{i=0}^{m}\binom{m}{i}q_{uf}^{i}q_{ub}^{m-i}\sum_{j=0}^{i} \binom{i}{j}q_{ur}^{j}q_{um}^{i-j}P_{ud/\{i-j\}^{f},\{m-i\}^{b}}^{f},\stepcounter{equation}\tag{\theequation}\label{eq:Tidf0}
\end{align*}
\begin{align*} 
T_{ud}^{b0}&=\sum_{m=0}^{N-1}\binom{N-1}{m}q_{tx}^{m}\overline{q}_{tx}^{\ N-1-m} \sum_{i=0}^{m}\binom{m}{i}q_{uf}^{i}q_{ub}^{m-i}\sum_{j=0}^{i} \binom{i}{j}q_{ur}^{j}q_{um}^{i-j}\times P_{ud/\{i-j\}^{f},\{m-i\}^{b}}^{b}\stepcounter{equation}\tag{\theequation}\label{eq:Tidb0},
\end{align*}  
\begin{align*} 
T_{ud}^{f1}&=\sum_{m=0}^{N-1}\binom{N-1}{m}q_{tx}^{m}\overline{q}_{tx}^{\ N-1-m}\sum_{i=0}^{m}\binom{m}{i}q_{uf}^{i}q_{ub}^{m-i}\sum_{j=0}^{i} \binom{i}{j}q_{ur}^{j}q_{um}^{i-j}P_{ud/\{i-j,r\}^{f},\{m-i\}^{b}}^{f}\stepcounter{equation}\tag{\theequation}\label{eq:Tidf1},
\end{align*}
\begin{align*} 
T_{ud}^{b1}&=\sum_{m=0}^{N-1}\binom{N-1}{m}q_{tx}^{m}\overline{q}_{tx}^{\ N-1-m}\sum_{i=0}^{m}\binom{m}{i}q_{uf}^{i}q_{ub}^{m-i}\sum_{j=0}^{i} \binom{i}{j}q_{ur}^{j}q_{um}^{i-j}P_{ud/\{i-j,r\}^{f},\{m-i\}^{b}}^{b}.\stepcounter{equation}\tag{\theequation}\label{eq:Tidb1}
\end{align*}
Finally, we derive the terms $T_{ur}^{f}$, $T_{ur}^{b0}$ and $T_{ur}^{b1}$ as follows:
\begin{align*} 
T_{ur}^{f}&=\sum_{m=0}^{N-1}\binom{N-1}{m}q_{tx}^{m}\overline{q}_{tx}^{\ N-1-m} \sum_{i=0}^{m}\binom{m}{i}q_{uf}^{i}q_{ub}^{m-i} \sum_{j=0}^{i} \binom{i}{j}q_{ur}^{j}q_{um}^{i-j}P_{ur/\{j\}^{f},\{m-i\}^{b}}^{f},\stepcounter{equation}\tag{\theequation}\label{eq:Tirf}
\end{align*} 

\begin{align*}
T_{ur}^{b0}&=\sum_{m=0}^{N-1}\binom{N-1}{m}q_{tx}^{m}\overline{q}_{tx}^{\ N-1-m}\sum_{i=0}^{m}\binom{m}{i}q_{uf}^{i}q_{ub}^{m-i} \sum_{j=0}^{i} \binom{i}{j}q_{ur}^{j}q_{um}^{i-j}P_{ur/\{j\}^{f},\{m-i\}^{b}}^{b}\overline{P}_{ud/\{i-j\}^{f},\{m-i\}^{b}}^{b},\stepcounter{equation}\tag{\theequation}\label{eq:Tirb0}
\end{align*} 

\begin{align*}
T_{ur}^{b1}&=\sum_{m=0}^{N-1}\binom{N-1}{m}q_{tx}^{m}\overline{q}_{tx}^{\ N-1-m} \sum_{i=0}^{m}\binom{m}{i}q_{uf}^{i}q_{ub}^{m-i} \sum_{i=0}^{i} \binom{i}{j}q_{ur}^{j}q_{um}^{i-j}P_{ur/\{j\}^{f},\{m-i\}^{b}}^{b}\overline{P}_{ud/\{i-j,r\}^{f},\{m-i\}^{b}}^{b}.\stepcounter{equation}\tag{\theequation}\label{eq:Tirb1}
\end{align*} 

\subsection{Delay Analysis}
\label{sec:Del}
We now compute the average delay for a packet that is in the head of the queue of a UE. The delay is constituted of three components: i) the transmission delay (i.e., on the links UE-mmAP, UE-R, and R-mmAP), ii) the queueing delay at the relay $D_{q}$, and iii) the beam alignment phase duration $D_{a}$. After a successful transmission, a new packet arrives at the head of the queue. At this point, as explained in Section~\ref{sec:NM}, the UE decides to transmit the packet with probability~$q_{u}$. 

Depending on the selected transmission strategies in the current timeslot and in the previous transmission attempt, the packet can be subject to different delays, $D_{i}$ with $i \in \mathcal{S}$, where, $\mathcal{S}=\{fm,fr,b\}$ and $P(s=fm)$, $P(s=fr)$, and $P(s=b)$ are defined in Section~\ref{sec:NM}. Given that the probability distributions of the transmission strategy selection $P(s=i)$ are independent and identical distributed (i.i.d.) for each timeslot, we can write the probability to use the $i$-th strategy in timeslot $k$ -- conditioned to using the $j$-th strategy in timeslot $h$ -- as follows: $P(s_{k}=i \cap s_{h}=j)=P(s_{k}=i)P(s_{h}=j)=P(s=i)P(s=j)$. Thus, we can express the average delay per packet as follows:
\begin{align*}
D=\sum_{i\in\mathcal{S}}P(s=i)\Bigl(D_{i}+(1-P(s=i))D_{a}\Bigr)\stepcounter{equation}\tag{\theequation}\label{eq:D1},
\end{align*} 

Then, we compute the terms $D_{i}$ of \eqref{eq:D1}, which are given by:
\begin{flalign*}
D_{fm}&=q_{u}T_{ud}^{f}+q_{u}\Bigl(1-T_{ud}^{f}\Bigr)\Bigl(1+q_{uf}q_{um}D_{fm}+q_{uf}q_{ur}(D_{a}+D_{fr})+q_{ub}(D_{a}+D_{b})\Bigr)\\
&+\overline{q}_{u}\Bigl(1+D_{fm}\Bigr), &\stepcounter{equation}\tag{\theequation}\label{eq:Dfd2}
\end{flalign*} 
\begin{flalign*}
D_{fr}&=q_{u}T_{ur}^{f}(1+D_{r})+q_{u}\Bigl(1-T_{ur}^{f}\Bigr)\Bigl(1+q_{uf}q_{ur}D_{fr}+q_{uf}q_{um}(D_{a}+D_{fm})+q_{ub}(D_{a}+D_{b})\Bigr)& \\
&+\overline{q}_{u}\Bigl(1+D_{fr}\Bigr), &\stepcounter{equation}\tag{\theequation}\label{eq:Dfr2}
\end{flalign*} 
\begin{flalign*}
D_{b}&=q_{u}T_{ud}^{b}+q_{u}T_{ur}^{b}(1+D_{r})+q_{u}\Bigl(1-T_{ud}^{b}-T_{ur}^{b}\Bigr) &\\
&\times\Bigl(1+q_{ub}D_{b}+q_{uf}q_{um}(D_{a}+D_{fm})+q_{uf}q_{ur}(D_{a}+D_{fr})\Bigr)+\overline{q}_{u}\Bigl(1+D_{b}\Bigr), & \stepcounter{equation}\tag{\theequation}\label{eq:Db2}
\end{flalign*} 
where, $T_{ud}^{f}$, $T_{ud}^{b}$, $T_{ur}^{f}$ and $T_{ur}^{b}$ are several contributions to the conditioned per-user throughput $T_{u}$ that are given in \ref{sec:Thr}. Since a UE transmits at most one packet per timeslot, $T_{u}$ can be also interpreted as the probability that a packet is successfully transmitted by a UE. The term $D_{r}$ is the total delay at the relay that is defined as the time when the packet entering the relay queue reaches the mmAP and it is given by:
\begin{align*}
D_{r}=D_{q}+\frac{1}{\mu_{r}} =\frac{\overline{Q}}{\lambda_{r}} + \frac{1}{\mu_{r}} \stepcounter{equation}\tag{\theequation}\label{eq:D32}.
\end{align*}
where, $D_{q}$ in~\eqref{eq:D32} is the queueing delay at the relay. The latter is the time when the packet being received by the relay reaches the head of its queue and it is computed by using the Little's law. More precisely, $\overline{Q}$ represents the average relay queue size and $\lambda_{r}$ the average arrival rate, which are given in Appendix D. Finally, by considering~\eqref{eq:D32} and replacing~\eqref{eq:Dfd2}, ~\eqref{eq:Dfr2}, and~\eqref{eq:Db2} in~\eqref{eq:D1}, the average delay per packet $D$ can be written as follows:
\begin{align*}
D&=\frac{1+q_{u}D_{r}\Bigl(q_{uf}q_{ur}T_{ur}^{f}+q_{ub}T_{ur}^{b}\Bigr)+D_{a}q_{u}C}{q_{u}T_{u}},\stepcounter{equation}\tag{\theequation}\label{eq:Delay2}
\end{align*}
\begin{align*}
C&=1+q_{uf}^{2}q_{um}^{2}\Bigl(T_{ud}^{f}-T_{u}-1\Bigr)+q_{uf}^{2}q_{ur}^{2}\Bigl(T_{ur}^{f}-T_{u}-1\Bigr)+q_{ub}^{2}\Bigl(T_{ud}^{b}+T_{ur}^{b}-T_{u}-1\Bigr).  \stepcounter{equation}\tag{\theequation}\label{eq:C}
\end{align*}

\section{Numerical \& Simulation Results}
\label{sec:Res}
In this section, we provide a numerical evaluation of the performance analysis derived for throughput and delay. Furthermore, we assess the validity of the analysis by comparing the numerical results of the analytical model with simulations. To compute the path loss and the LOS and NLOS probabilities, we use the 3GPP model for urban micro cells in outdoor street canyon environment~\cite{3GPP}, whose path loss term includes a lognormal shadowing whose variance depends on whether the link is in LOS or NLOS. Moreover, the path loss depends on the height of the mmAP, $10$\,m, the height of the UE, $1.5$\,m, the carrier frequency, $f_{c} = 30$\,GHz, and the distance between the transmitter and the receiver. The transmit power and the noise power are set to $P_{t} = 24$\,dBm and $P_{N} = -80$\,dBm, respectively. Furthermore, we consider a scenario where the relay $R$ is chosen to be a node that is placed in a position that guarantees the LOS with the mmAP, therefore we assume that $P(\text{LOS}_{rd})=1$. Then, the SINR in~\eqref{eq:SINR} and the success probability in~\eqref{eq:prob_succA} are numerically computed by considering $100,000$ instances of the lognormal shadowing. This success probability represents the input for both the numerical evaluations of the analytical model and simulations results, which are computed over $100,000$ timeslots. 

Moreover, unless otherwise specified, we set $d_{ur} = 30$\,m, $d_{ud} = 50$\,m, $\gamma = 10$\,dB, $\alpha = 0.1$\, and, in case of  $\mathrm{FD}$ transmissions, $\theta_{BW}=\theta_{BW}^{f}=5^\circ$. Instead, when a $\mathrm{BR}$ transmission is used, we set $\theta_{BW}=\theta_{BW}^{b}=\theta_{rd}$, which is the angle between the mmAP and $R$ with the UE as vertex. Throughout this section, we use solid lines for numerical evaluations of the analytical model and dotted lines for the simulation results.

In Fig.~\ref{fig:T_N0} and Fig.~\ref{fig:T_N5}, we show the throughput, $T$, while varying the number of UEs ($N$) for several UE transmit probability values, i.e., $q_{u}$, when $D_{a}=0$ and $D_{a}=5$, respectively. For both the cases, we can observe that the analytical model and the simulations almost coincide. Furthermore, in Fig.~\ref{fig:T_N0}, we can observe that for $q_{u} = 0.1$ the throughput is an increasing function of $N$. In contrast, for $q_{u}= 0.5$ and $q_{u}= 0.9$ the curves have non-monotonic behaviors. Indeed, after that the throughput reaches the maximum (at $N= 6$ and $N= 3$ for $q_{u}= 0.5$ and $q_{u}= 0.9$, respectively), increasing $N$ causes a decrease in $T$. Namely, high values of $N$ and $q_{u}$ lead to high interference that decreases the number of packets successfully received by $R$ and the mmAP. In Fig.~\ref{fig:T_N5}, the larger value of the beam alignment delay, $D_{a}=5$, decreases the transmit probability that causes a decrease of the interference. This explains the monotonic or quasi-monotonic behaviors of the throughput in Fig.~\ref{fig:T_N5}, in which, however, the maximum value of $T$ is lower with respect to Fig.~\ref{fig:T_N0}.
\begin{figure}[!tbp]
  \begin{subfigure}[b]{0.45\textwidth}
    \includegraphics[width=8cm]{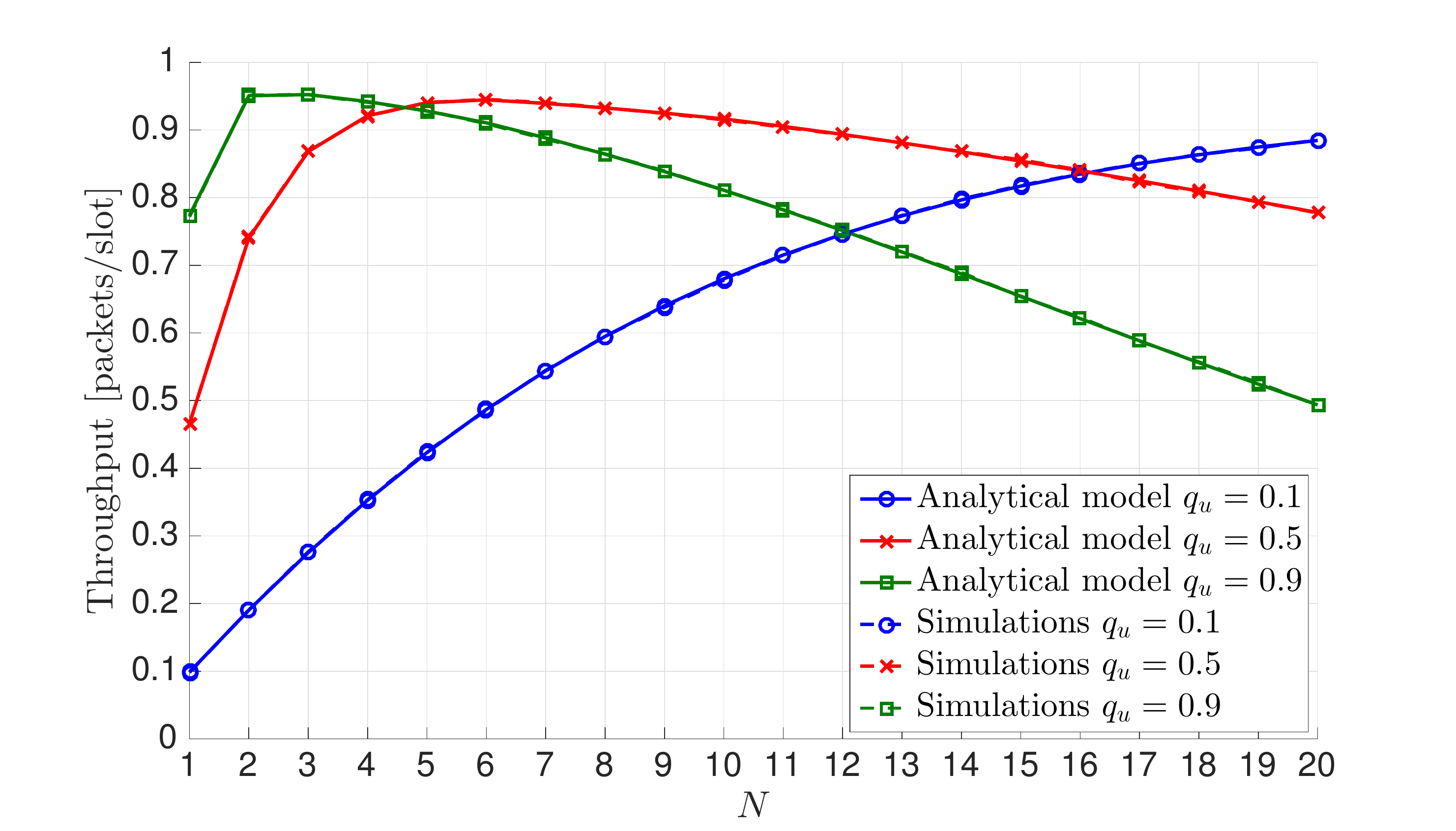}
    \caption{$D_{a}=0$}
    \label{fig:T_N0}
  \end{subfigure}
  \hfill
  \begin{subfigure}[b]{0.45\textwidth}
    \includegraphics[width=8cm]{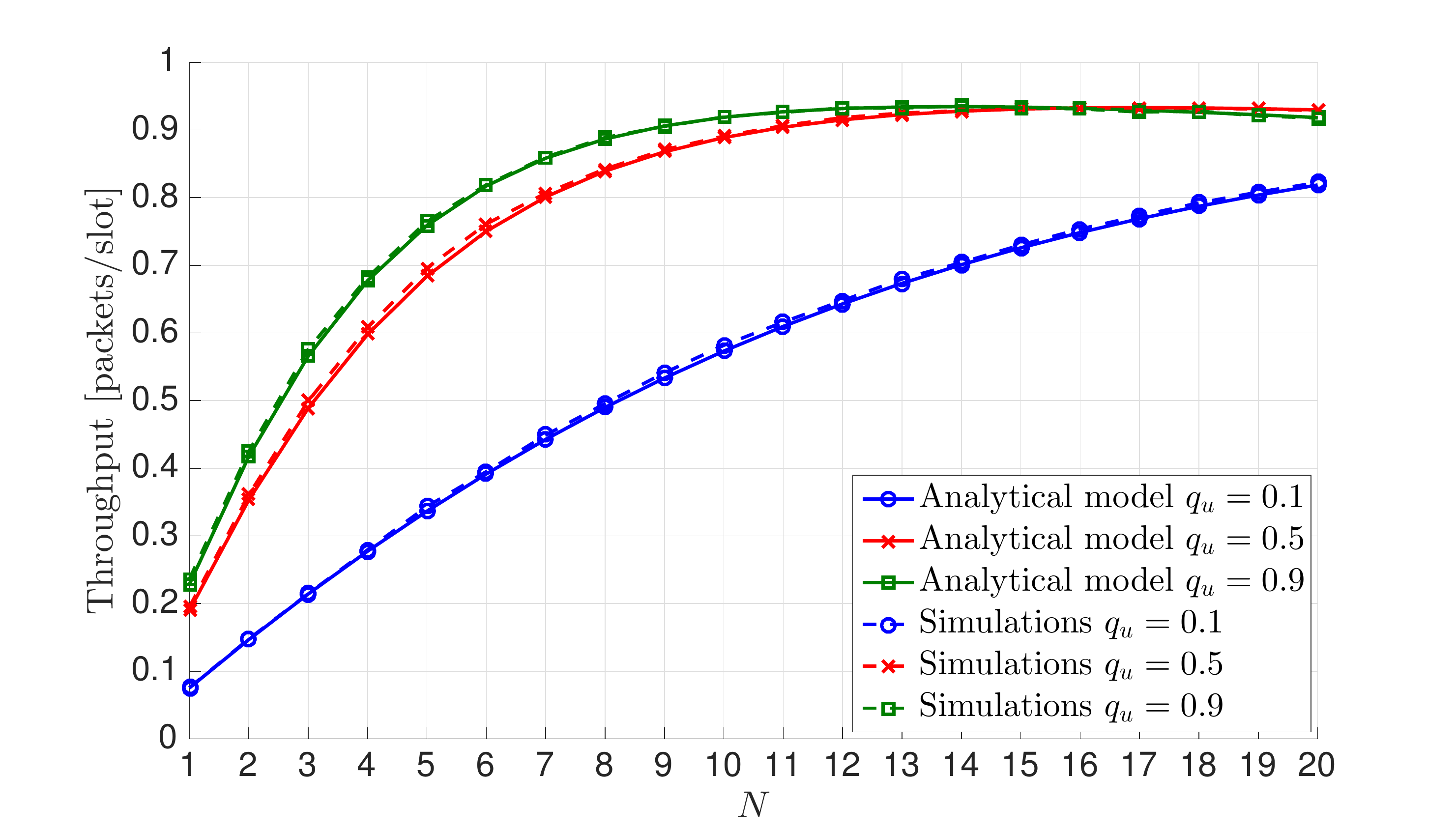}
    \caption{$D_{a}=5$}
    \label{fig:T_N5}
  \end{subfigure}
  \caption{Throughput, $T$, while varying the number of UEs for several UE transmit probability values, i.e., $q_{u}$, when a) $D_{a}=0$ and b) $D_{a}=5$, respectively.}
\end{figure}
\begin{figure}[!tbp]
  \begin{subfigure}[b]{0.45\textwidth}
    \includegraphics[width=8cm]{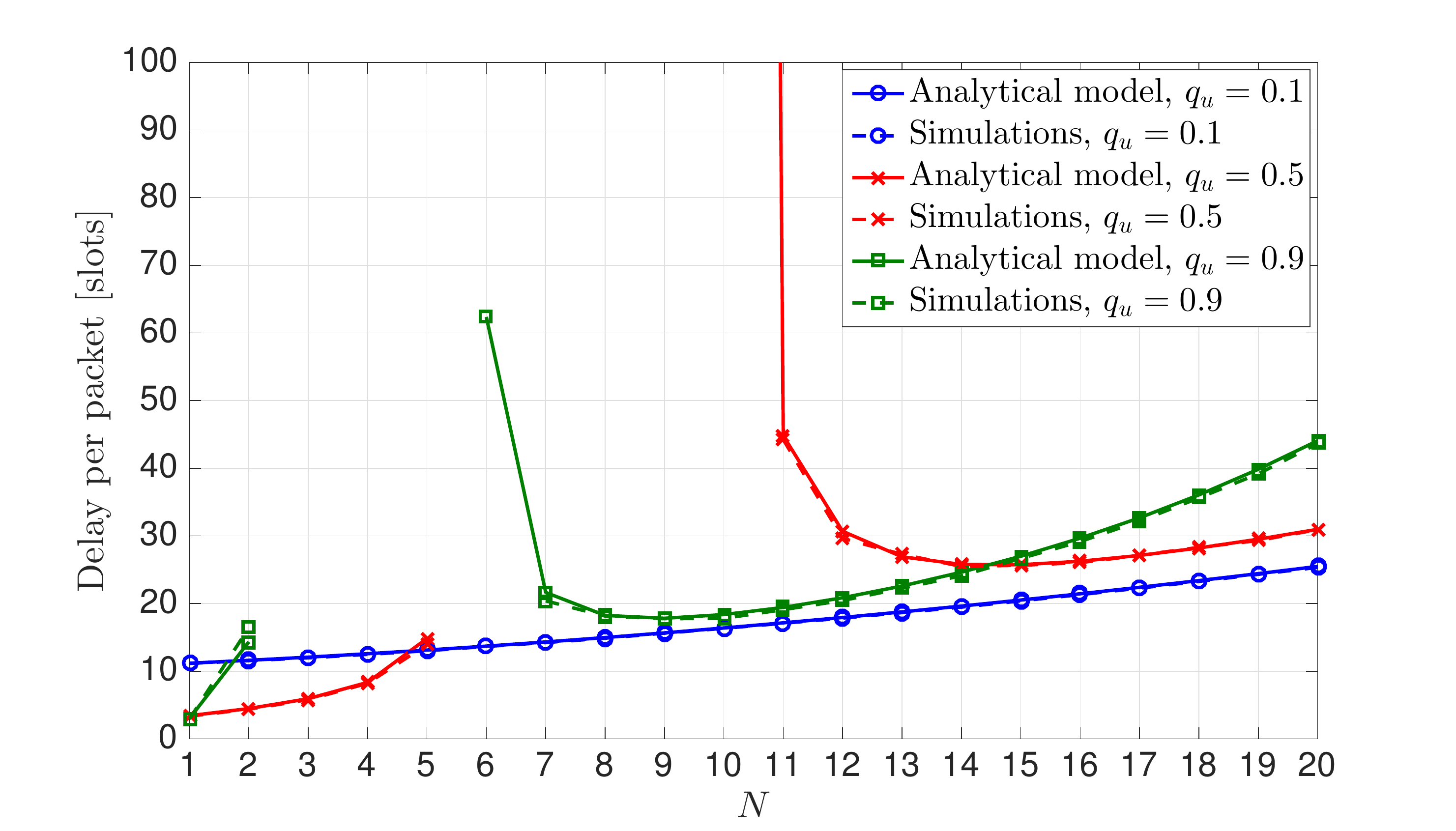}
    \caption{$D_{a}=0$.}
    \label{fig:D_N0}
  \end{subfigure}
  \hfill
  \begin{subfigure}[b]{0.45\textwidth}
    \includegraphics[width=8cm]{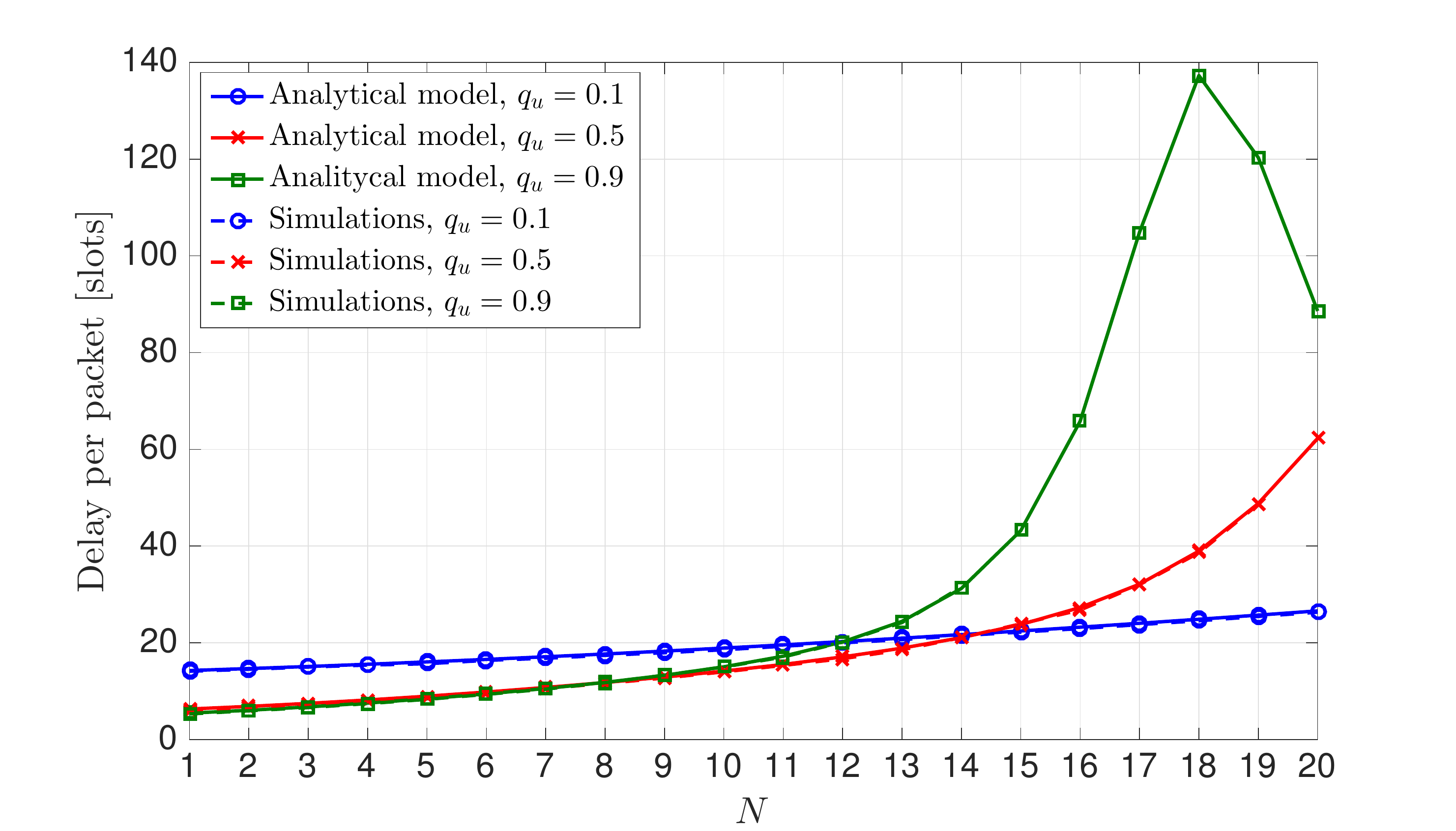}
    \caption{$D_{a}=5$.}
    \label{fig:D_N5}
  \end{subfigure}
  \caption{Delay per packet, $D$, while varying the number of UEs for several UE transmit probability values, i.e., $q_{u}$, when a) $D_{a}=0$ and b) $D_{a}=5$, respectively.}
\end{figure}

Now, by considering the same system parameters of Fig.~\ref{fig:T_N0} and Fig.~\ref{fig:T_N5}, we show the results for the delay in Fig.~\ref{fig:D_N0} and Fig.~\ref{fig:D_N5}. In these figures, we can identify the regions for which the queue at the relay is stable ($\lambda_{r} < \mu_{r}$), unstable ($\lambda_{r}>\mu_{r}$) and instable ($\lambda_{r} = \mu_{r}$). In this latest case, the arrival rate $\lambda_{r}$ is still below the service rate $\mu_{r}$, but very close to it. The three regions can be easily distinguished. Namely, for the case of instable queue we report only the analytical results (since simulation results are meaningless), for the unstable queue we do not report any results, because the delay increases towards infinity, and only for the stable case we report both analytical and simulation results. In Fig.~\ref{fig:D_N0}, for $q_{u}=0.1$, there is neither unstability nor instability regions. In contrast, for $q_{u}= 0.5$ and $q_{u}= 0.9$ the queue becomes unstable at $N= 6$ and $N= 3$, respectively, which is also approximately the point at which the throughput reaches its maximum. Furthermore, as explained for Fig.~\ref{fig:T_N0} and Fig.~\ref{fig:T_N5}, increasing $N$ causes a higher interference and a smaller arrival rate at the relay, whose queue becomes again stable at $N= 11$ and $N= 7$ for $q_{u}= 0.5$ and $q_{u}= 0.9$, respectively. Whereas, at  $N= 10$ ($q_{u}= 0.5$) and $N= 6$ ($q_{u}= 0.5$), we can clearly observe the region for which the queue is instable, $\lambda_{r} \approx \mu_{r}$, where the delay values is finite, but very high.
 \begin{figure}[!tbp]
  \begin{subfigure}[b]{0.45\textwidth}
    \includegraphics[width=8cm]{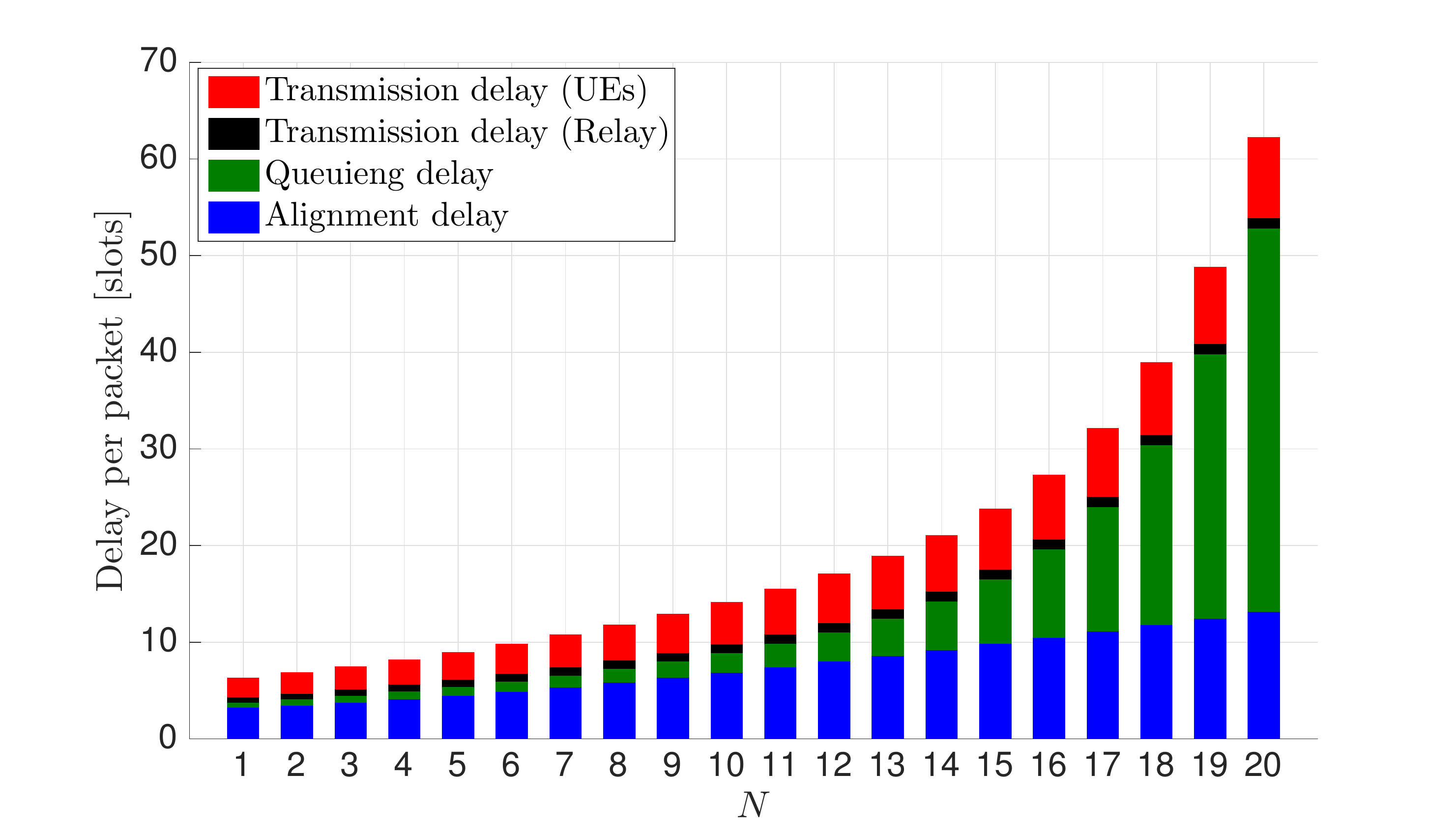}
    \caption{$\gamma=10$ dB.}
    \label{fig:D_C10}
  \end{subfigure}
  \hfill
  \begin{subfigure}[b]{0.45\textwidth}
    \includegraphics[width=8cm]{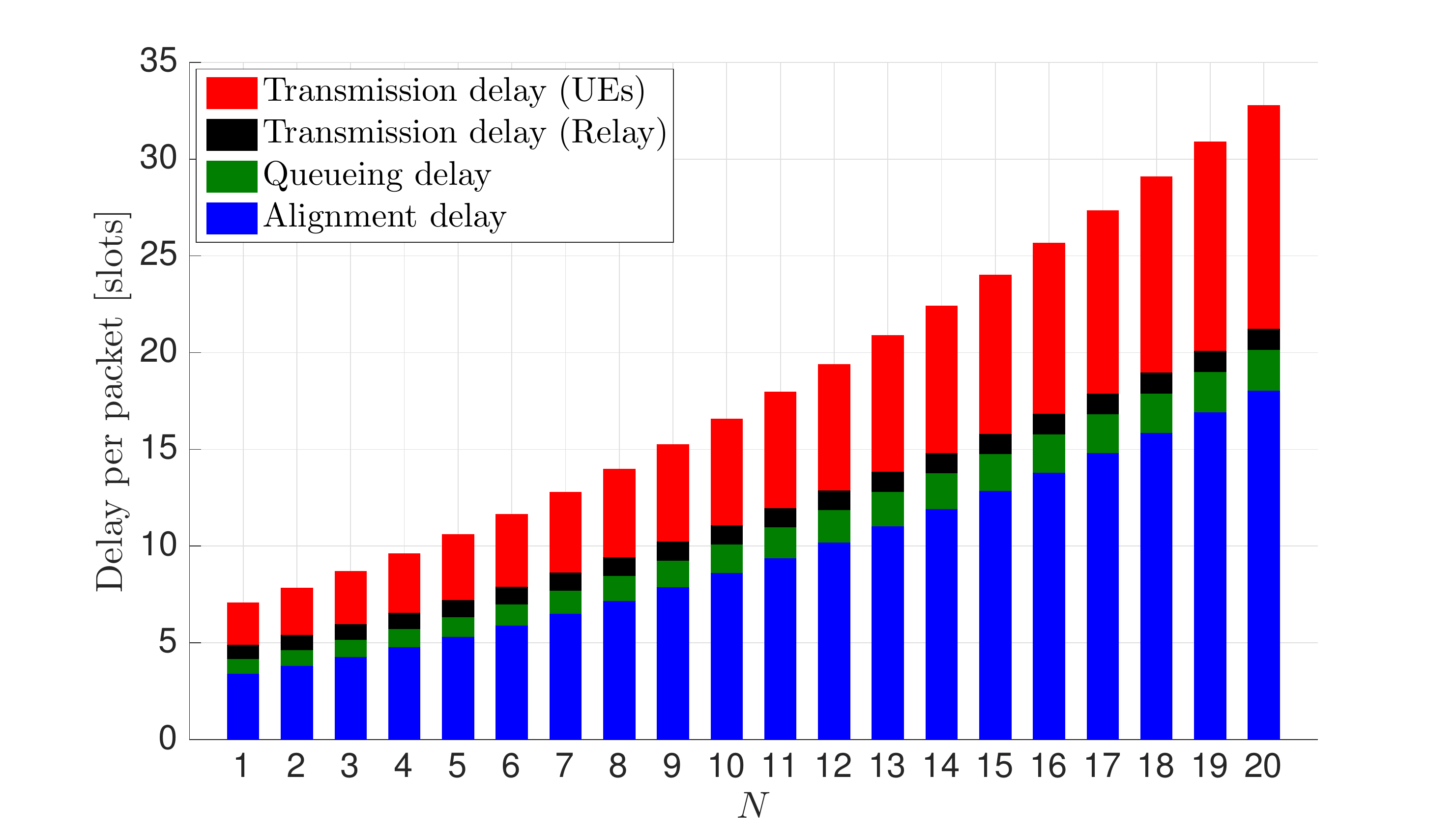}
    \caption{$\gamma=15$ dB.}
    \label{fig:D_C15}
  \end{subfigure}
  \caption{Delay per packet while varying the number of UEs, $N$, with $q_{u}=0.5$ and $D_{a}=5$.}
\end{figure}


The transmit probability $q_{tx}$ and the arrival packet rate at the relay decrease by increasing the value of the beam alignment delay to $D_{a}=5$. The effects of this can be observed in Fig.~\ref{fig:D_N5}, where the queue is never unstable and the instability regions change as well. Namely, the instability region for $q_{u}=0.5$ is visible at $N=20$, whereas, for $q_{u}=0.9$ it is between $N=15$ and $N=20$. For regions far from instability, the analytical model and the simulations almost coincide and the delay increases with the increasing number of UEs for all the curves. 

The increasing trend of the delay is caused by two main reasons: i) an increasing number of packets inside the queue and ii) increasing interference (that reduces the success probability of transmission). 
For a better understanding of this phenomenon, we show in Fig.~\ref{fig:D_C10} the delay per packet as sum of its components, i.e., UE's and relay transmission delays and queueing and alignment delays, for the red curve shown in Fig.~\ref{fig:D_N5} ($q_{u}=0.5$, $D_{a}=5$, and $\gamma=10$ dB). 
We can observe that close to the instability region ($N=20$) the biggest delay component is given by the queueing delay, whereas the transmission delays (both UE and $R$) as well as the alignment delay are barely increasing with $N$.  A different behavior can be observed in Fig.~\ref{fig:D_C15}, where the same scenario of Fig.~\ref{fig:D_C10} is considered, but for a higher SINR threshold, i.e., $\gamma=15$\,dB. In this case, the higher value of $\gamma$ reduces the success probability of transmission and the arrivals at the relay. Thus, the queueing delay does not represent anymore the main issue, nor does the relay transmission delay. In contrast, the increased unsuccessful transmission attempts make the packets waiting for being transmitted for most of the time inside the UEs' queue that increases the UE transmission and the alignment delays. This is also due to the fact that, as explained in Section \ref{sec:Ass}, after a transmission attempt the UE can change transmission strategy. 

\begin{figure}[!tbp]
  \centering
  \begin{minipage}[b]{0.45\textwidth}
    \includegraphics[width=8cm]{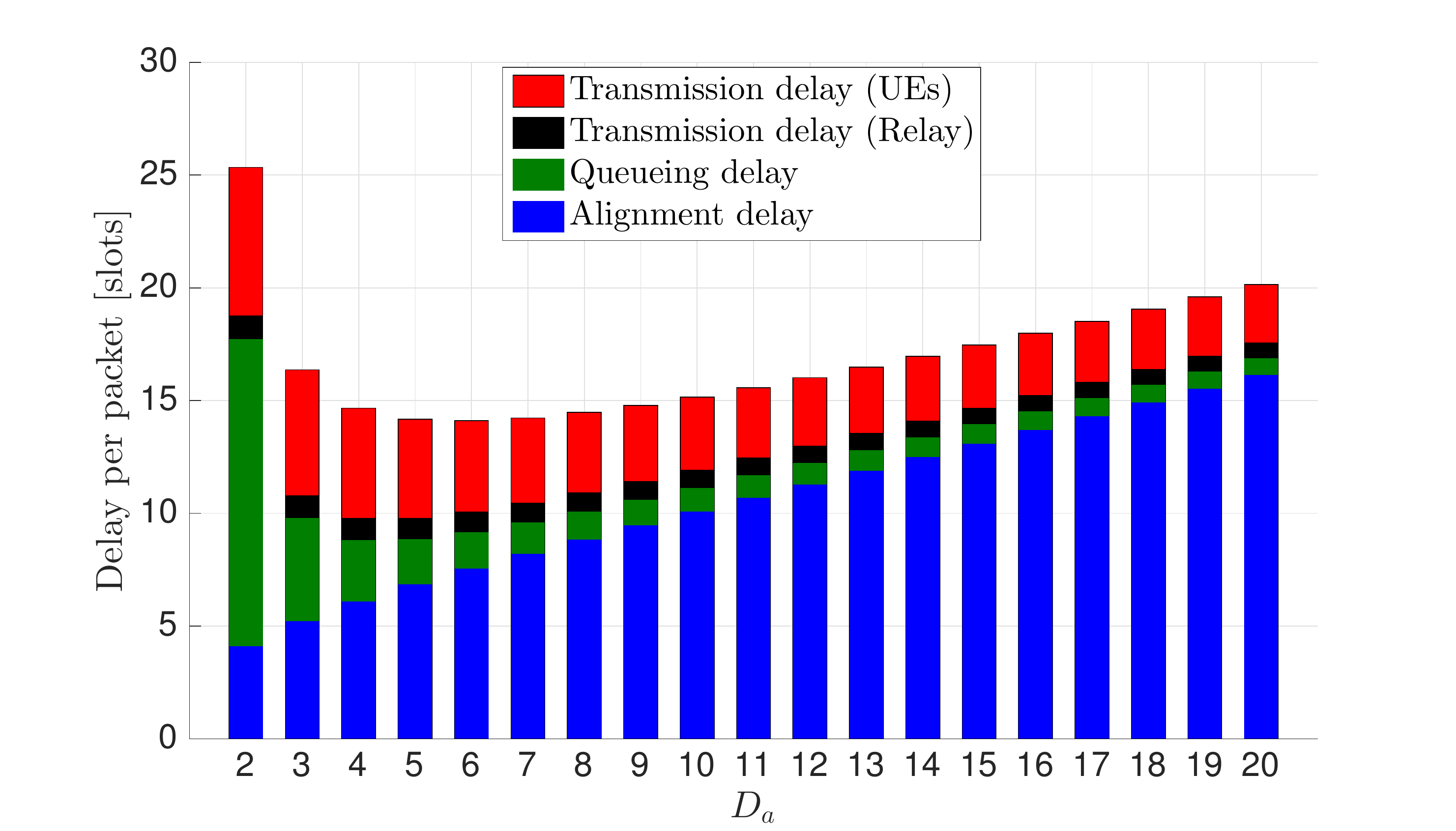}
    \caption{Delay per packet while varying the beam alignment duration, $D_{a}$, with $q_{u}=0.5$ and $N_{a}=10$.}
    \label{fig:D_Da}
  \end{minipage}
  \hfill
  \begin{minipage}[b]{0.45\textwidth}
    \includegraphics[width=8cm]{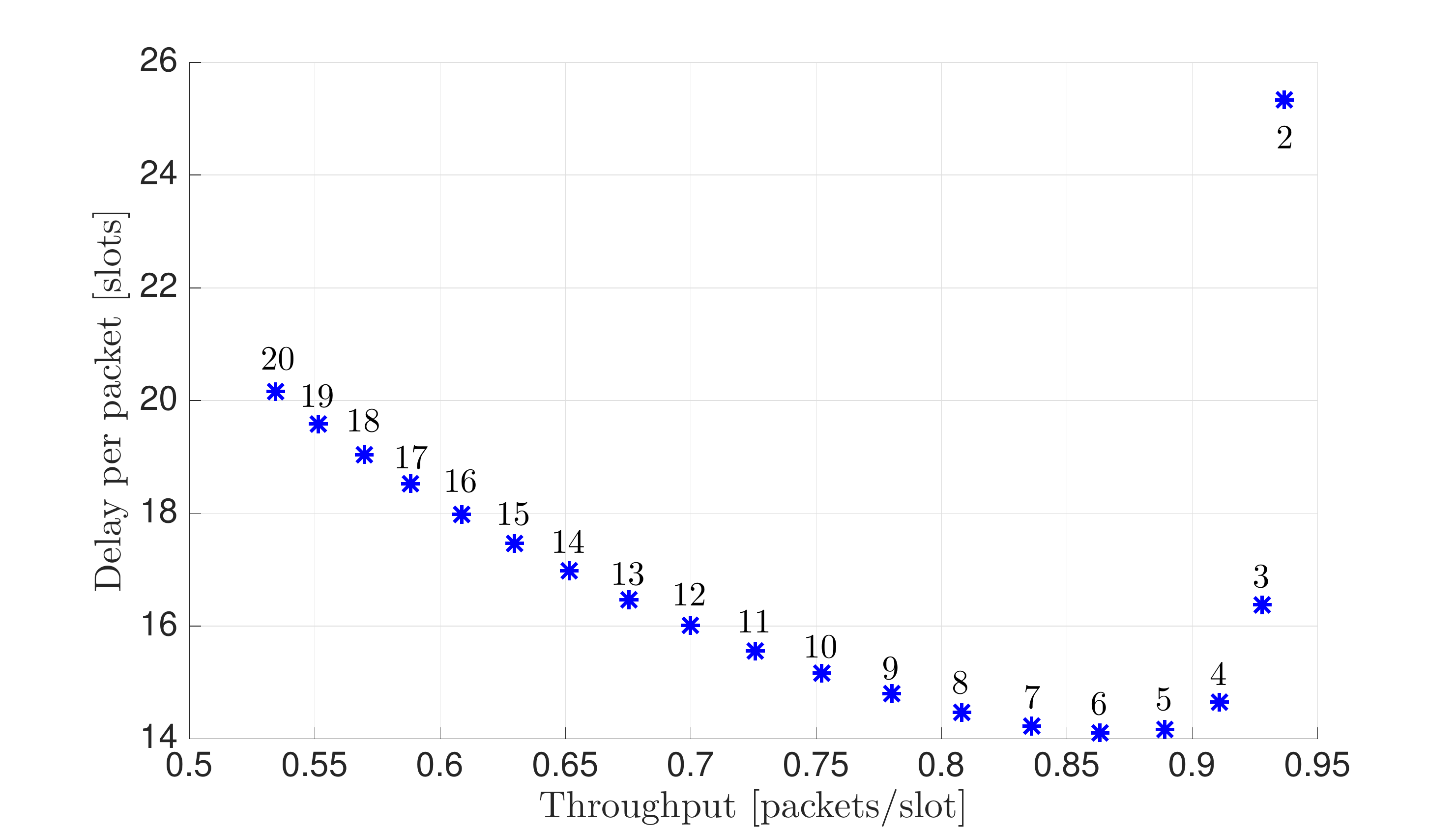}
    \caption{Throughput vs. delay tradeoff for several values of $D_{a}$, with $q_{u}=0.5$ and $N_{a}=10$. For each point we report the corresponding value of $D_{a}$.}
    \label{fig:T_D}
  \end{minipage}
\end{figure}
In Fig.~\ref{fig:D_Da}, we show the impact of the beam alignment duration $D_{a}$ on the delay components. We set the number of UEs $N=10$, $q_{u}=0.5$, and $\gamma=10$\ dB while increasing the value of $D_{a}$. First, we can observe that the delay has a non-monotonic behavior. More precisely, for $D_{a}=0$ and $D_{a}=1$ the queue is instable and unstable, respectively, and we do not report any value. Then, the delay decreases at first, mostly because the increased $D_{a}$ reduces the transmit probability and the arrival rate at the relay queue. This is confirmed by the decrease in the queueing delay that represents the highest delay component for lower values of $D_{a}$. However, above a certain value of  $D_{a}$, the delay start increasing again mostly because the alignment delay increases. In contrast to the delay, the throughput has a slightly different behavior. In Fig.~\ref{fig:T_D}, we show the throughput and delay tradeoff for several values of $D_{a}$ and it can be observed that the throughput monotonically decreases.

\subsection{Optimal Transmission Strategy}
\label{sec:OptRes}
Hereafter, we set $q_{u}= 0.1$ and $N=10$ (i.e., the queue is always stable) and we study the effect of the two transmission strategies ($\mathrm{FD}$ and $\mathrm{BR}$) on the throughput and the delay. In Fig.~\ref{fig:T_quf_0}, we set $q_{ur}= 0.5$ and show the aggregate throughput, $T$, while varying the probability of using the $\mathrm{FD}$ transmission, $q_{uf}$, and $\theta_{rd}$, for $D_{a}=0$. The solid blue line shows the values of $q_{uf}$ that maximizes the throughput for each value of $\theta_{rd}$. Namely, for small values of $\theta_{rd}$, the $\mathrm{BR}$ transmission is preferable (corresponding to small values of $q_{uf}$). In this case, we can use a narrow beam with high beamforming gain to transmit simultaneously to $R$ and the mmAP. In contrast, for higher values of $\theta_{rd}$, the optimal value of $q_{uf}$ becomes $1$, which corresponds to using the $\mathrm{FD}$ transmission. For $D_{a}=5$, in Fig.~\ref{fig:T_quf_5}, we have almost the same behavior. However, the selection of the best strategy is either $q_{uf}=0$ or $q_{uf}=1$, since the number of beam alignments is minimized and $q_{tx}$ maximized when $q_{uf}=0$ or $q_{uf}=1$. 
\begin{figure}[!tbp]
  \begin{subfigure}[b]{0.45\textwidth}
    \includegraphics[width=8cm]{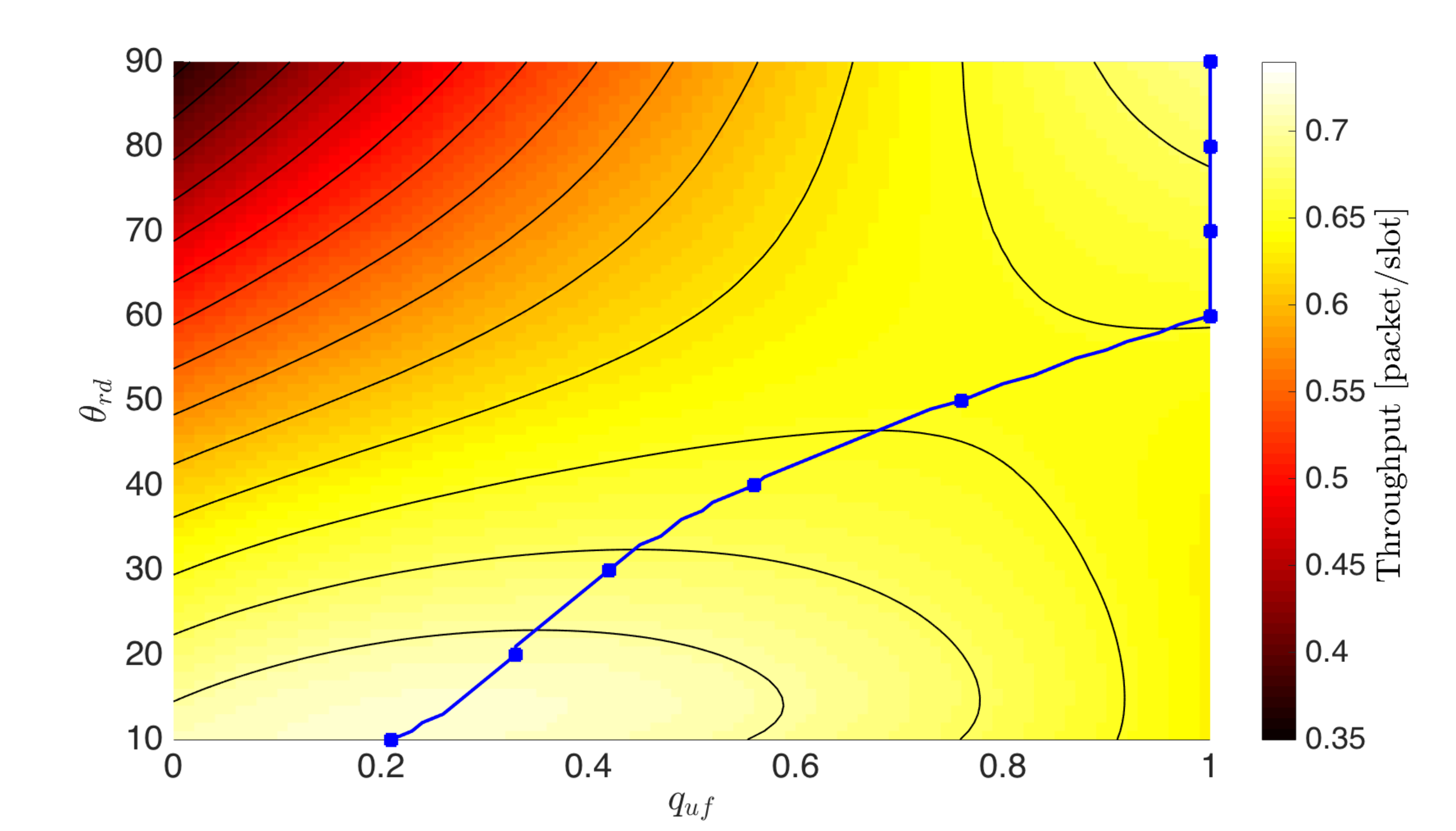}
    \caption{$D_{a}=0$.}
    \label{fig:T_quf_0}
  \end{subfigure}
  \hfill
  \begin{subfigure}[b]{0.45\textwidth}
    \includegraphics[width=8cm]{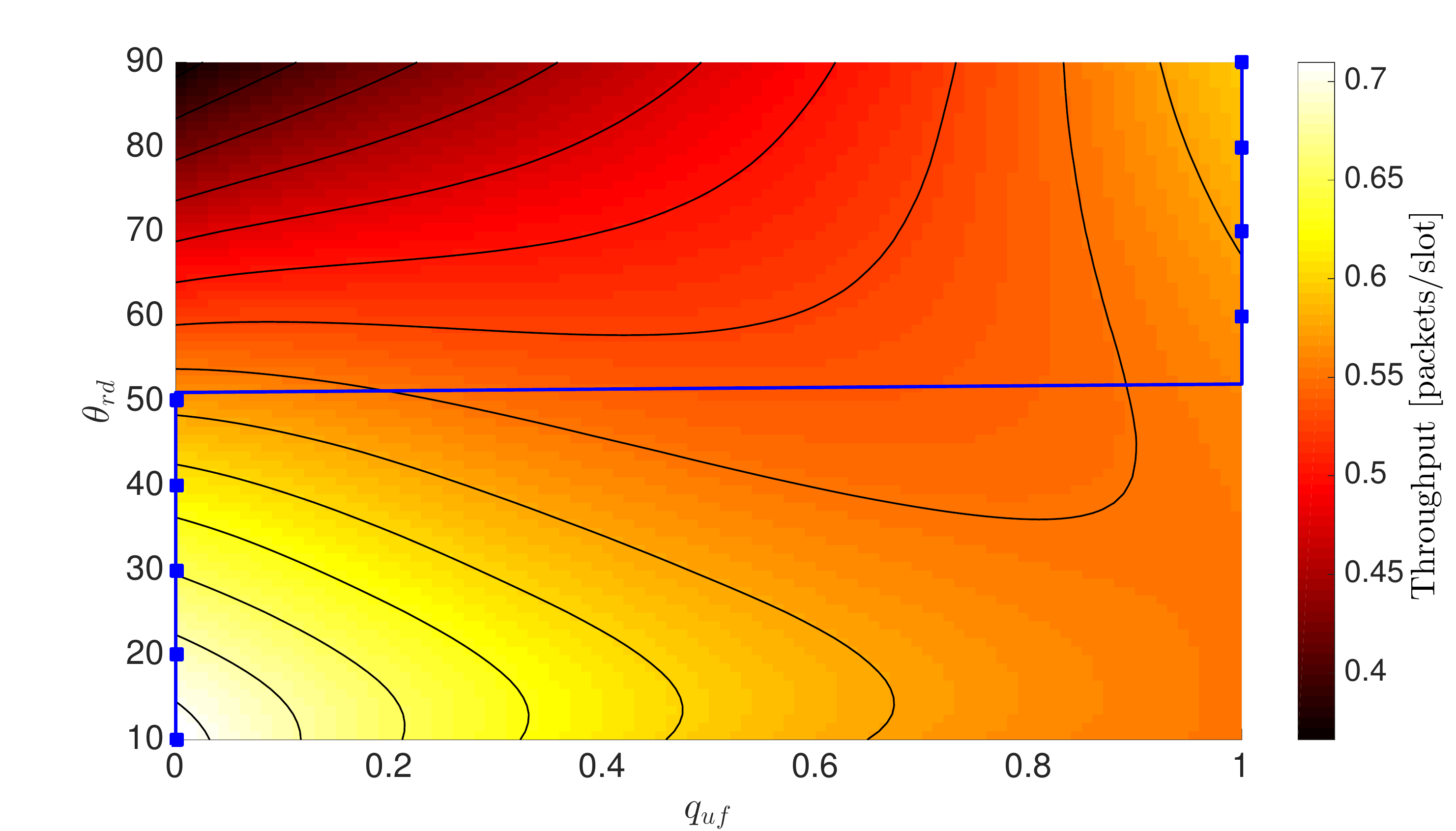}
    \caption{$D_{a}=5$.}
    \label{fig:T_quf_5}
  \end{subfigure}
  \caption{Throughput, $T$, while varying $q_{uf}$ and $\theta_{rd}$, with $q_{ur}=0.5$. For each value of $\theta_{rd}$, we represent with a blue solid line the value of $q_{uf}$ that maximizes the throughput.}
\end{figure}
\begin{figure}[!tbp]
  \begin{subfigure}[b]{0.45\textwidth}
    \includegraphics[width=8cm]{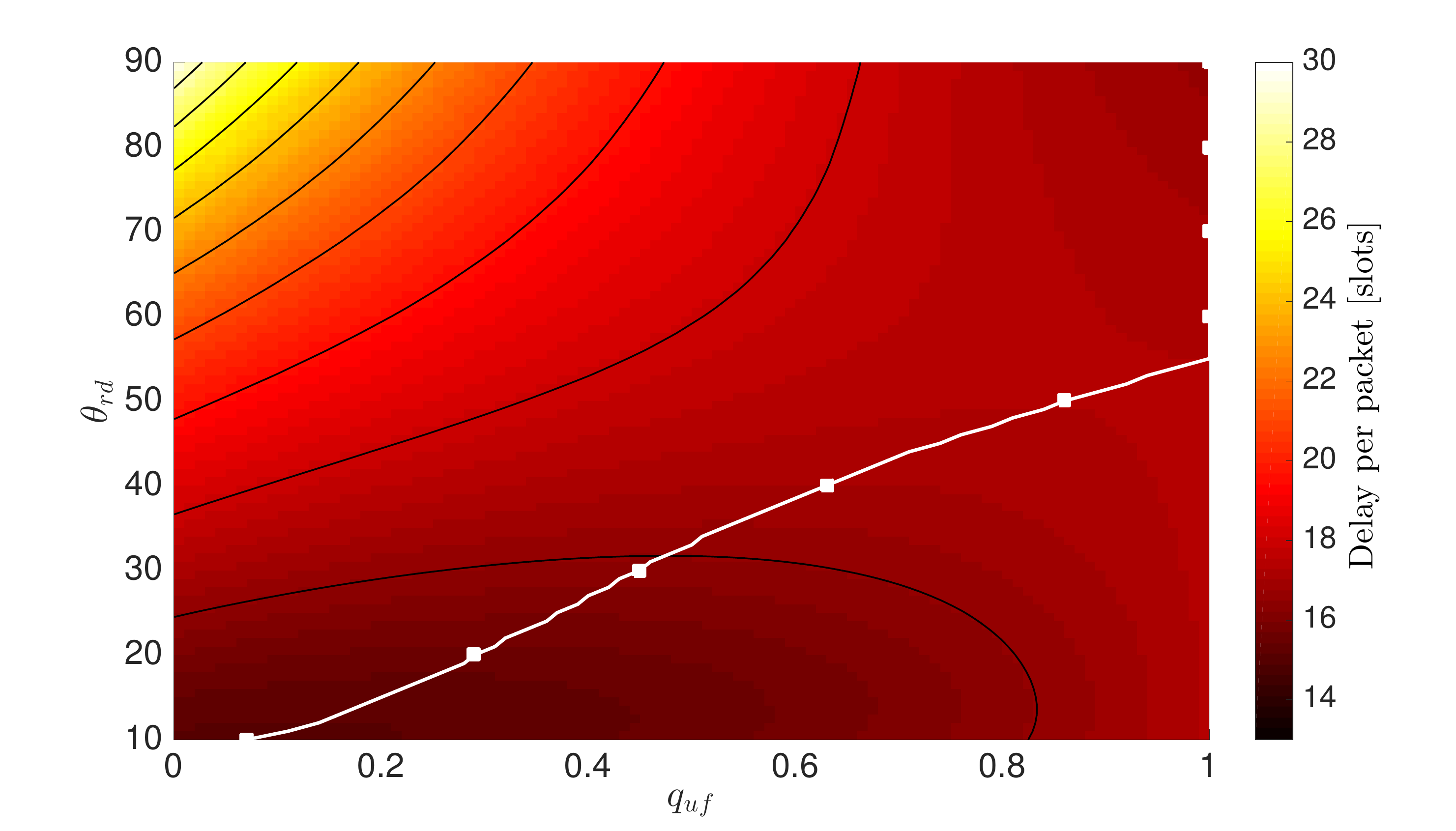}
    \caption{$D_{a}=0$.}
    \label{fig:D_quf_0}
  \end{subfigure}
  \hfill
  \begin{subfigure}[b]{0.45\textwidth}
    \includegraphics[width=8cm]{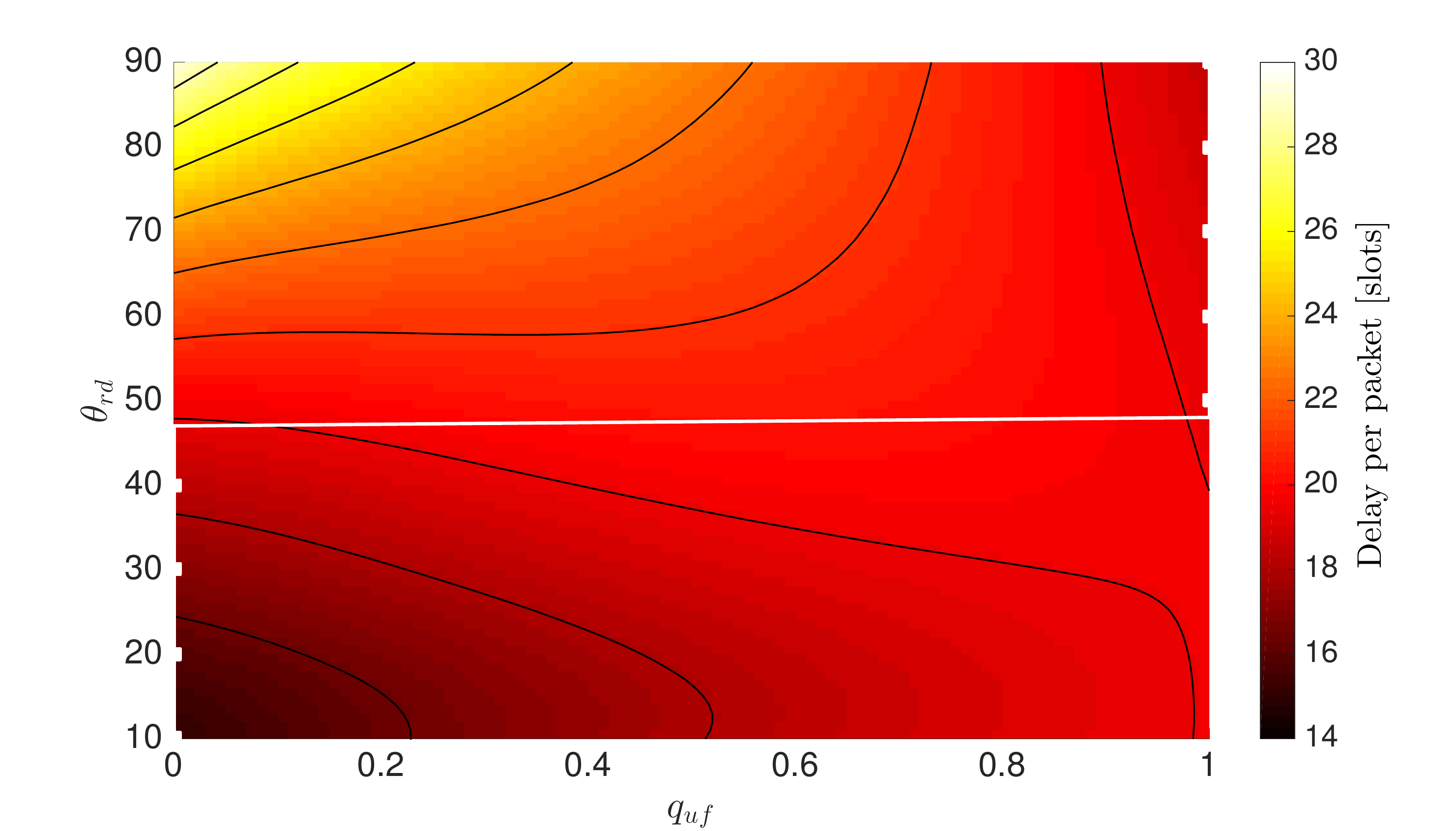}
    \caption{$D_{a}=5$.}
    \label{fig:D_quf_5}
  \end{subfigure}
  \caption{Delay per packet, $D$, while varying $q_{uf}$ and $\theta_{rd}$, with $q_{ur}=0.5$. For each value of $\theta_{rd}$, we represent with a white solid line the value of $q_{uf}$ that minimizes the delay.}
\end{figure} 
In Fig.~\ref{fig:D_quf_0} and Fig.~\ref{fig:D_quf_5} we show the delay for the same setting of Fig.~\ref{fig:T_quf_0} and Fig.~\ref{fig:T_quf_5}, respectively, while varying $q_{uf}$, and $\theta_{rd}$. For this scenario, where the queue is stable, the highest contributions to the delay are given by the transmission and alignment delays. For this reason the strategy that minimizes the delay, which is shown with a white solid line, follows almost the same behavior of the transmission strategy that maximizes the throughput.

\begin{figure}[!tbp]
  \begin{subfigure}[b]{0.45\textwidth}
    \includegraphics[width=8cm]{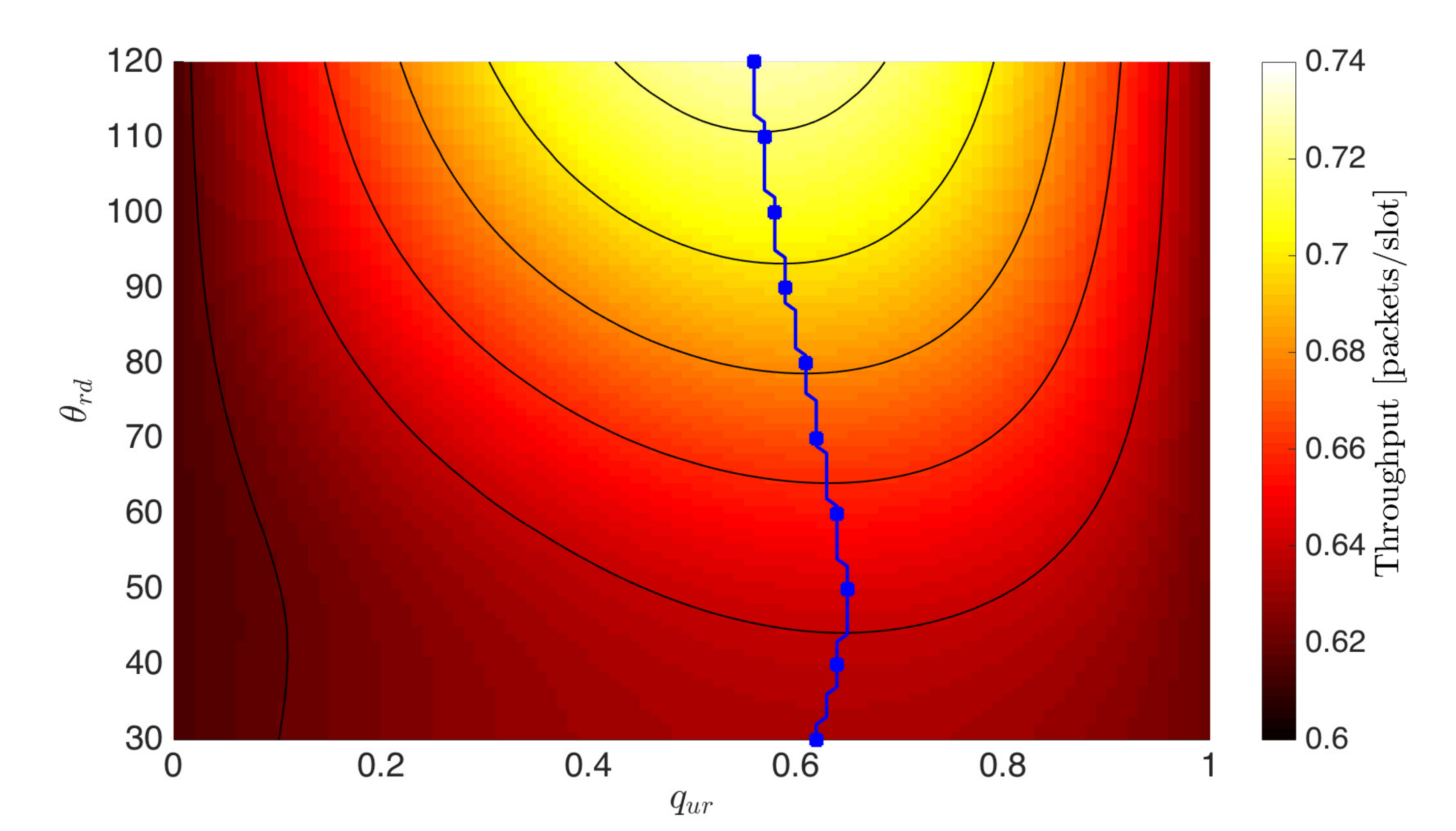}
    \caption{$D_{a}=0$.}
    \label{fig:T_qur_0}
  \end{subfigure}
  \hfill
  \begin{subfigure}[b]{0.45\textwidth}
    \includegraphics[width=8cm]{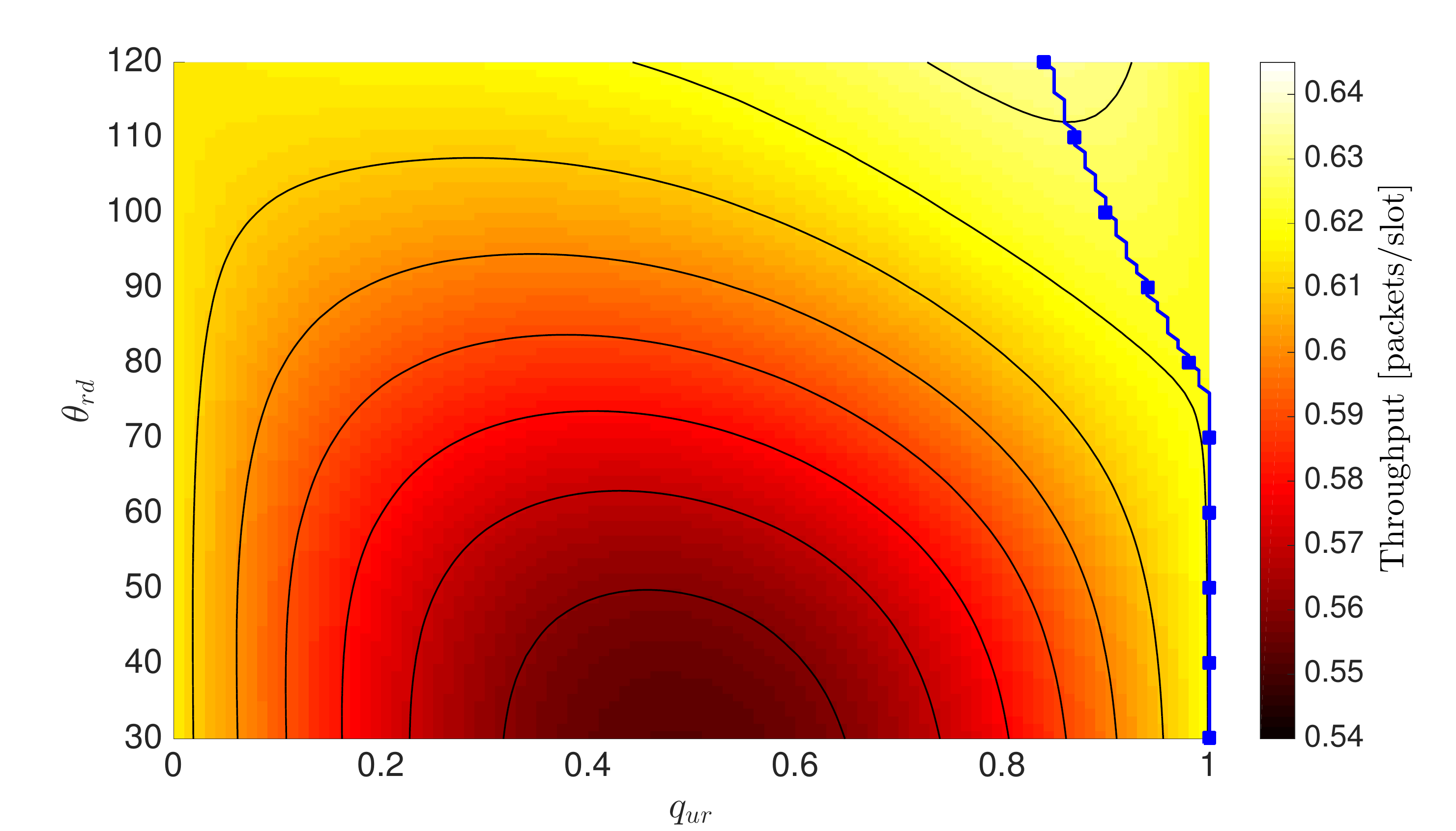}
    \caption{$D_{a}=5$.}
    \label{fig:T_qur_5}
  \end{subfigure}
  \caption{Throughput, $T$, while varying $q_{ur}$ and $\theta_{rd}$ with $q_{uf}=1$. For each value of $\theta_{rd}$, we represent with a solid blue line the value of $q_{ur}$ that maximizes the throughput.}
	\label{fig:T_qur}
\end{figure}
\begin{figure}[!tbp]
  \begin{subfigure}[b]{0.45\textwidth}
    \includegraphics[width=8cm]{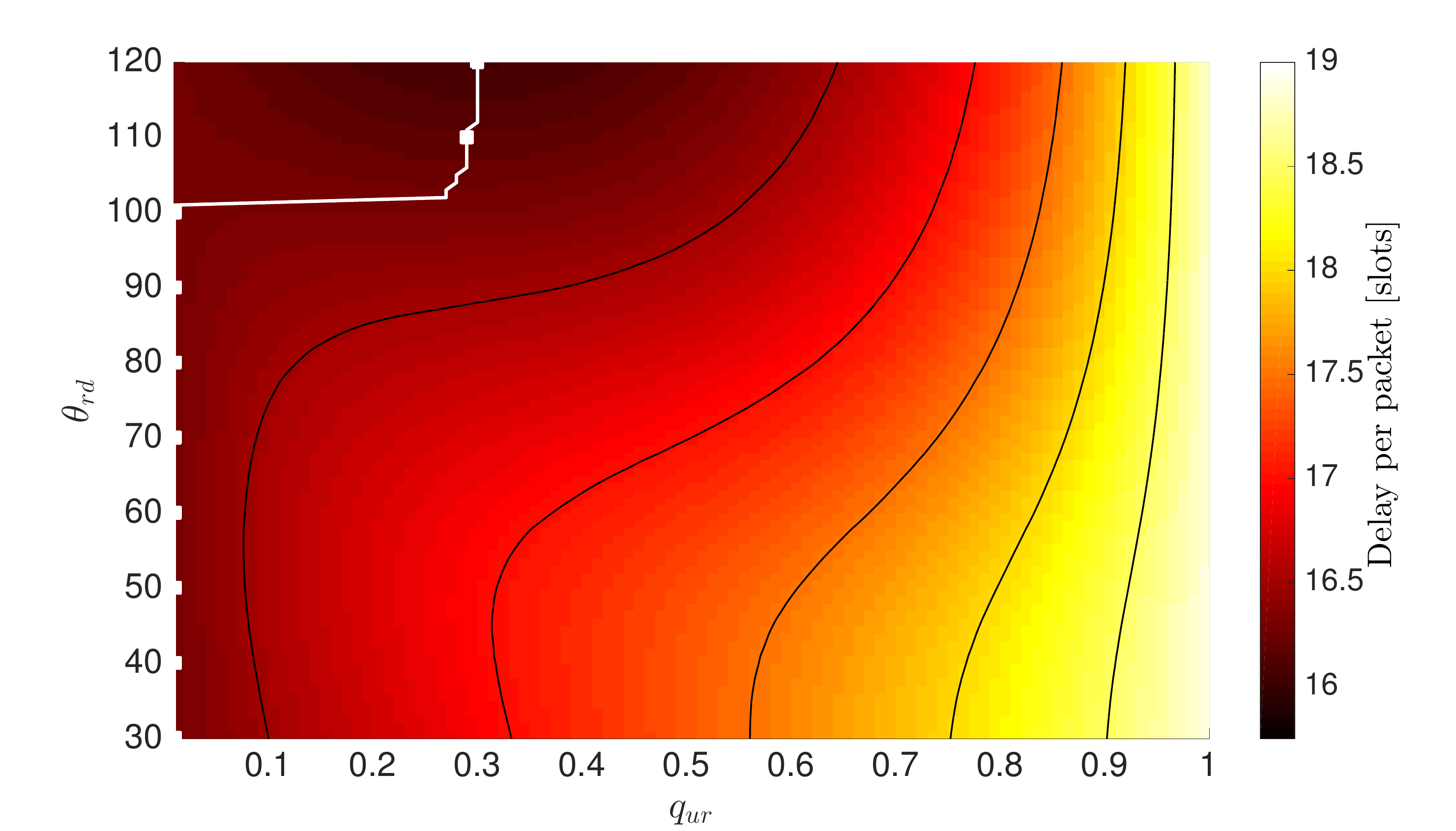}
    \caption{$D_{a}=0$.}
    \label{fig:D_qur_0}
  \end{subfigure}
  \hfill
  \begin{subfigure}[b]{0.45\textwidth}
    \includegraphics[width=8cm]{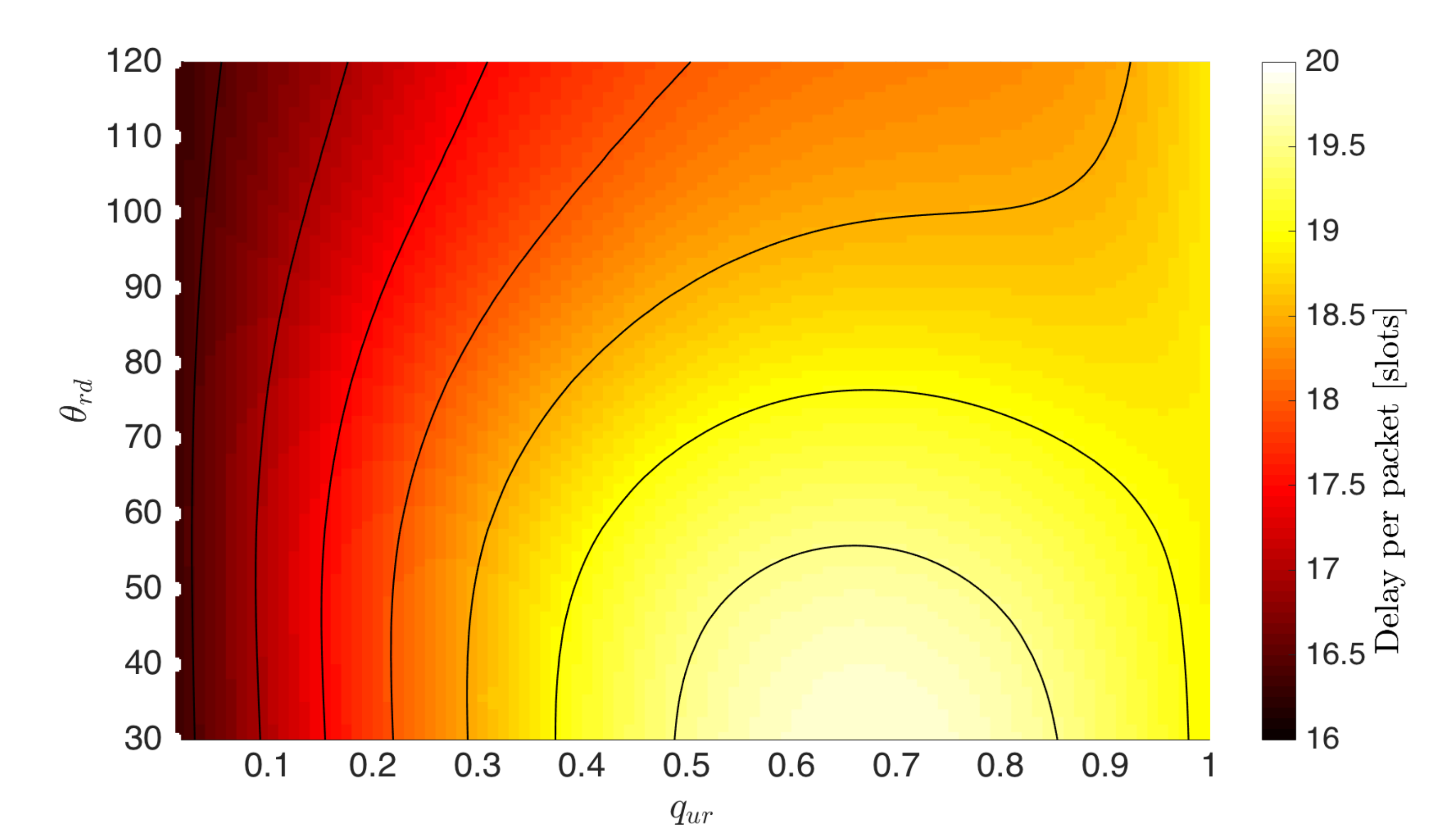}
    \caption{$D_{a}=5$.}
    \label{fig:D_qur_5}
  \end{subfigure}
  \caption{Delay per packet, $D$, while varying $q_{ur}$ and $\theta_{rd}$, with $q_{uf}=1$. For each value of $\theta_{rd}$, we represent with a white solid line the value of $q_{ur}$ that minimizes the delay.}
	\label{fig:D_qur}
\end{figure}

In Fig.~\ref{fig:T_qur}, we show the throughput $T$ while varying the probability to transmit at the relay $q_{ur}$ and $\theta_{rd}$, when the $\mathrm{FD}$ transmission is used, i.e., $q_{uf} = 1$. 
Note that larger values of $\theta_{rd}$ correspond to longer distances between $R$ and the mmAP, i.e., $d_{rd}$ (see Fig.~\ref{fig:FB}).  Thus, when $\theta_{rd}$ increases, the interference of the relay on the transmission to the mmAP decreases. As explained in Section~\ref{sec:Ass}, the relay is in LOS with the mmAP and uses always the $\mathrm{FD}$ transmission with high beamforming gain. Such high gain can cause high interference at the receiver side of the mmAP. As results of this, we can observe higher throughput for larger values of $\theta_{rd}$. Indeed, packets that are successfully transmitted by the relay are barely affected by increasing the distance between $R$ and the mmAP. In contrast, the packets that are successfully transmitted by UEs to the mmAP increases for wider $\theta_{rd}$ because the interference caused by $R$ decreases. 
In Fig.~\ref{fig:T_qur_0}, for which $D_{a}=0$, the strategy that maximizes  $T$ is shown by a solid blue line and is $q_{ur}\approx 0.6$ for all the values of $\theta_{rd}$. In contrast, in Fig.~\ref{fig:T_qur_5}, where $D_{a}=5$, we observe a different behavior. The optimal strategy coincides with $q_{ur}=1$ for lower value of $\theta_{rd}$ that allows to minimize the probability to change strategy. When $\theta_{rd}$ increases, the highest throughput is provided by a slightly smaller value of $q_{ur}$, with an increase in the transmissions to the mmAP.

\begin{figure}[!tbp]
  \begin{subfigure}[b]{0.45\textwidth}
    \includegraphics[width=8cm]{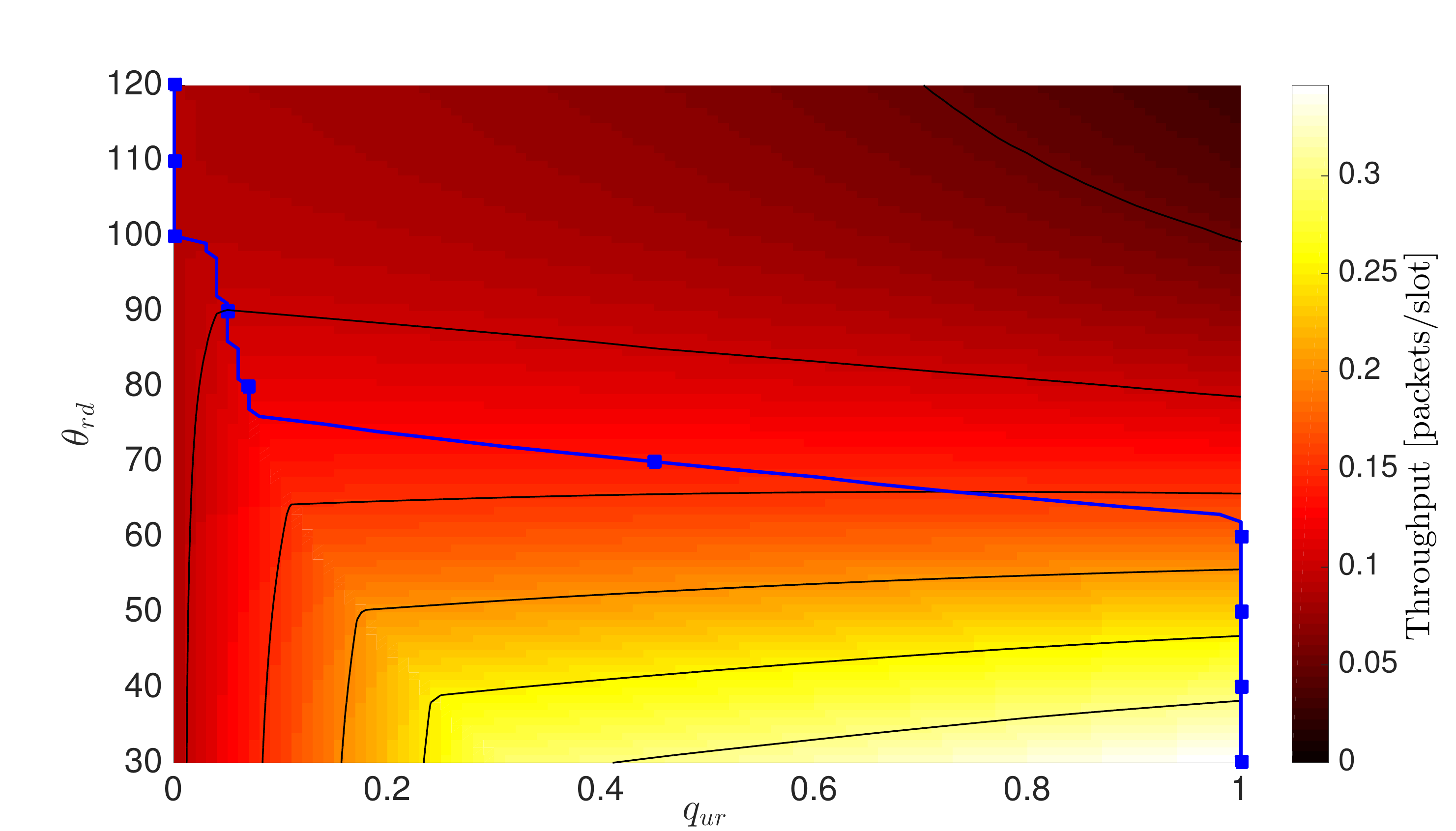}
    \caption{$D_{a}=0$.}
    \label{fig:T_qur_0_200}
  \end{subfigure}
  \hfill
  \begin{subfigure}[b]{0.45\textwidth}
    \includegraphics[width=8cm]{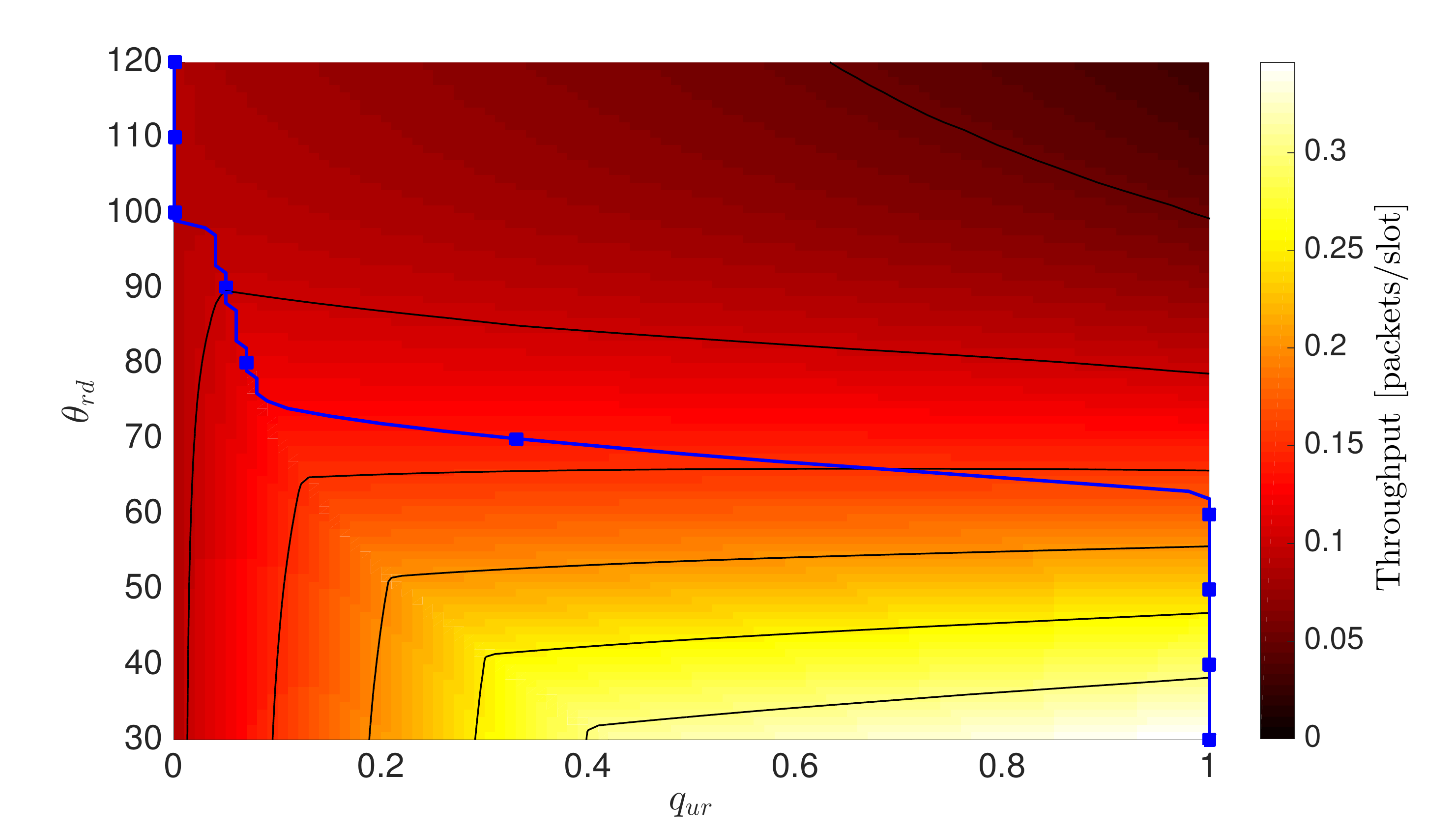}
    \caption{$D_{a}=5$.}
    \label{fig:T_qur_5_200}
  \end{subfigure}
  \caption{Throughput, $T$, while varying $q_{ur}$ and $\theta_{rd}$ with $q_{uf}=1$, $d_{ur}=30$~m and $d_{ud}=50$~m. For each value of $\theta_{rd}$, we represent with a solid blue line the value of $q_{ur}$ that maximizes the throughput.}
\end{figure}
\begin{figure}[!tbp]
  \begin{subfigure}[b]{0.45\textwidth}
    \includegraphics[width=8cm]{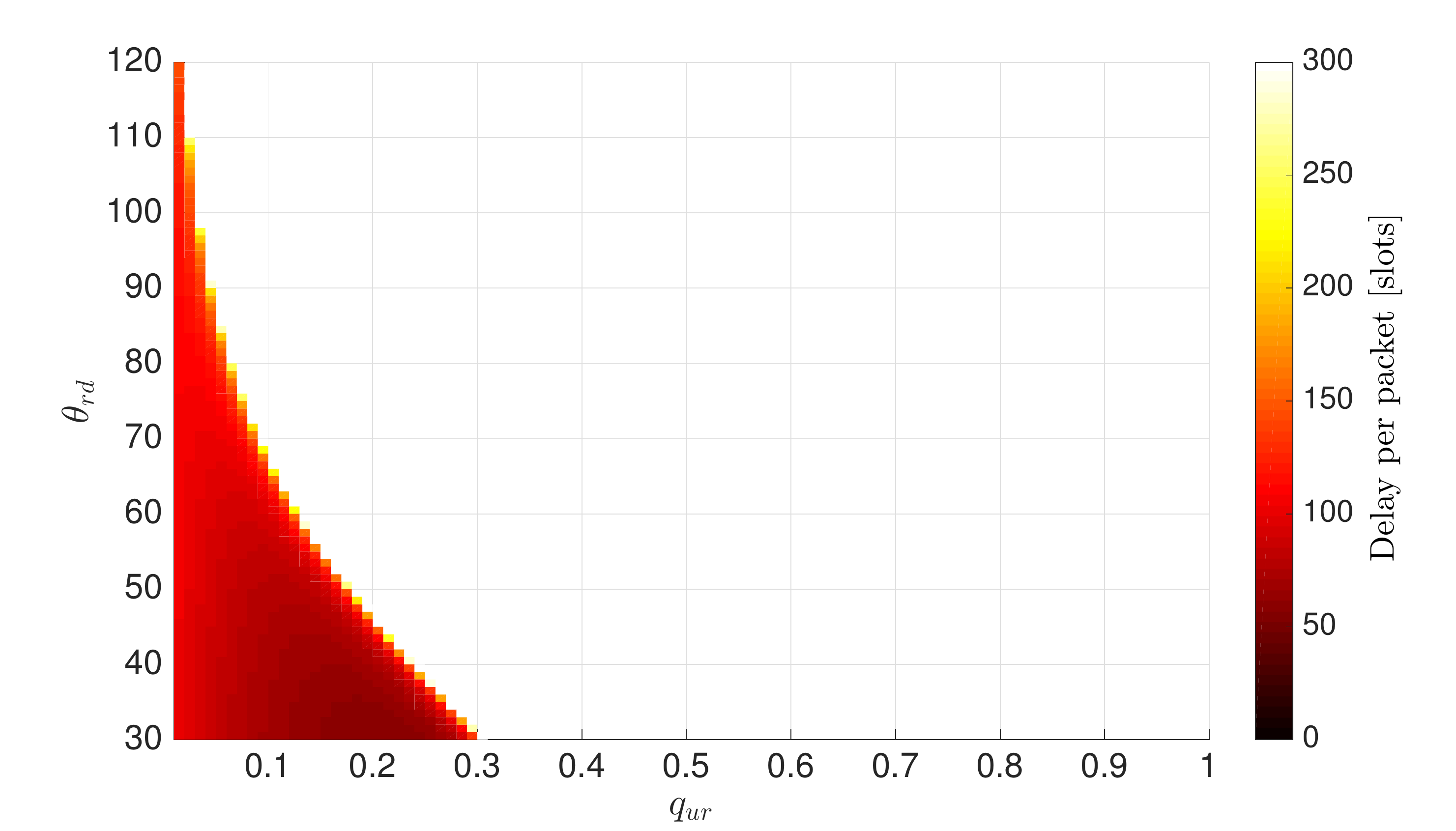}
    \caption{$D_{a}=0$.}
    \label{fig:D_qur_0_200}
  \end{subfigure}
  \hfill
  \begin{subfigure}[b]{0.45\textwidth}
    \includegraphics[width=8cm]{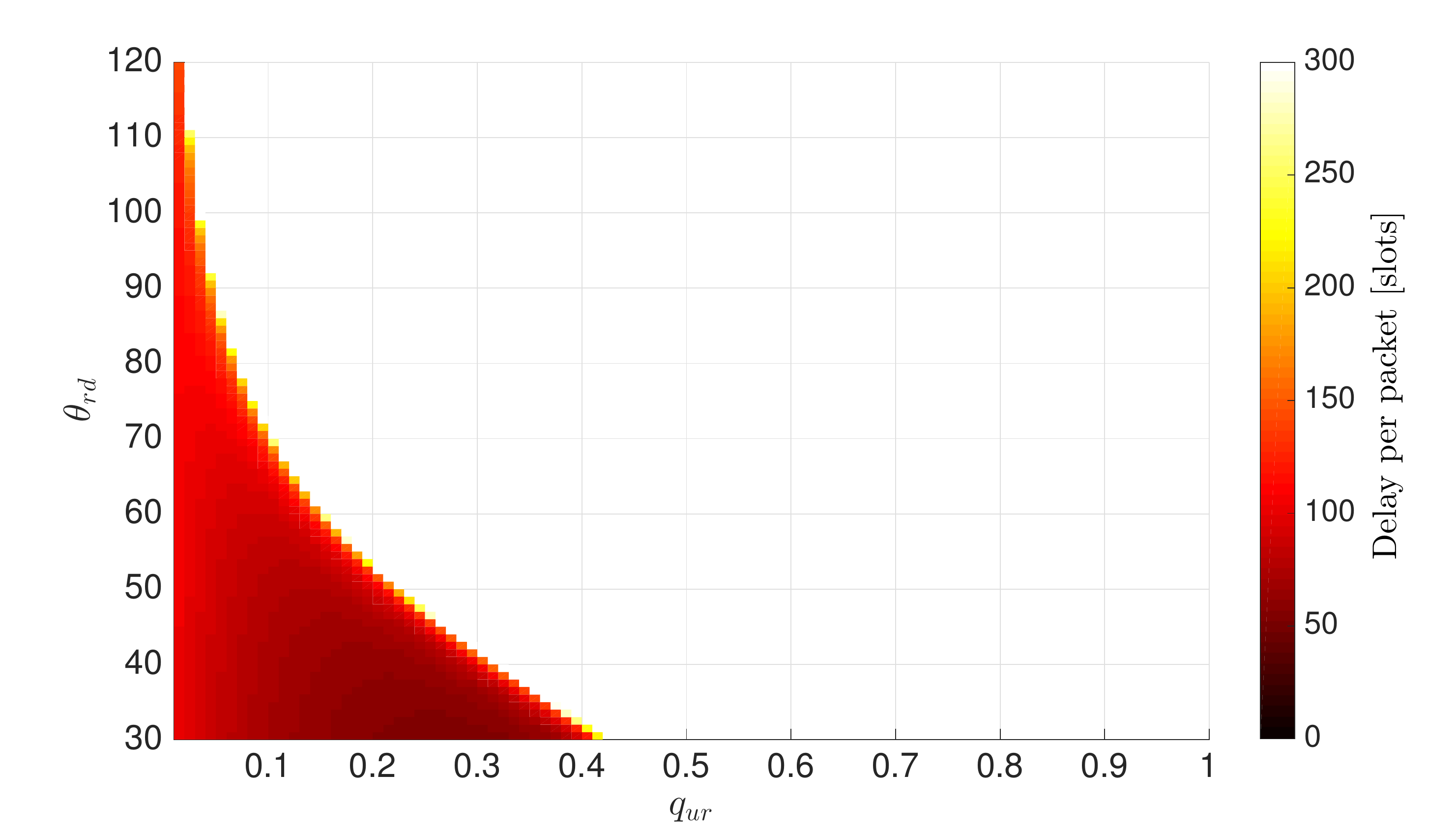}
    \caption{$D_{a}=5$.}
    \label{fig:D_qur_5_200}
  \end{subfigure}
  \caption{Delay per packet, $D$, while varying $q_{ur}$ and $\theta_{rd}$, with $q_{uf}=1$, $d_{ur}=30$~m and $d_{ud}=50$~m.}
\end{figure}

For the same settings of Fig.~\ref{fig:T_qur_0} and Fig.~\ref{fig:T_qur_5}, we show the results for delay in Fig.~\ref{fig:D_qur_0} and Fig.~\ref{fig:D_qur_5}, respectively. First, we can observe that the best transmission strategy for the delay is different from the one for the throughput. Indeed, for the throughput it is more beneficial to transmit to the relay $R$, while it is preferable to transmit to the mmAP for minimizing the delay. Namely, by transmitting to the mmAP the packets avoid the queueing delay at the relay and the interference that this creates on the mmAP, since the relay transmits a lower amount of packets.
In Fig.~\ref{fig:T_qur_0} and Fig.~\ref{fig:T_qur_5}, we observe that for the $\mathrm{FD}$ transmission and short distances ($d_{ur} = 30$ m and $d_{ud} = 50$ m), we have higher values of $T$ as $\theta_{rd}$ increases. 


Finally, we show the effects of increasing the distances, i.e, $d_{ud}$ and $d_{ur}$, when considering the same scenario of Fig.~\ref{fig:T_qur}. In Fig.~\ref{fig:T_qur_0_200}, we show the throughput with longer distances i.e., $d_{ur}=50$ m, $d_{ud}=200$ m, and $D_{a}=0$. The blue solid line shows the transmission strategy for maximizing the throughput. For lower values of $\theta_{rd}$ the transmissions to the relay are preferable, as in Fig.~\ref{fig:T_qur_0}. However, for longer distances $d_{rd}$ (high values of $\theta_{rd}$), the relay does not become anymore beneficial and in general $T$ decreases. As a result, the transmissions between the UEs and the mmAP are barely affected by the interference of $R$ and the path loss between $R$ and the mmAP is dominant. This path loss decreases the success probability for a packet from $R$ to the mmAP and makes the queue at $R$ not stable when $q_{ur}$ is above certain values. The unstabilty and instability regions can be better observed in Fig.~\ref{fig:D_qur_0_200}, where we show the delay for $D_{a}=0$. Here, we do not report (i.e., white area) the values of $q_{ur}$ and $\theta_{rd}$ for which the queue is unstable. For higher value of $D_{a}$, we can observe in Fig~\ref{fig:D_qur_0_200} and Fig~\ref{fig:D_qur_5_200} that the unstability region changes, but the optimal strategy for the throughput has almost the same behavior.

\subsection{Imperfect Beam Alignment}
\label{sec:Beam_Err}
Let us consider the problem of misalignment. Given the sectored antenna model described in Section~\ref{sec:SINRexp}, we introduce a beam alignment error ($\epsilon$) that is modelled by using the truncated gaussian error model in~\cite{BeamErr}. Thus, the transmitter and receiver gains in~\eqref{eq:SINR} can be computed as follows:
\begin{align*} 
    g=
    \begin{cases}
      \frac{2\pi}{\theta_{BW}} & \text{with probability  }  P_{G}(\sigma), \\
      0 & \text{with probability  } 1-P_{G}(\sigma).
    \end{cases}\stepcounter{equation}\tag{\theequation}\label{eq:Gain_err}
\end{align*} 
The term $P_{G}(\sigma)$ is the probability that the absolute value of the error is less than the beamwidth and it is given by:
\begin{align*} 
P_{G}(\sigma)=P(|\epsilon|\le \theta_{BW})=\frac{\text{Erf}\bigl(\theta_{BW}/\sqrt{2\sigma^2}\bigr)}{\text{Erf}\bigl(\pi/\sqrt{2\sigma^2}\bigr)},\stepcounter{equation}\tag{\theequation}\label{eq:prob_succ_gain}
\end{align*} 
\begin{figure}[!tbp]
  \begin{subfigure}[b]{0.45\textwidth}
    \includegraphics[width=8cm]{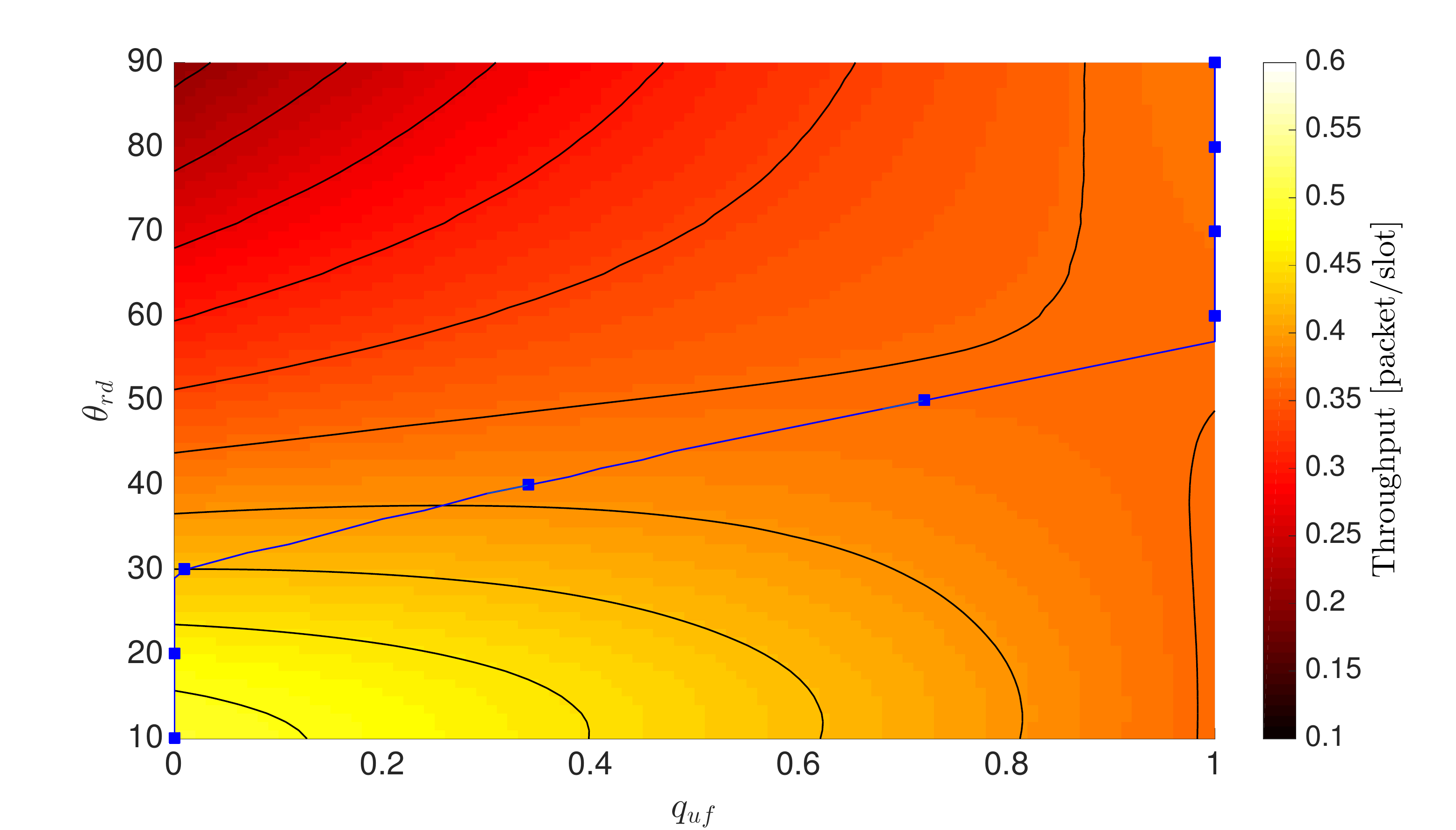}
    \caption{$\sigma^{2}=5^{\circ}$.}
    \label{fig:T_quf_th_ERR_BW5}
  \end{subfigure}
  \hfill
  \begin{subfigure}[b]{0.45\textwidth}
    \includegraphics[width=8cm]{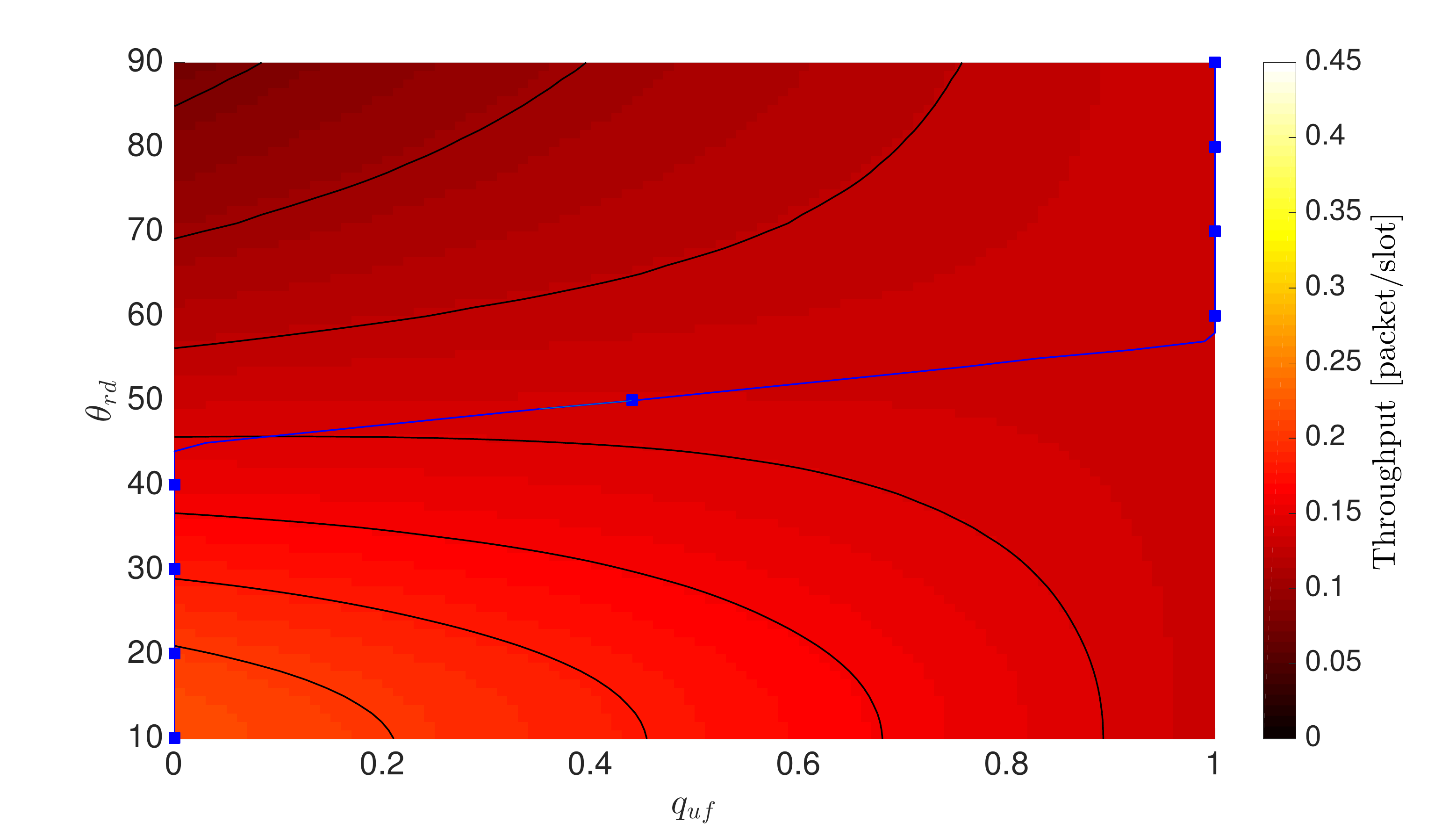}
    \caption{$\sigma^{2}=10^{\circ}$.}
    \label{fig:T_quf_th_ERR_BW10}
  \end{subfigure}
  \caption{Throughput with imperfect beam alignment while varying $q_{uf}$ and $\theta_{rd}$. Moreover, we set $q_{ur}=0.5$, $d_{ur}=30$~m and $d_{ud}=50$~m.}
\end{figure}
\begin{figure}[!tbp]
  \begin{subfigure}[b]{0.45\textwidth}
    \includegraphics[width=8cm]{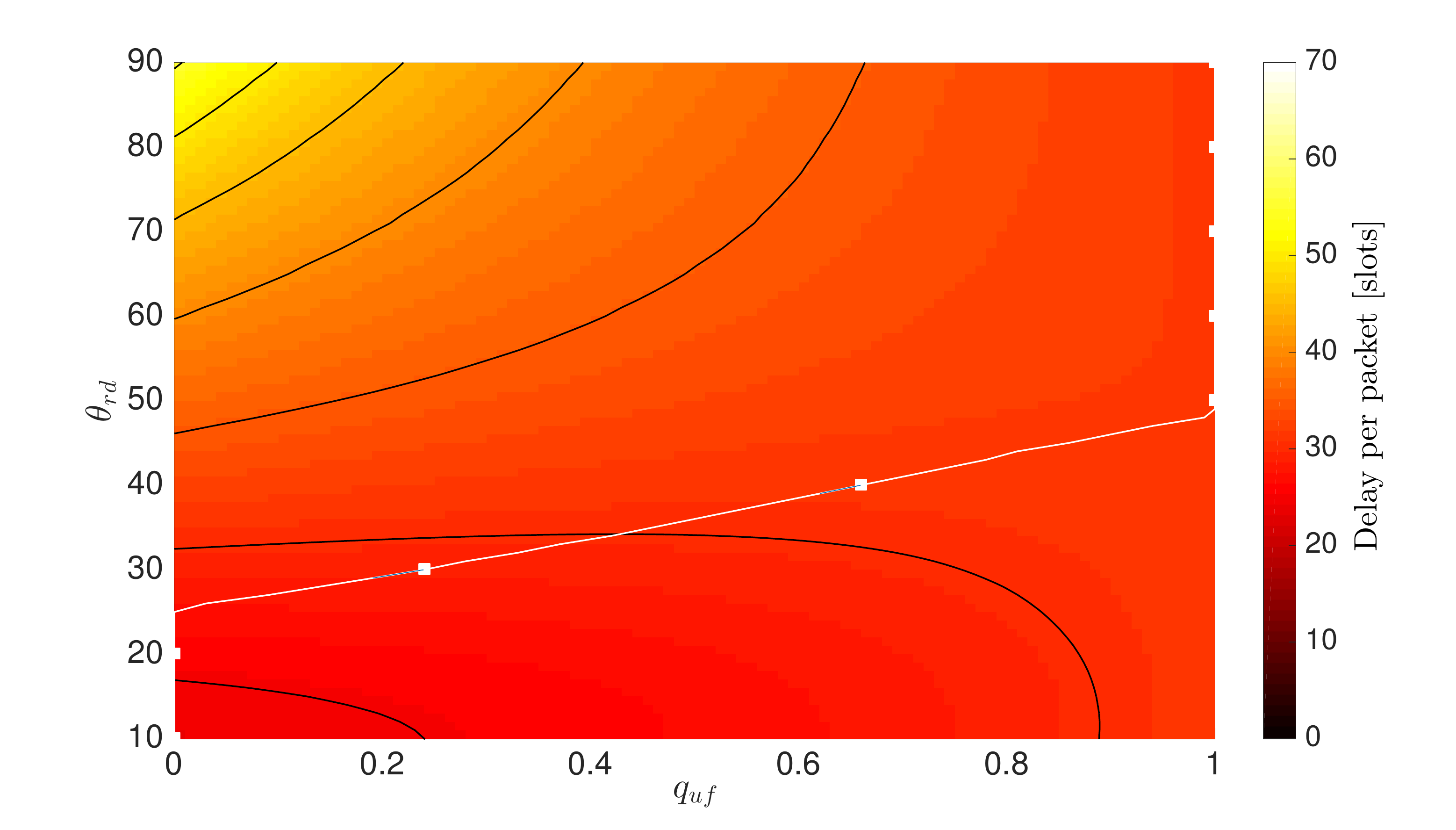}
    \caption{$\sigma^{2}=5^{\circ}$.}
    \label{fig:D_quf_th_ERR_BW5}
  \end{subfigure}
  \hfill
  \begin{subfigure}[b]{0.45\textwidth}
    \includegraphics[width=8cm]{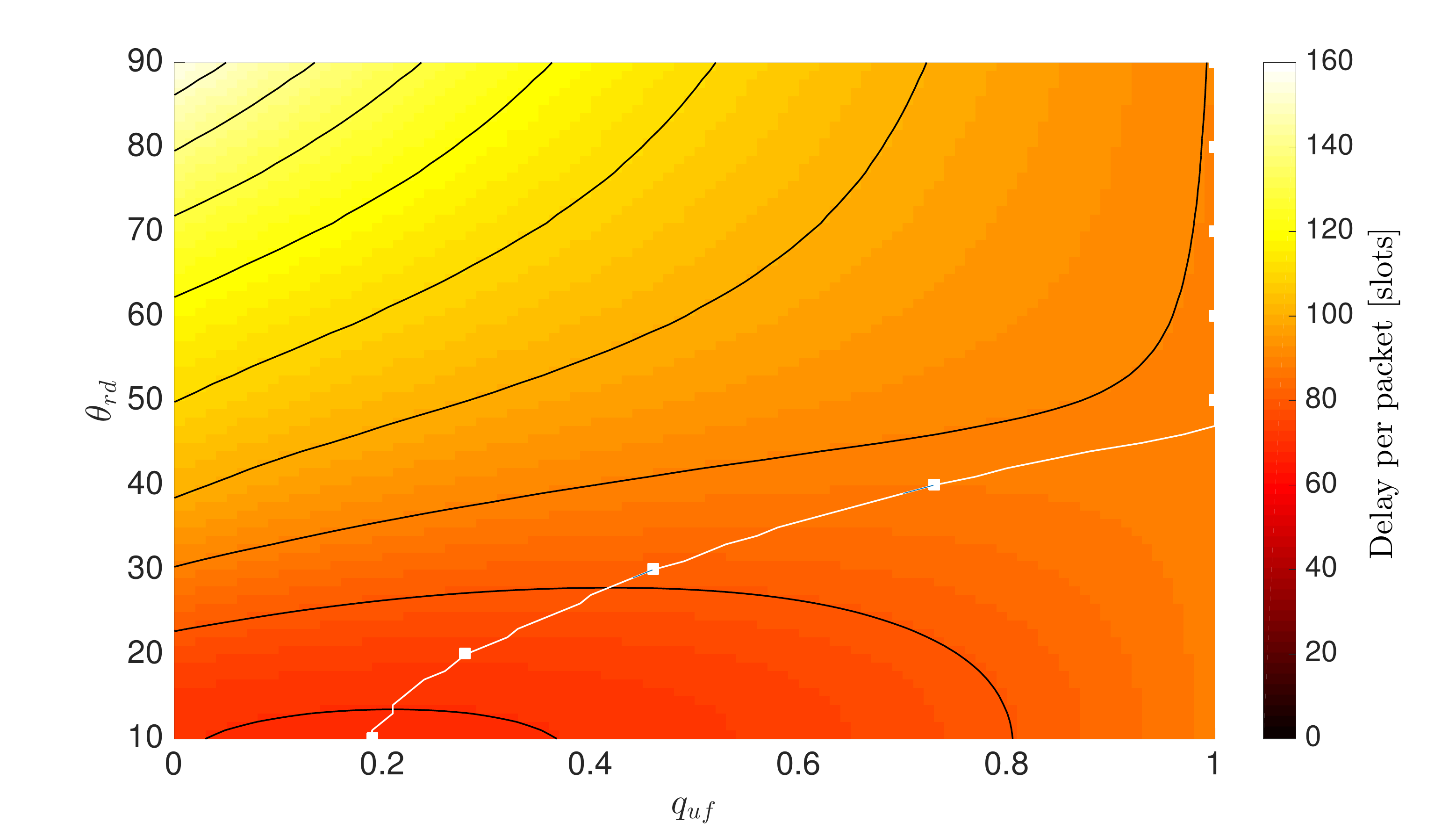}
    \caption{$\sigma^{2}=10^{\circ}$.}
    \label{fig:D_quf_th_ERR_BW10}
  \end{subfigure}
  \caption{Throughput with imperfect beam alignment while varying $q_{uf}$ and $\theta_{rd}$. Moreover, we set $q_{ur}=0.5$, $d_{ur}=30$~m and $d_{ud}=50$~m.}
\end{figure}
where, $\text{Erf}$ is the error function and $\sigma^2$ represents the variance of $\epsilon$. Note that, the beam alignment error affects only the computation of the success probability transmission in Appendix A, whereas the rest of the analysis remains the same. 

Now, we can show the results for throughput and delay when errors in the beam alignment phase are considered. In Fig~\ref{fig:T_quf_th_ERR_BW5} and~\ref{fig:T_quf_th_ERR_BW10} we show the throughput while varying $q_{uf}$ and $\theta_{rd}$ for $\sigma^{2}=5^{\circ}$ and $\sigma^{2}=10^{\circ}$, respectively. By comparing these figures with~Fig.~\ref{fig:T_quf_0}, where $\sigma^{2}=0^{\circ}$, we can observe that the throughput, $T$, decreases with the increase of $\sigma^{2}$. Moreover, for higher values of this parameter, the $\mathrm{BR}$ scheme increases its range of $\theta_{rd}$ values for which it represents the optimal choice. Namely, since $P_{G}(\sigma)$ increases for higher values of $\theta_{BW}$, wider beams are less subject to alignment errors. The delay for $\sigma^{2}=5^{\circ}$ has the same behavior of $T$, as it is shown in~\ref{fig:D_quf_th_ERR_BW5}. Whereas, in Fig~\ref{fig:D_quf_th_ERR_BW10} we can observe that, for $\sigma^{2}=10^{\circ}$, the delay is not minimum when $q_{uf}=0$. Indeed, although $\mathrm{BR}$ transmissions are more robust to alignment errors, they provide a lower beamforming gain that decreases the number of packets that are successfully received by the mmAP. However, these packets are still successfully received by $R$, whose queue size (and queueing delay) increases significantly with respect to Fig~\ref{fig:D_quf_th_ERR_BW5}.
\begin{figure}[!tbp]
  \centering
  \begin{minipage}[b]{0.45\textwidth}
    \includegraphics[width=8cm]{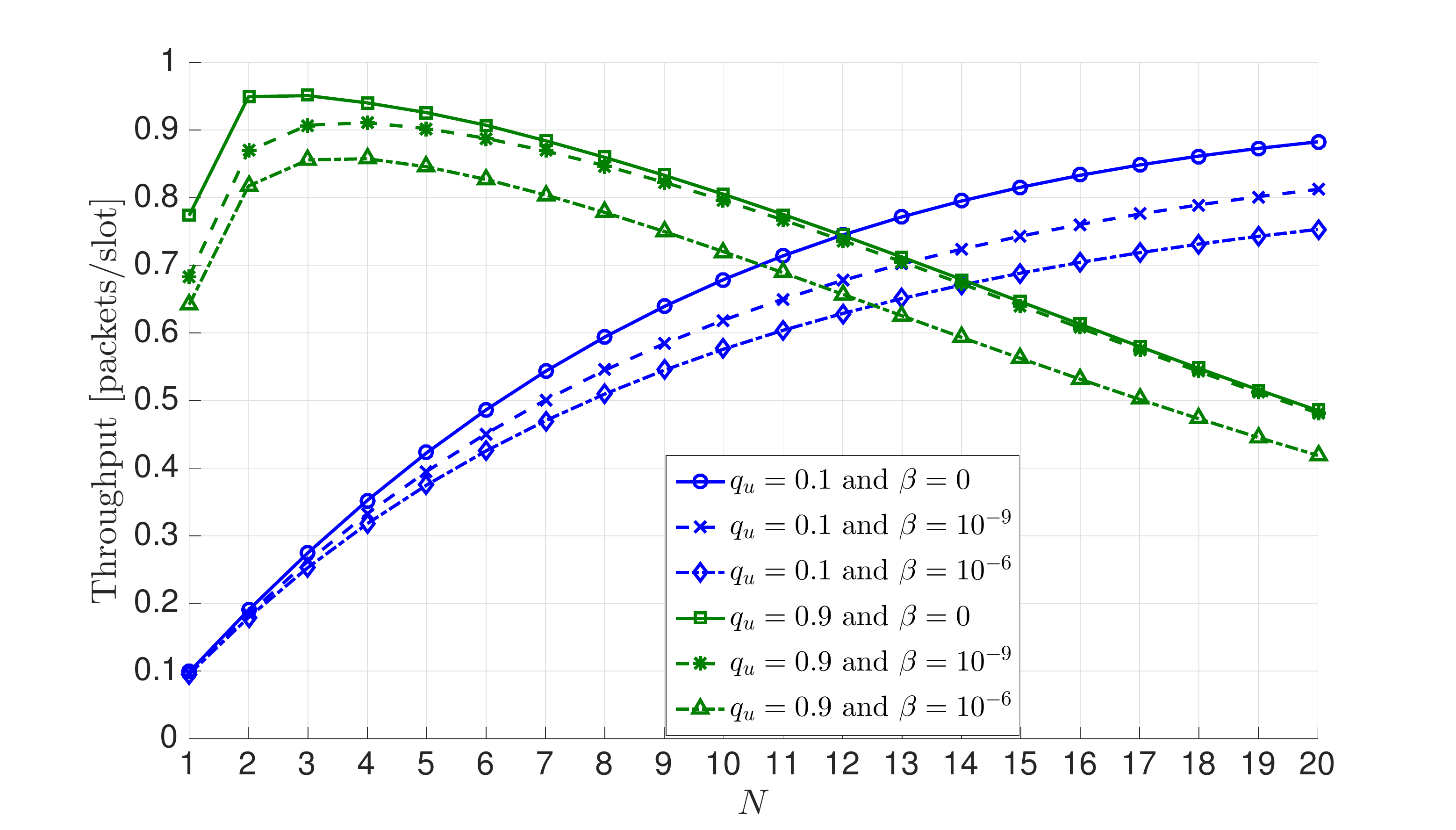}
    \caption{Throughput, $T$, while varying the number of UEs for several values of $\beta$ and transmit probability (i.e., $q_u=0.1$ and $q_u=0.9$) when $D_{a}=0$}
    \label{fig:T_N_quf_Self}
  \end{minipage}
  \hfill
  \begin{minipage}[b]{0.45\textwidth}
    \includegraphics[width=8cm]{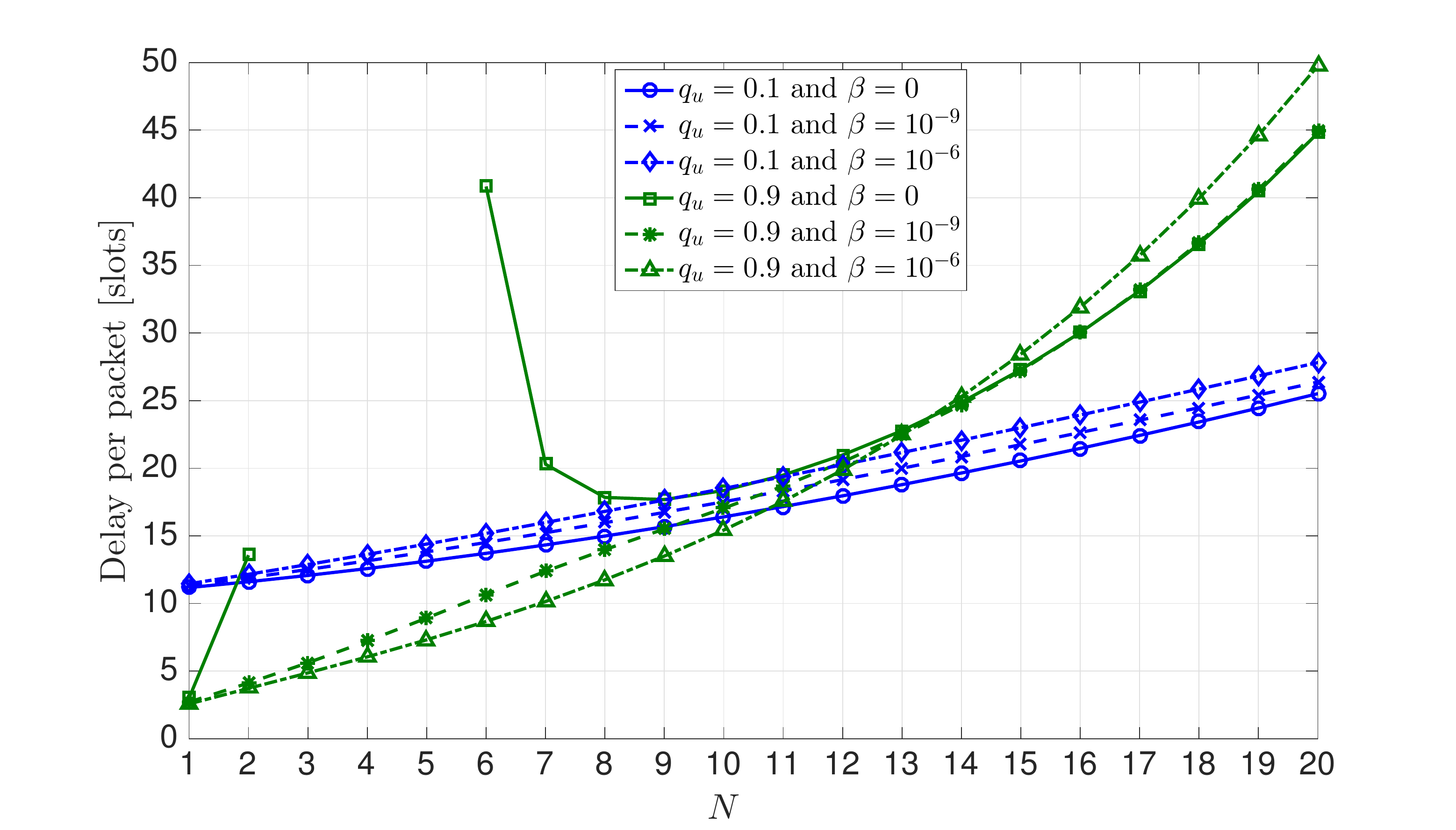}
    \caption{Delay per packet while varying the number of UEs for several values of $\beta$ and transmit probability (i.e., $q_u=0.1$ and $q_u=0.9$) when $D_{a}=0$}
    \label{fig:D_N_quf_Self}
  \end{minipage}
\end{figure}
\subsection{Imperfect Full-Duplex Communications}
\label{sec:FDself}
In this section, we consider non-perfect full-duplex relay operations, where packets that are transmitted by UEs to the relay are subject to an additional interference term when the relay is transmitting~\cite{Niko,FDmmWave}. We assume that the relay implements a self-interference mitigation technique , whose efficiency is modelled by a scalar $0 \le \beta \le 1$. This is similar to the parameter $\alpha$ used to model the interference cancellation at the receiver side of the relay in Section~\ref{sec:SINRexp}. Namely, $\beta=0$ models a perfect self-interference cancellation, whereas, when $\beta=1$ the self-interference cancellation is not used. The term $\beta$ affects the SINR expression in~\eqref{eq:SINR} and the success probability transmission, but the rest of the analysis in Section~\ref{sec:PA} and Section~\ref{sec:ThrDel} remain the same. In Fig.~\ref{fig:T_N_quf_Self} and Fig.~\ref{fig:D_N_quf_Self} we show the throughput and delay, respectively, while varying $N$ for several values of $\beta$ and transmit probability, i.e., $q_{u}$. From these figures, we can observe that, when $\beta>0$, the additional interference term (self-interference at $R$) decreases the number of packets successfully received by $R$ and, therefore, the throughput. Moreover, Fig.~\ref{fig:D_N_quf_Self} shows that for $q_{u}=0.9$, the lower number of packets at the relay makes the queue always stable. However, in general, the behavior of the curves for $\beta>0$ does not change significantly with respect to the case when $\beta=0$.

\subsection{Transmission Strategy-Dependent Alignment Delay}
\label{sec:TSDAD}
In this section, we show new results for the throughput and delay when the beam alignment duration is not constant. Indeed, $D_{a}$ depends on the beamforming technique, the beam alignment algorithm, and the number of beams that are used by transmitters and receivers~\cite{PadComp,PadMult}. In the following, we first describe the adopted model and next we show the results for throughput and delay when $D_{a}$ is not constant. We consider a bi-dimensional scenario where the nodes can form beams that are taken from a discrete set (codebook). These beams are not overlapped and cover the whole space. Moreover, we assume that the mmAP and the relay send downlink pilots for the beam alignment, which is performed at the UEs by using an exhaustive search algorithm. In according to these assumptions and the analysis in~\cite{PadMult}, the duration of the beam alignment, $D_{a}$, can be computed as follows:
\begin{align*} 
    PL_{UMi-LOS}=
    \begin{cases}
      \frac{N_{B}^{UE}N_{B}^{m}T_{sig}}{LM} & \text{when  }  s=fm, \\
      \frac{N_{B}^{UE}N_{B}^{r}T_{sig}}{LM} & \text{when  }  s=fr, \\
      \frac{N_{B}^{UE}N_{B}^{m}N_{B}^{r}}{LM} & \text{when  }  s=b,
    \end{cases}\stepcounter{equation}\tag{\theequation}\label{eq:Da_new}
\end{align*}
where, $s$ is the transmission strategy\footnote{Note that, when $s=b$ ($\mathrm{BR}$ case), UEs must align with both $R$ and the mmAP. Although faster beam alignment algorithm are also presented in~\cite{PadComp,PadMult} for a multi-connectivity case (e.g., uplink strategy), we consider here a worst case scenario, where each combination of beams at the UE, at the mmAP and at $R$ must be evaluated.}, $N_{B}^{UE}$, $N_{B}^{m}$, and $N_{B}^{r}$ are the numbers of beams at the UE, the mmAP and the relay, respectively. These depend on the beamwidth, i.e., $N_{B}=\frac{2\pi}{\theta_{BW}}$. The term $M=\frac{\Tau_{slot}}{\Tau_{sig}}$, is the number of pilots that can be sent in a timeslot whose duration is $\Tau_{slot}$ and $\Tau_{sig}$ is the transmitting time of a pilot. $L$ is the number of directions that the receiver can look simultaneously and depends on the beamformers at the mmAP and the relay. Since they have multi-packet reception capability, we can assume that $R$ and the mmAP are equipped with a digital or hybrid beamformers and $L$ is at least equal to the number of UEs. Thus, given~\eqref{eq:Da_new}, the delay can be rewritten as follows:
\begin{align*}
D&=\frac{1+q_{u}D_{r}\Bigl(q_{uf}q_{ur}T_{ur}^{f}+q_{ub}T_{ur}^{b}\Bigr)+q_{u}(D_{a}^{f}C'+D_{a}^{b}C'')}{q_{u}T_{u}},\stepcounter{equation}\tag{\theequation}\label{eq:Delay3}
\end{align*}
where, $D_{a}^{f}$ and $D_{a}^{b}$ are the beam alignment duration when $s=fm=fr$ and $s=b$, respectively. The terms $C'$ and $C''$ are given by:
\begin{align*}
C'&=q_{uf}^{2}q_{um}^{2}\Bigl(T_{ud}^{f}-T_{u}-1\Bigr)+q_{uf}^{2}q_{ur}^{2}\Bigl(T_{ur}^{f}-T_{u}-1\Bigr).  \stepcounter{equation}\tag{\theequation}\label{eq:C'}
\end{align*} 
\begin{align*}
C''&=1+q_{ub}^{2}\Bigl(T_{ud}^{b}+T_{ur}^{b}-T_{u}-1\Bigr),  \stepcounter{equation}\tag{\theequation}\label{eq:C''}
\end{align*} 
\begin{figure}[!tbp]
  \begin{subfigure}[b]{0.45\textwidth}
    \includegraphics[width=8cm]{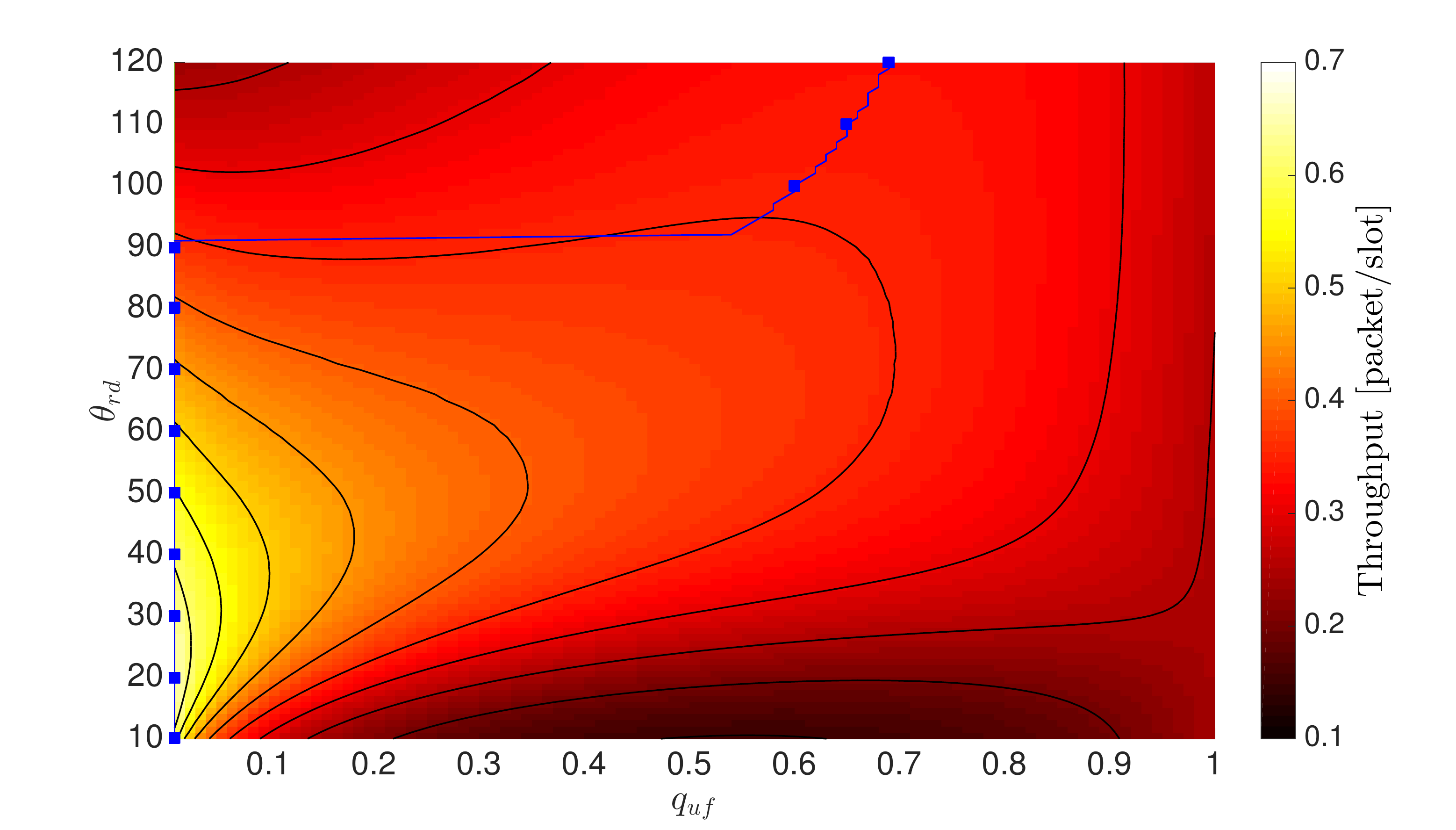}
    \caption{$\theta_{BW}^{f}=5^{\circ}$.}
    \label{fig:T_quf_th_Da_BW5}
  \end{subfigure}
  \hfill
  \begin{subfigure}[b]{0.45\textwidth}
    \includegraphics[width=8cm]{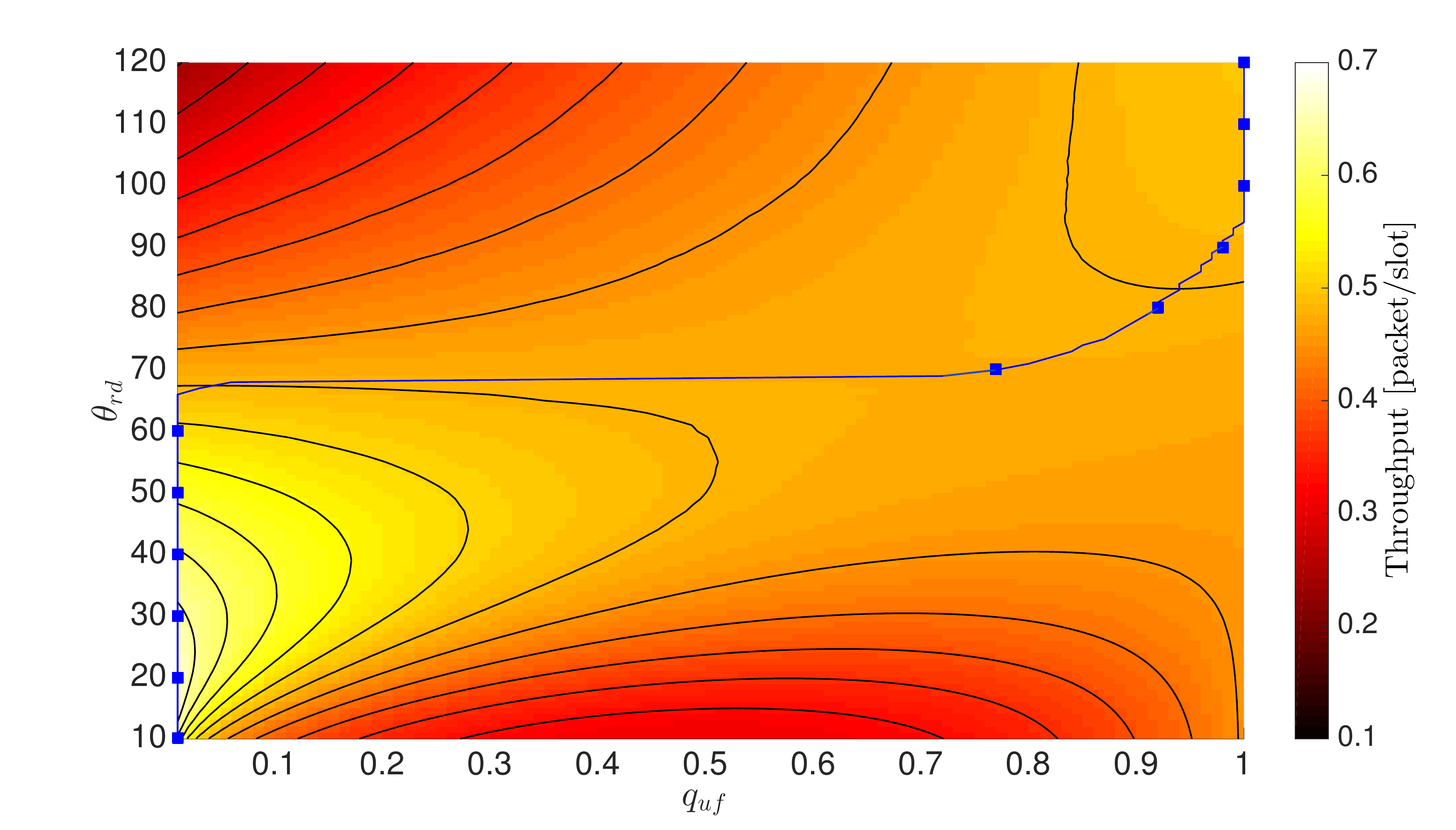}
    \caption{$\theta_{BW}^{f}=10^{\circ}$.}
    \label{fig:T_quf_th_Da_BW10}
  \end{subfigure}
  \caption{Throughput when $D_{a}$ is not constant while varying $q_{uf}$ and $\theta_{rd}$. Moreover, we set $q_{ur}=0.5$, $d_{ur}=30$~m and $d_{ud}=50$~m.}
\end{figure}
\begin{figure}[!tbp]
  \begin{subfigure}[b]{0.45\textwidth}
    \includegraphics[width=8cm]{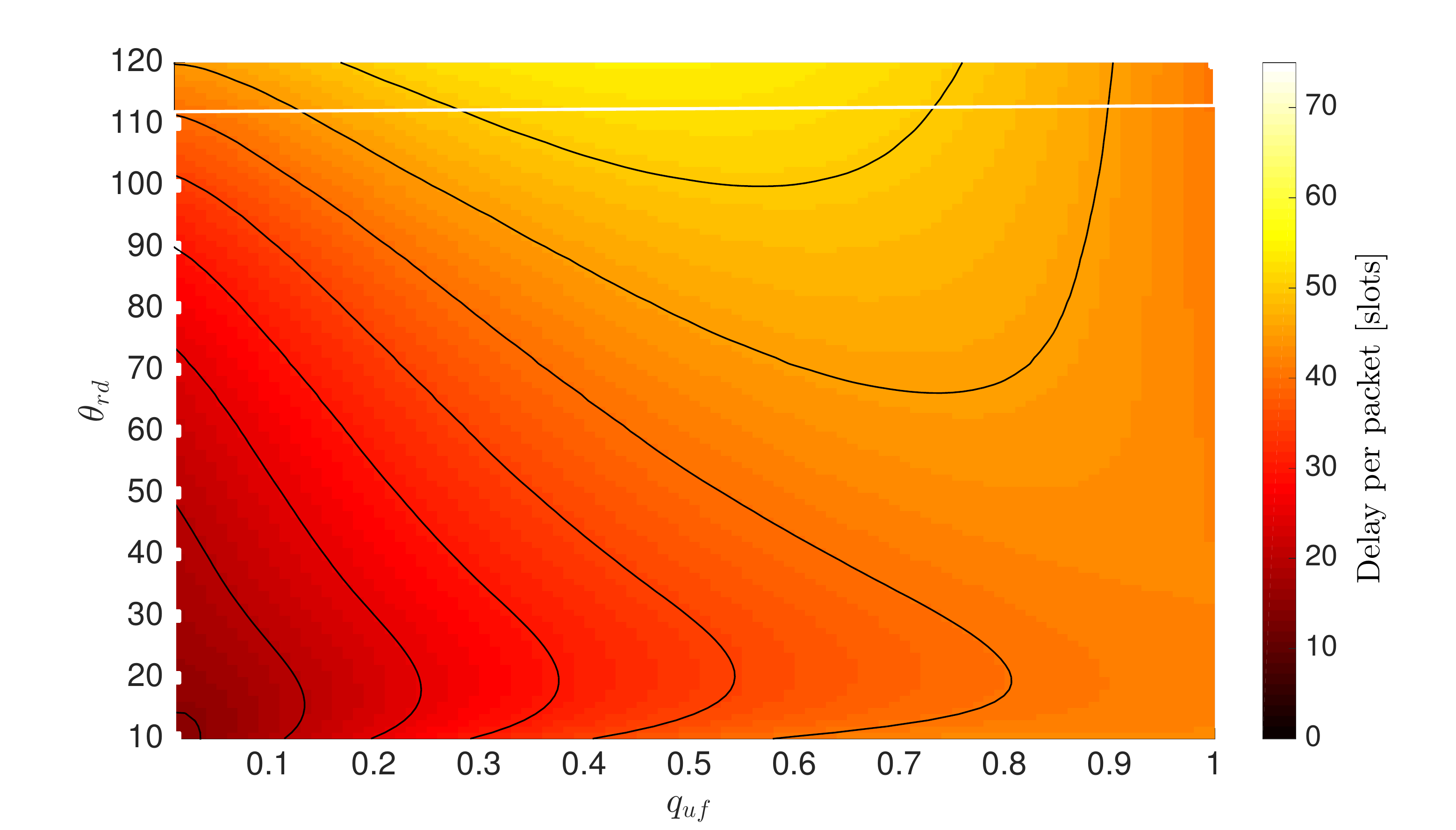}
    \caption{$\theta_{BW}^{dFD}=5^{\circ}$.}
    \label{fig:D_quf_th_Da_BW5}
  \end{subfigure}
  \hfill
  \begin{subfigure}[b]{0.45\textwidth}
    \includegraphics[width=8cm]{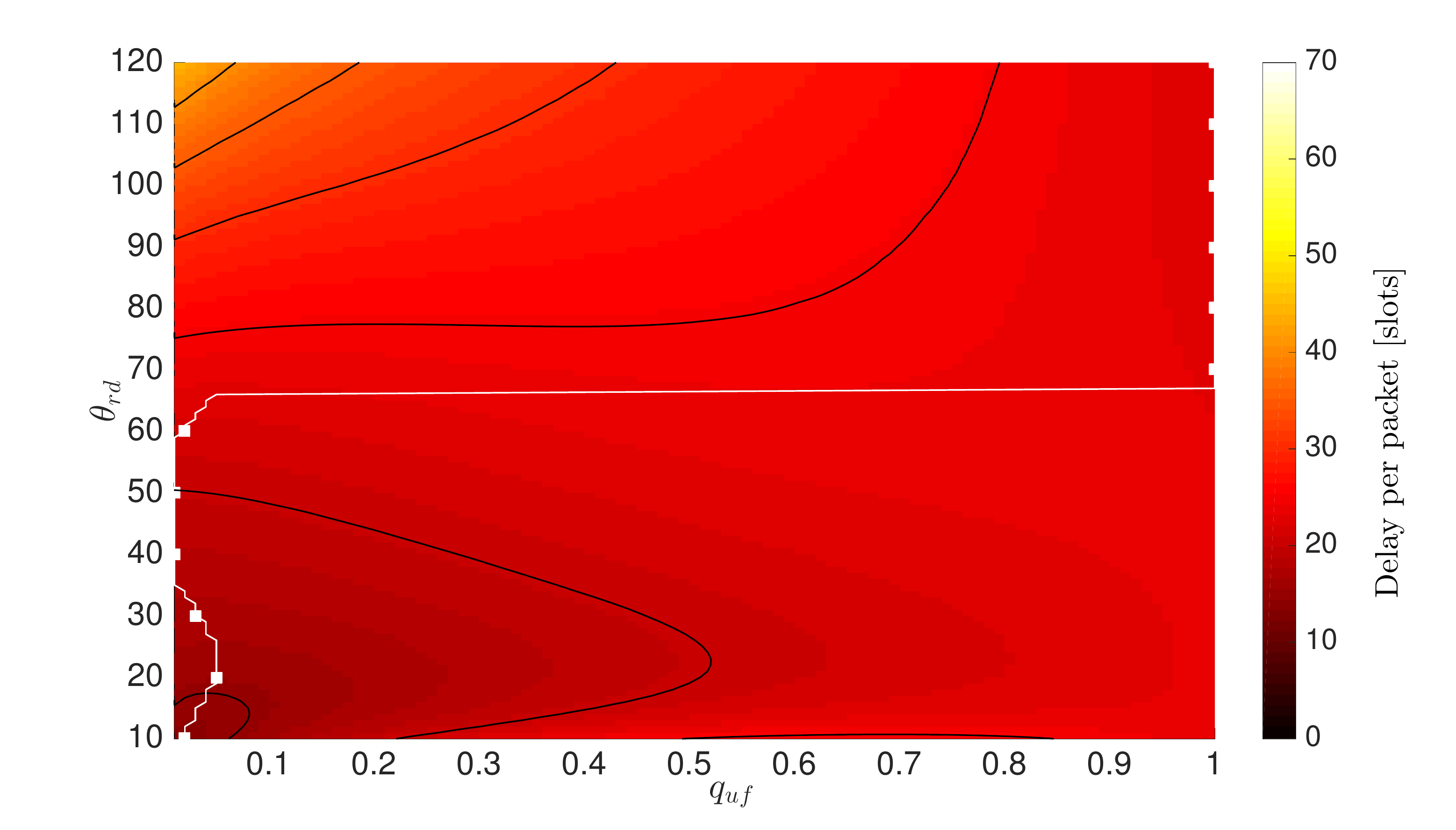}
    \caption{$\theta_{BW}^{f}=10^{\circ}$.}
    \label{fig:D_quf_th_Da_BW10}
  \end{subfigure}
  \caption{Delay when $D_{a}$ is not constant while varying $q_{uf}$ and $\theta_{rd}$. Moreover, we set $q_{ur}=0.5$, $d_{ur}=30$~m and $d_{ud}=50$~m.}
\end{figure}
where, we assume $N_{B}^{m}=N_{B}^{r}$. The rest of the analysis remains the same. Therefore, by setting  $\Tau_{sig}=10$ $\mu$s~\cite{PadComp}, $\Tau_{slot}=1$ ms, $M=100$, and $L=16$, we show in Fig.~\ref{fig:T_quf_th_Da_BW5} and~\ref{fig:T_quf_th_Da_BW10} the throughput while varying $\theta_{rd}$ and $q_{uf}$ with $\theta_{BW}^{f}=5^{\circ}$ and $\theta_{BW}^{f}=10^{\circ}$, respectively. Recall that $\theta_{BW}^{f}$ is the beamwidth for $\mathrm{FD}$ transmissions, whereas, for the $\mathrm{BR}$ case $\theta_{BW}^{b}=\theta_{rd}$. The rest of the parameters are set as in Fig.~\ref{fig:T_quf_0}. As in this figure, where $D_{a}$ is constant, we can observe that in both Fig.~\ref{fig:T_quf_th_Da_BW5} and~\ref{fig:T_quf_th_Da_BW10} $\mathrm{BR}$ transmissions are preferable for lower values of $\theta_{rd}$. More specifically, when $q_{uf}=0$, in every timeslot we can transmit to both the mmAP and $R$ and minimize the number of beam alignments. When $\theta_{rd}$ increases, the contribution to the throughput of $\mathrm{FD}$ transmissions increases as well. Moreover, this is even higher in Fig.~\ref{fig:T_quf_th_Da_BW10} when the beamwidh increases from $\theta_{BW}^{f}=5^{\circ}$ to $\theta_{BW}^{f}=10^{\circ}$. Indeed, although wider beams provide lower gain, they decrease the alignment delay and increase the transmit probability. Finally, in Fig.~\ref{fig:D_quf_th_Da_BW5} and~\ref{fig:D_quf_th_Da_BW10} we show the delay for the same parameters of the previous two figures. We can observe that the optimal strategy for the delay is highly affected from the beam alignment delay and it has almost the same behavior of the one that maximizes the throughput. Indeed in Fig.~\ref{fig:D_quf_th_Da_BW5} and~\ref{fig:D_quf_th_Da_BW10} we can observe that, in most of the cases, the delay is minimum when $q_{uf}=0$ or $q_{uf}=1$.

\section{Conclusion}
\label{sec:Conc}
We have presented a throughput and delay analysis for relay assisted mm-wave wireless networks, where the UEs can adopt either an $\mathrm{FD}$ or a $\mathrm{BR}$ transmission. In particular, we have analyzed the performance of the queue at the relay by deriving the stability conditions and the arrival and service rates. We have numerically evaluated the analytical model and validated our analysis with simulations.
The analytical model matches well the simulation results. The latter show that beam alignment causes a decrease in the transmit probability inversely proportional to the beam alignment duration, $D_{a}$, and the probability to change the strategy. We have shown how, in case of queue stability, the increase in $D_{a}$ decreases the throughput and the delay. However, for dense scenarios, where the queue at the relay is close to becoming unstable, the increase in $D_{a}$ can decrease the delay per packet. More precisely, when being close to the instability condition, we could show that the highest delay component is given by the queueing delay. Whereas, when the queue is stable, the delay is affected mostly by the transmission and alignment delay.

Moreover, we have showed that the optimal transmission strategy highly depends on the network topology, e.g., $d_{ud}$, $d_{ur}$, $\theta_{rd}$, and the queue condition. When the queue is stable, the values of $q_{uf}$ and $q_{ur}$ that maximize the throughput and minimize the delay usually coincide. 
However, when being close to the instability and unstability regions, the throughput and delay present a tradeoff. Furthermore, as expected, we showed that is not always beneficial to use narrow beams ($\mathrm{FD}$) compared to wider beams ($\mathrm{BR}$).  As a matter of fact, for short distances and a beamwidth of $30^{\circ}$, a broadcast transmission is still preferable, although it provides a lower beamforming gain than $\mathrm{FD}$. We have additionally showed that wider beams are more robust to beam alignment errors. However, when the angle between the mmAP and the relay is too wide or the distances and the SINR threshold increase, an $\mathrm{FD}$ strategy should be~chosen.

Finally, we could observe, that for the evaluated scenarios, the interference caused by the relay and the link path loss represent the main impediments for the success probability, hence the throughput and the delay, in case of short and long distances among the nodes, respectively.  

\appendices
\section{}
\label{sec:SPE}
We now derive the success probability expression, conditioned to the sets $\mathcal{I}_{f}$ and $\mathcal{I}_{b}$, for a generic link $ij$ with $N$ symmetric UEs. In order to average on all the possible scenarios for the LOS and NLOS links, we consider that $k$ and $h$ UEs over $|\mathcal{I}_{f}|$ and $|\mathcal{I}_{b}|$ interferers, respectively, are in LOS. Thus, the success probability can be derived as follows:
\begin{align*} 
&P_{ij/\mathcal{I}_{f},\mathcal{I}_{b}}^{f}=P(\text{LOS}_{ij})P(\text{SINR}_{ij/\mathcal{I}_{f},\mathcal{I}_{b}}^{f} \ge \gamma |\text{LOS}_{ij})+P(\text{NLOS}_{ij})P(\text{SINR}_{ij/\mathcal{I}_{f},\mathcal{I}_{b}}^{f} \ge \gamma | \text{NLOS}_{ij})\\
&=P(\text{LOS}_{ij})\Biggr[\sum_{k=0}^{|\mathcal{I}_{f}|}\binom{|\mathcal{I}_{f}|}{k}P(\text{LOS}_{ij})^{k}P(\text{NLOS}_{ij})^{|\mathcal{I}_{f}|-k}\sum_{h=0}^{|\mathcal{I}_{b}|}\binom{|\mathcal{I}_{b}|}{h}P(\text{LOS}_{ij} f^{h}P(\text{NLOS}_{ij})^{|\mathcal{I}_{b}|-h}\stepcounter{equation}\\
&\times P(\text{SINR}_{ij/\mathcal{I}_{fl},\mathcal{I}_{fn},\mathcal{I}_{bl},\mathcal{I}_{bn}}^{f} \ge \gamma |\text{LOS}_{ij})\Biggr]+P(\text{NLOS}_{ij})\Biggr[\sum_{k=0}^{|\mathcal{I}_{f}|}\binom{|\mathcal{I}_{f}|}{k}P(\text{LOS}_{ij})^{k}P(\text{NLOS}_{ij})^{|\mathcal{I}_{f}|-k}\\
&\times \sum_{h=0}^{|\mathcal{I}_{b}|}\binom{|\mathcal{I}_{b}|}{h}P(\text{LOS}_{ij})^{h}P(\text{NLOS}_{ij})^{|\mathcal{I}_{b}|-h}P(\text{SINR}_{ij/\mathcal{I}_{fl},\mathcal{I}_{fn},\mathcal{I}_{bl},\mathcal{I}_{bn}}^{f} \ge \gamma |\text{NLOS}_{ij})\Biggr]\tag{\theequation}\label{eq:prob_succA},
\end{align*} 
where, $P(\text{SINR}_{ij/\mathcal{I}_{f},\mathcal{I}_{b}}^{f} \ge \gamma |\text{LOS}_{ij})$ and $P(\text{SINR}_{ij/\mathcal{I}_{f},\mathcal{I}_{b}}^{f} \ge \gamma | \text{NLOS}_{ij})$ are the probabilities that the received SINR is above $\gamma$, when link $ij$ is in LOS and NLOS, respectively, conditioned to the specific scenarios of interferers, $\mathcal{I}_{f}$ and $\mathcal{I}_{b}$. The expression for $P(\text{SINR}_{ij/\mathcal{I}_{fl},\mathcal{I}_{fn},\mathcal{I}_{bl},\mathcal{I}_{bn}}^{f} \ge \gamma |\text{LOS}_{ij})$ is given in~\eqref{eq:SINR}.

\section{}
\label{sec:TAL}
\begin{proof}
To prove Theorem~\ref{theo}, 
we compute the probability that a UE transmits in timeslot $k$. We identify the following mutually exclusive events:
\begin{enumerate}
\item  the UE is performing a beam alignment. In this case the UE cannot transmit and $q_{tx}=0$.
\item $Al_{k-D_{a}}$: the UE starts a beam alignment in timeslot $k-D_{a}$. In this case $q_{tx}=1$. This is because the UE may start a beam alignment only if it decides to transmit.
\item $I_{k}$: the UE has not started an alignment in the previous $D_{a}$ timeslots with respect to the $k$-th timeslot ($k-D_{a}$, $k-D_{a}-1$,..., $k-1$). In this case, the UE can transmit in timeslot $k$, if and only if it decides to transmit with the same strategy used in the previous transmission attempt. Thus, $q_{tx}=q_{u}P(s_{k}=i \cap s_{\hat{k}}=i)$, with, $i \in \mathcal{S}$, and $\hat{k}$ represents the timeslot where the previous transmission attempt occurs with respect to the $k$-th timeslot.
\end{enumerate}
First, we analyze the second event. It occurs when in timeslot $k-D_{a}$ the following two independent events hold: i) $I_{k-D_{a}}$ and ii) \textit{the UE decides to transmit with a different strategy than that used in the previous transmission attempt}. Thus the event $Al_{k-D_{a}}$ occurs with probability:
\begin{align*} 
P\bigl(Al_{k-D_{a}}\bigr)=P\bigl(I_{k-D_{a}}\bigr)q_{u}\bigl(1-P(s_{k-D_{a}}=i \cap s_{\widehat{k-D_{a}}}=i)\bigr)\stepcounter{equation}\tag{\theequation}\label{eq:AlkDa}
\end{align*} 
Thus, we can express $q_{tx}$ as follows:
\begin{equation}\label{eq:qtx}
\begin{split} 
q_{tx}&=P\bigl(Al_{k-D_{a}}\bigr)+P\bigl(I_{k}\bigr)q_{u}P(s_{k}=i \cap s_{\hat{k}}=i)\\
&=P\bigl(I_{k-D_{a}}\bigr)q_{u}\bigl(1-P(s_{k-D_{a}}=i \cap s_{\widehat{k-D_{a}}}=i)\bigr)+ P\bigl(I_{k}\bigr)q_{u}P(s_{k}=i \cap s_{\hat{k}}=i)\\
&=P\bigl(I_{k}\bigr)q_{u},
\end{split}
\end{equation}
where the first equality exploits the mutual exclusivity of the three events. In the second equality of \eqref{eq:qtx}, we take into account~\eqref{eq:AlkDa} and in the last equality we have assumed that $P(s_{k}=i)$ is i.d. for any timeslot, giving $P\bigl(I_{k-D_{a}}\bigr)=P\bigl(I_{k}\bigr)$ and $P(s_{k-D_{a}}=i \cap s_{\widehat{k-D_{a}}}=i)=P(s_{k}=i \cap s_{\hat{k}}=i)$.

Consider the probability of the complementary event of $I_{k}$, i.e., $\overline{I}_{k}$. This event occurs when the UE starts a beam alignment in one of timeslots $k-D_{a}$, $k-D_{a}-1$, ..., $k-1$. Thus, $P\bigl(\overline{I}_{k}\bigr)$ is derived as the probability of the union of the following mutually exclusive events: $\bigcup^{k-1}_{i=k-D_{a}} Al_{i}$:
\begin{equation}\label{eq:PNIk}
\begin{split} 
&P\bigl(\overline{I}_{k}\bigr)=P\biggl(\bigcup^{k-1}_{i=k-D_{a}} Al_{i}\biggr)=\sum_{i=k-D_{a}}^{k-1}P\bigl(Al_{i}\bigr)
=\sum_{i=k-D_{a}}^{k-1}P\bigl(I_{i}\bigr)q_{u}\bigl(1-P(s_{i}=j \cap s_{\hat{i}}=j)\bigr)\\
&=D_{a}P\bigl(I_{k}\bigr)q_{u}\bigl(1-P(s_{k}=i \cap s_{\hat{k}}=i)\bigr),
\end{split}
\end{equation}
where, in the last step of~\eqref{eq:PNIk}, we use the same reasoning for the last equality of~\eqref{eq:qtx}.
Finally, by replacing~\eqref{eq:PNIk} in $P\bigl(\overline{I}_{k}\bigr)=1-P\bigl(I_{k}\bigr)$ we obtain~\eqref{eq:qtxf}.
\end{proof}

\section{}
\label{sec:TrP}
In this appendix, we provide the transition probabilities $p_{k}^{0}$ and $p_{k}^{1}$ for the two UE case. In this scenario, in every timeslot, the queue size can increase by a maximum of two, i.e., when both UEs successfully transmit to $R$ and $R$ itself does not successfully transmit any packet. Moreover, $R$ can transmit a packet only when the queue is not empty. Therefore, by considering all the possible transmission strategies and combinations of successfully received packets at $R$, we can write the following:
\begin{align*} 
p_{1}^{0}&=2q_{tx}\overline{q}_{tx}q_{uf}q_{ur}P_{ur}^{f}+2q_{tx}\overline{q}_{tx}q_{ub}P_{ur}^{b}\overline{P}_{ud}^{b}+2q_{tx}^{2}q_{uf}^{2}q_{ur}^{2}P_{ur/\{1\}^{f}}^{f}\overline{P}_{ur/\{1\}^{f}}^{f}+2q_{tx}^{2}q_{uf}^{2}q_{ur}q_{um}P_{ur}^{f}\\
&+2q_{tx}^{2}q_{uf}q_{ub}q_{ur}\Bigl[P_{ur/\{1\}^{b}}^{f}\Bigl(1-P_{ur/\{1\}^{f}}^{b}\overline{P}_{ud}^{b}\Bigl)+\overline{P}_{ur/\{1\}^{b}}^{f}P_{ur/\{1\}^{f}}^{b}\overline{P}_{ud}^{b}\Bigl]+2q_{tx}^{2}q_{ub}q_{uf}q_{um}P_{ur}^{b}\overline{P}_{ud/\{1\}^{f}}^{b}\\
&+q_{tx}^{2}q_{ub}^{2}\Bigl[2P_{ur/\{1\}^{b}}^{b}\overline{P}_{ud/\{1\}^{b}}^{b}\Bigl(1-P_{ur/\{1\}^{b}}^{b}\overline{P}_{ud/\{1\}^{b}}^{b}\Bigl)\Bigl]\stepcounter{equation}\tag{\theequation}\label{eq:p01b}.
\end{align*} 
\begin{align*} 
p_{2}^{0}&=\Bigl(q_{tx}q_{uf}q_{ur}P_{ur/\{1\}^{f}}^{f}\Bigl)^{2}+\Bigl(q_{tx}q_{ub}P_{ur/\{1\}^{b}}^{b}\overline{P}_{ud/\{r\}^{f},\{1\}^{b}}^{b}\Bigl)^{2}+2q_{tx}^{2}q_{ub}q_{uf}q_{ur}P_{ur/\{1\}^{f}}^{b}\overline{P}_{ud}^{b}P_{ur/\{1\}^{b}}^{f}\stepcounter{equation}\tag{\theequation}\label{eq:p02}.
\end{align*} 
\begin{align*} 
p_{-1}^{1}&=q_{r}\Bigl[P_{rd}^{f}\Bigl(\overline{q}_{tx}^{2}+2q_{tx}\overline{q}_{tx}q_{uf}q_{ur}\overline{P}_{ur}^{f}+(q_{tx}q_{uf}q_{ur}\overline{P}_{ur/\{1\}^{f}}^{f})^{2}\Bigl)+P_{rd/\{1\}^{f}}^{f}\\
&\times \Bigl(2q_{tx}\overline{q}_{tx}q_{uf}q_{um}+2q_{tx}^{2}q_{uf}^{2}q_{um}q_{ur}\overline{P}_{ur}^{f}\Bigl)+P_{rd/\{1\}^{b}}^{f}\Bigl(2q_{tx}\overline{q}_{tx}q_{ub}(1-P_{ur}^{b}\overline{P}_{ud/\{r\}^{f}}^{b})\\
&+2q_{tx}^{2}q_{ub}q_{uf}q_{ur}(1-P_{ur/\{1\}^{f}}^{b}\overline{P}_{ud/\{r\}^{f}}^{b})\overline{P}_{ur/\{1\}^{b}}^{f}\Bigl)+P_{rd/\{1\}^{f},\{1\}^{b}}^{f}2q_{tx}^{2}q_{uf}q_{ub}q_{um}(1-P_{ur}^{b}\overline{P}_{ud/\{1,r\}^{f}}^{b})\\
&+P_{rd/\{2\}^{b}}^{f}\Bigl(q_{tx}q_{ub}(1-P_{ur/\{1\}^{b}}^{b}\overline{P}_{ud/\{r\}^{f},\{1\}^{b}}^{b})\Bigl)^{2}\Bigl]+P_{rd/\{2\}^{f}}^{f}q_{tx}^{2}q_{uf}^{2}q_{um}^{2}\stepcounter{equation}\tag{\theequation}\label{eq:p_1f}.
\end{align*} 
\begin{align*} 
p_{1}^{1}&=\overline{q}_{r}p_{1}^{0}+q_{r}\Bigl[2q_{tx}\overline{q}_{tx}q_{uf}q_{ur}P_{ur}^{f}\overline{P}_{rd}^{f}+2q_{tx}\overline{q}_{tx}q_{ub}P_{ur}^{b}\overline{P}_{ud/\{r\}^{f}}^{b}\overline{P}_{rd/\{1\}^{b}}^{f}\\
&+2q_{tx}^{2}q_{uf}^{2}q_{um}q_{ur}P_{ur}^{f}\overline{P}_{rd/\{1\}^{f}}^{f}+2q_{tx}^{2}q_{uf}q_{ub}q_{um}P_{ur}^{b}\overline{P}_{ud/\{1,r\}^{f}}^{b}\overline{P}_{rd/\{1\}^{f},\{1\}^{b}}^{f}\\
&+q_{tx}^{2}q_{uf}^{2}q_{ur}^{2}\Bigl(P_{ur/\{1\}^{f}}^{f}\overline{P}_{ur/\{1\}^{f}}^{f}\overline{P}_{rd}^{f}+(P_{ur/\{1\}^{f}}^{f})^{2}P_{rd}^{f}\Bigl)+q_{tx}^{2}q_{ub}^{2}\Bigl(2P_{ur/\{1\}^{b}}^{b}\overline{P}_{ud/\{r\},\{1\}}^{b}\overline{P}_{rd/\{2\}^{b}}^{f}\\
&\times(1-P_{ur/\{1\}^{b}}^{b}\overline{P}_{ud/\{r\}^{f},\{1\}^{b}}^{b})+(P_{ur/\{1\}^{b}}^{b}\overline{P}_{ud/\{r\}^{f},\{1\}^{b}}^{b})^{2} {P}_{rd/\{2\}^{b}}^{f}\Bigl)\\
&+2q_{tx}^{2}q_{ub}q_{uf}q_{ur}\Bigl(P_{ur/\{1\}^{f}}^{b}\overline{P}_{ud/\{r\}^{f}}^{b}\overline{P}_{ur/\{r\}^{f},\{1\}^{b}}^{f}\overline{P}_{rd/\{1\}^{b}}^{f}+(1-P_{ur/\{1\}^{f}}^{b}\overline{P}_{ud/\{r\}^{f}}^{b})P_{ur/\{1\}^{b}}^{f}\overline{P}_{rd/\{1\}^{b}}^{f}\\
&+P_{ur/\{2\}^{f}}^{b}\overline{P}_{ud/\{r\}^{f}}^{b}P_{ur/\{1\}^{b}}^{f}P_{rd/\{1\}^{b}}^{f}\Bigl)\Bigl]\stepcounter{equation}\tag{\theequation}\label{eq:p11f}.
\end{align*}
\begin{align*} 
p_{2}^{1}&=\overline{q}_{r}p_{2}^{0}+q_{r}\Bigl[\Bigl(q_{tx}q_{uf}q_{ur}P_{ur/\{1\}^{f}}^{f}\Bigl)^{2}\overline{P}_{rd}^{f}+\Bigl(q_{tx}q_{ub}P_{ur/\{1\}^{b}}^{b}\overline{P}_{ud/\{r\}^{f},\{1\}^{b}}^{b}\Bigl)^{2}\overline{P}_{rd/\{2\}^{b}}^{f}\\
&+2q_{tx}^{2}q_{ub}q_{uf}q_{ur}P_{ur/\{1\}^{f}}^{b}\overline{P}_{ud}^{b}P_{ur/\{1\}^{b}}^{f}\overline{P}_{rd/\{1\}^{b}}^{f}\Bigl]\stepcounter{equation}\tag{\theequation}\label{eq:p12}.
\end{align*}

\section{}
\label{sec:QNS}
Hereafter, we analyze the performance of the queue at the relay for $N$ UEs. The average arrival rate that is given in~\eqref{eq:lambdaN}.
In order to compute the number of packets successfully received by $R$, i.e., $r_{k}^{0}$ and $r_{k}^{1}$, we consider all the possible combinations of UE's transmission strategies and interference scenarios. Thus, we indicate with $m$ out of $N$ the number of UEs that transmit a packet, with $i$ (at most $m$) the number of transmitting UEs that use $\mathrm{FD}$ transmissions ($m-i$ UEs use the $\mathrm{BR}$ transmission), and with $j$ the number of $\mathrm{FD}$ UEs that transmit to $R$ ($i-j$ UEs transmit to the mmAP). Moreover, $k$ is the total number of packets successfully received by the relay in a timeslot, and $k_{f}$ out of them are received by using $\mathrm{FD}$ transmissions ($k-k_{f}$ packets are received by using a $\mathrm{BR}$ transmission). Thus, we can express $r_{k}^{0}$ and $r_{k}^{1}$ as follows:
\begin{align*}
r_{k}^{0}&=\sum_{m=k}^{N}\binom{N}{m}q_{tx}^{m}\overline{q}_{tx}^{\ N-m}\sum_{i=0}^{m}\binom{m}{i}q_{uf}^{i}q_{ub}^{m-i} \sum_{j=\max(0,k+i-m)}^{i} \binom{i}{j}q_{ur}^{j}q_{um}^{i-j} \sum_{k_{f}=\max(0,k+i-m)}^{\min(j,k)}\binom{j}{k_{f}}\\
&\times (P_{ur/\{j-1\}^{f},\{m-i\}^{b}}^{f})^{k_{f}} (\overline{P}_{ur/\{j-1\}^{f},\{m-i\}^{b}}^{f})^{j-k_{f}}\times \binom{m-i}{k-k_{f}}(P_{ur/\{j\}^{f},\{m-i-1\}^{b}}^{b}\overline{P}_{ud/\{i-j\}^{f},\{m-i-1\}^{b}}^{b})^{k-k_{f}}\\
&\times (1-P_{ur/\{j\}^{f},\{m-i-1\}^{b}}^{b}\overline{P}_{ud/\{i-j\}^{f},\{m-i-1\}^{b}}^{b})^{m-i-k+k_{f}},\stepcounter{equation}\tag{\theequation}\label{eq:NRarrival0}
\end{align*} 

\begin{align*}
r_{k}^{1}&=\overline{q}_{r}r_{k}^{0}+q_{r}\sum_{m=k}^{N}\binom{N}{m}q_{tx}^{m}\overline{q}_{tx}^{\ N-m}\sum_{i=0}^{m}\binom{m}{i}q_{uf}^{i}q_{ub}^{m-i}\sum_{j=\max(0,k+i-m)}^{i} \binom{i}{j}q_{ur}^{j}q_{um}^{i-j} \sum_{k_{f}=\max(0,k+i-m)}^{\min(j,k)}\binom{j}{k_{f}}\\
&\times (P_{ur/\{j-1\}^{f},\{m-i\}^{b}}^{f})^{k_{f}} (\overline{P}_{ur/\{j-1\}^{f},\{m-i\}^{b}}^{f})^{j-k_{f}} \binom{m-i}{k-k_{f}}(P_{ur/\{j\}^{f},\{m-i-1\}^{b}}^{b}\overline{P}_{ud/\{i-j,r\}^{f},\{m-i-1\}^{b}}^{b})^{k-k_{f}}\\
&\times (1-P_{ur/\{j\}^{f},\{m-i-1\}^{b}}^{b}\overline{P}_{ud/\{i-1,r\}^{f},\{m-i-1\}^{b}}^{b})^{m-i-k+k_{f}}.\stepcounter{equation}\tag{\theequation}\label{eq:NRarrival1}
\end{align*}
Then, for deriving the relay's service rate $\mu_{r}$, we follow the same reasoning that is done above and the successful packet transmission of $R$ is averaged over all the possible interference scenarios:
\begin{align*}
\mu_{r}&=q_{r}\sum_{m=0}^{N}\binom{N}{m}q_{tx}^{m}\overline{q}_{tx}^{\ N-m}\sum_{i=0}^{m}\binom{m}{i}q_{uf}^{m}q_{ub}^{N-m}\sum_{j=0}^{i} \binom{i}{j}q_{ur}^{j}q_{um}^{i-j} P_{rd/\{i-j\}^{f},\{m-i\}^{b}}^{f},\stepcounter{equation}\tag{\theequation}\label{eq:ServiceRN}
\end{align*}
Now, we can study the evolution of the queue at $R$ by using a discrete time Markov Chain, whose transition matrix is a lower Hessenberg matrix, whose elements are represented by the probabilities that the queue size increases by $k$ packets in a timeslot when the queue is empty or not, i.e., $p_{k}^{0}$ and $p_{k}^{1}$. By using the same notation as in~\eqref{eq:NRarrival0}, these terms are given by:
\begin{equation}\label{eq:p_k0}
\begin{split} 
p_{k}^{0}&=r_{k}^{0},
\end{split}
\end{equation}
\begin{align*}
&p_{-1}^{1}=q_{r}\sum_{m=0}^{N}\binom{N}{m}q_{tx}^{m}\overline{q}_{tx}^{\ N-m}\sum_{i=0}^{m}\binom{m}{i}q_{uf}^{i}q_{ub}^{m-i}\sum_{j=0}^{i} \binom{i}{j}q_{ur}^{j}q_{um}^{i-j} P_{rd/\{i-j\}^{f},\{m-i\}^{b}}^{f}\Bigl(\overline{P}_{ur/\{j-1\}^{f},\{m-i\}^{b}}^{f}\Bigr)^{j}\\
&\times(1-P_{ur/\{j\}^{f},\{m-i-1\}^{b}}^{b}\overline{P}_{ud/\{i-j,r\}^{f},\{m-i-1\}^{b}}^{b})^{m-i},\stepcounter{equation}\tag{\theequation}\label{eq:p_11N}
\end{align*}
\begin{align*}
&p_{k}^{1}=\overline{q}_{r}r_{k}^{0}+q_{r}\sum_{m=k}^{N}\binom{N}{m}q_{tx}^{m}\overline{q}_{tx}^{\ N-m}\sum_{i=0}^{m}\binom{m}{i}q_{uf}^{i}q_{ub}^{m-i}\sum_{j=\max(0,k+i-m)}^{i} \binom{i}{j}q_{ur}^{j}q_{um}^{i-j} \sum_{k_{f}=\max(0,k+i-m)}^{\min(j,k)}\binom{j}{k_{f}}\\
&\times (P_{ur/\{j-1\}^{f},\{m-i\}^{b}}^{f})^{k_{f}} (\overline{P}_{ur/\{j-1\}^{f},\{m-i\}^{b}}^{f})^{j-k_{f}} \times \binom{m-i}{k-k_{f}}(P_{ur/\{j\}^{f},\{m-i-1\}^{b}}^{b}\overline{P}_{ud/\{i-j,r\}^{f},\{m-i-1\}^{b}}^{b})^{k-k_{f}}\\
&\times (1-P_{ur/\{j\}^{f},\{m-i-1\}^{b}}^{b}\overline{P}_{ud/\{i-j,r\}^{f},\{m-i-1\}^{b}}^{b})^{m-i-k+k_{f}}\overline{P}_{rd/\{i-j\}^{f},\{m-i\}^{b}}^{f} + q_{r}\sum_{m=k+1}^{N}\binom{N}{m}q_{tx}^{m}\overline{q}_{tx}^{\ N-m}\\
&\times \sum_{i=0}^{m}\binom{m}{i}q_{uf}^{i}q_{ub}^{m-i} \sum_{j=\max(0,k+1+i-m)}^{i} \binom{i}{j}q_{ur}^{j}q_{um}^{i-j} \sum_{k_{f}=\max(0,k+1+i-m)}^{\min(j,k+1)}\binom{j}{k_{f}} \binom{m-i}{k+1-k_{f}}\\
&\times (P_{ur/\{j-1\}^{f},\{m-i\}^{b}}^{f})^{k_{f}} (\overline{P}_{ur/\{j-1\}^{f},\{m-i\}^{b}}^{f})^{j-k_{f}} (P_{ur/\{j\}^{f},\{m-i-1\}^{b}}^{b}\overline{P}_{ud/\{i-j,r\}^{f},\{m-i-1\}^{b}}^{b})^{k+1-k_{f}}\\
&\times (1-P_{ur/\{j\}^{f},\{m-i-1\}^{b}}^{b}\overline{P}_{ud/\{i-j,r\}^{f},\{m-i-1\}^{b}}^{b})^{m-i-k-1+k_{f}}P_{rd/\{i-j\}^{f},\{m-i\}^{b}}^{f},\stepcounter{equation}\tag{\theequation}\label{eq:pk1}
\end{align*}
\begin{equation}\label{eq:p01N}
\begin{split} 
p_{0}^{1}&=1-p_{-1}^{1}-\sum_{k=1}{N}p_{k}^{1}.
\end{split}
\end{equation}
Finally, we can compute the probability that the queue is empty, $P(Q=0)$, and the average relay queue size, $\overline{Q}$. Hereafter, we show the main steps of the derivations that are illustrated with more details in~\cite{book}. First, we can note that the queue at $R$ can be modelled as an $M^{N}/M/1$ queue, therefore, the equation that describes the evolution of the states is given by:
\begin{equation}\label{eq:Ps}
\begin{split} 
s_{i}=a_{i}s_{0}+\sum_{j=1}^{i+1}b_{i-j+1}s_{j},
\end{split}
\end{equation}
where, $s_{i}$ represents the probability of finding our system in state $i$ at equilibrium. Let $s$ be the steady-state distribution vector and $S(z)$ its Z-transformation, we have:
\begin{align*}
S(z)=\sum_{i=1}^{\infty}s_{i}z^{-i} \Rightarrow \overline{Q}=-S'(1)=-s_{0}\frac{K''(1)}{L''(1)}.\stepcounter{equation}\tag{\theequation}\label{eq:Sz}
\end{align*}
The terms $K''(z)$ and $L''(z)$ in~\eqref{eq:Sz} are the second derivatives of $K(z)$ and $L(z)$, respectively. These are given by~\cite{book}:
\begin{align*}
K(z)&=(-z^{-2}A(z)+z^{-1}A'(z)-B'(z))(z^{-1}-B(z))\\
&-(z^{-1}A'(z)-B(z))(-z^{-2}-B'(z)),\stepcounter{equation}\tag{\theequation}\label{eq:Kz}
\end{align*}
\begin{align*}
L(z)=(z^{-1}-B(z))^{2},\stepcounter{equation}\tag{\theequation}\label{eq:Lz}
\end{align*}
where, $A(z)=\sum_{i=1}^{N}a_{i}z^{-i}$ and $B(z)=\sum_{i=1}^{N+1}b_{i}z^{-i}$  ($a_{i}=p_{i}^{0}$ and $b_{i}=p_{i-1}^{1}$). The term $s_{0}$, in~\eqref{eq:Sz}, is the probability that the queue is empty at equilibrium, which can be written as follows~\cite{book}:
\begin{align*}
P(Q=0)=\frac{1+B'(1)}{1+B'(1)-A'(1)}.\stepcounter{equation}\tag{\theequation}\label{eq:PQ01}
\end{align*}
Then, by replacing the first derivative of $A(z)$ and $B(z)$ in \eqref{eq:PQ01}, we obtain:
\begin{equation}\label{eq:QN}
\begin{split} 
P(Q=0)=\frac{p_{-1}^{1}-\sum_{i=1}^{N}ip_{i}^{1}}{p_{-1}^{1}-\sum_{i=1}^{N}ip_{i}^{1}+\lambda_{r}^{0}}.
\end{split}
\end{equation}
Finally, by considering~\eqref{eq:Sz},~\eqref{eq:Kz},~\eqref{eq:Lz}, and~\eqref{eq:QN}, we can express $\overline{Q}$ as follows:
\begin{align*}
\overline{Q}&=\frac{\Bigl(\sum_{k=1}^{N}kp_{k}^{1}-p_{-1}^{1}\Bigr)\sum_{k=1}^{N}k(k+3)p_{k}^{0}+\lambda_{r}^{0}\Bigl(2p_{-1}^{1}-\sum_{k=1}^{N}k(k+3)p_{k}^{1}\Bigr)}{2\Bigl(\sum_{k=1}^{N}kp_{k}^{1}-p_{-1}^{1}\Bigr)\Bigl(p_{-1}^{1}-\sum_{k=1}{N}kp_{k}^{1}+\lambda_{r}^{0}\Bigr)}.\stepcounter{equation}\tag{\theequation}\label{eq:Qsize}
\end{align*}

\bibliography{ref}

\begin{thebibliography}{10}
\providecommand{\url}[1]{#1}
\csname url@samestyle\endcsname
\providecommand{\newblock}{\relax}
\providecommand{\bibinfo}[2]{#2}
\providecommand{\BIBentrySTDinterwordspacing}{\spaceskip=0pt\relax}
\providecommand{\BIBentryALTinterwordstretchfactor}{4}
\providecommand{\BIBentryALTinterwordspacing}{\spaceskip=\fontdimen2\font plus
\BIBentryALTinterwordstretchfactor\fontdimen3\font minus
  \fontdimen4\font\relax}
\providecommand{\BIBforeignlanguage}[2]{{%
\expandafter\ifx\csname l@#1\endcsname\relax
\typeout{** WARNING: IEEEtran.bst: No hyphenation pattern has been}%
\typeout{** loaded for the language `#1'. Using the pattern for}%
\typeout{** the default language instead.}%
\else
\language=\csname l@#1\endcsname
\fi
#2}}
\providecommand{\BIBdecl}{\relax}
\BIBdecl

\bibitem{RelConfO}
\BIBentryALTinterwordspacing
C.~Tatino, N.~Pappas, I.~Malanchini, L.~Ewe, and D.~Yuan, ``Throughput analysis
  for relay-assisted millimeter-wave wireless networks,'' \emph{to appear at
  GLOBECOM Workshops}, 2018. [Online]. Available:
  \url{http://arxiv.org/abs/1804.09450}
\BIBentrySTDinterwordspacing

\bibitem{5GEnabler}
E.~Hossain and M.~Hasan, ``{5G} cellular: key enabling technologies and
  research challenges,'' \emph{IEEE Instrumentation Measurement Magazine},
  vol.~18, no.~3, pp. 11--21, June 2015.

\bibitem{RappMeas}
T.~S. Rappaport, Y.~Xing, G.~R. MacCartney, A.~F. Molisch, E.~Mellios, and
  J.~Zhang, ``Overview of millimeter wave communications for fifth-generation
  ({5G}) wireless networks-with a focus on propagation models,'' \emph{IEEE
  Transactions on Antennas and Propagation}, vol.~65, no.~12, pp. 6213--6230,
  Dec. 2017.

\bibitem{Meas28}
H.~Zhao, R.~Mayzus, S.~Sun, M.~Samimi, J.~K. Schulz, Y.~Azar, K.~Wang, G.~N.
  Wong, F.~Gutierrez, and T.~S. Rappaport, ``28 {GH}z millimeter wave cellular
  communication measurements for reflection and penetration loss in and around
  buildings in new york city,'' in \emph{IEEE International Conference on
  Communications (ICC)}, June 2013, pp. 5163--5167.

\bibitem{Meas73}
G.~R. MacCartney and T.~S. Rappaport, ``73 {GH}z millimeter wave propagation
  measurements for outdoor urban mobile and backhaul communications in new york
  city,'' in \emph{IEEE International Conference on Communications (ICC)}, June
  2014, pp. 4862--4867.

\bibitem{SurMM}
\BIBentryALTinterwordspacing
Y.~Niu, Y.~Li, D.~Jin, L.~Su, and A.~V. Vasilakos, ``A survey of millimeter
  wave communications (mmwave) for 5g: Opportunities and challenges,''
  \emph{Wireless Networks}, vol.~21, no.~8, pp. 2657--2676, Nov. 2015.
  [Online]. Available: \url{https://doi.org/10.1007/s11276-015-0942-z}
\BIBentrySTDinterwordspacing

\bibitem{Our1}
C.~Tatino, I.~Malanchini, D.~Aziz, and D.~Yuan, ``Beam based stochastic model
  of the coverage probability in {5G} millimeter wave systems,'' in \emph{the
  15th International Symposium on Modeling and Optimization in Mobile, Ad Hoc,
  and Wireless Networks (WiOpt)}, May 2017, pp. 1--6.

\bibitem{MCcapacity}
D.~Moltchanov, A.~Ometov, S.~Andreev, and Y.~Koucheryavy, ``Upper bound on
  capacity of {5G} mmwave cellular with multi-connectivity capabilities,''
  \emph{Electronics Letters}, vol.~54, no.~11, pp. 724--726, 2018.

\bibitem{Our2}
C.~Tatino, I.~Malanchini, N.~Pappas, and D.~Yuan, ``Maximum throughput
  scheduling for multi-connectivity in millimeter-wave networks,'' in \emph{the
  16th International Symposium on Modeling and Optimization in Mobile, Ad Hoc,
  and Wireless Networks (WiOpt)}, May 2018, pp. 1--6.

\bibitem{Coop1}
G.~Kramer, I.~Mari\'{c}, and R.~D. Yates, ``Cooperative communications,''
  \emph{Found. Trends Netw.}, vol.~1, no.~3, pp. 271--425, Aug. 2006.

\bibitem{Coop2}
A.~K. Sadek, K.~J.~R. Liu, and A.~Ephremides, ``Cognitive multiple access via
  cooperation: Protocol design and performance analysis,'' \emph{IEEE
  Transactions on Information Theory}, vol.~53, no.~10, pp. 3677--3696, Oct.
  2007.

\bibitem{Coop3}
B.~Rong and A.~Ephremides, ``Cooperative access in wireless networks: Stable
  throughput and delay,'' \emph{IEEE Transactions on Information Theory},
  vol.~58, no.~9, pp. 5890--5907, Sept. 2012.

\bibitem{CoopSpa}
O.~Simeone, Y.~Bar-Ness, and U.~Spagnolini, ``Stable throughput of cognitive
  radios with and without relaying capability,'' \emph{IEEE Transactions on
  Communications}, vol.~55, no.~12, pp. 2351--2360, Dec 2007.

\bibitem{Rel1}
N.~Pappas, A.~Ephremides, and A.~Traganitis, ``Relay-assisted multiple access
  with multi-packet reception capability and simultaneous transmission and
  reception,'' in \emph{IEEE Information Theory Workshop}, Oct. 2011, pp.
  578--582.

\bibitem{MonCoop}
\BIBentryALTinterwordspacing
G.~Kramer, I.~Marić, and R.~D. Yates, ``Cooperative communications,''
  \emph{Foundations and Trends® in Networking}, vol.~1, no. 3–4, pp.
  271--425, 2007. [Online]. Available:
  \url{http://dx.doi.org/10.1561/1300000004}
\BIBentrySTDinterwordspacing

\bibitem{CoopWC}
W.~{Zhuang} and M.~{Ismail}, ``Cooperation in wireless communication
  networks,'' \emph{IEEE Wireless Communications}, vol.~19, no.~2, pp. 10--20,
  April 2012.

\bibitem{Niko}
N.~Pappas, M.~Kountouris, A.~Ephremides, and A.~Traganitis, ``Relay-assisted
  multiple access with full-duplex multi-packet reception,'' \emph{IEEE
  Transactions on Wireless Communications}, vol.~14, no.~7, pp. 3544--3558,
  July 2015.

\bibitem{Rel2}
G.~Papadimitriou, N.~Pappas, A.~Traganitis, and V.~Angelakis, ``Network-level
  performance evaluation of a two-relay cooperative random access wireless
  system,'' \emph{Computer Networks}, vol.~88, pp. 187--201, Sept. 2015.

\bibitem{ITC30}
N.~Pappas, I.~Dimitriou, and Z.~Cheng, ``Network-level cooperation in random
  access iot networks with aggregators,'' in \emph{30th International
  Teletraffic Congress (ITC 30)}, vol.~1, Sept. 2018.

\bibitem{CoopSurv}
N.~Zlatanov, A.~Ikhlef, T.~Islam, and R.~Schober, ``Buffer-aided cooperative
  communications: opportunities and challenges,'' \emph{IEEE Communications
  Magazine}, vol.~52, no.~4, pp. 146--153, Apr. 2014.

\bibitem{CoopSel}
N.~Nomikos, T.~Charalambous, I.~Krikidis, D.~N. Skoutas, D.~Vouyioukas,
  M.~Johansson, and C.~Skianis, ``A survey on buffer-aided relay selection,''
  \emph{IEEE Communications Surveys Tutorials}, vol.~18, no.~2, pp. 1073--1097,
  Secondquarter 2016.

\bibitem{PhaseArray}
S.~{Kutty} and D.~{Sen}, ``Beamforming for millimeter wave communications: An
  inclusive survey,'' \emph{IEEE Communications Surveys Tutorials}, vol.~18,
  no.~2, pp. 949--973, Secondquarter 2016.

\bibitem{InitAcc1}
M.~Giordani, M.~Mezzavilla, and M.~Zorzi, ``Initial access in {5G} mmwave
  cellular networks,'' \emph{IEEE Communications Magazine}, vol.~54, no.~11,
  pp. 40--47, Nov. 2016.

\bibitem{InitAcc2}
C.~N. Barati, S.~A. Hosseini, M.~Mezzavilla, T.~Korakis, S.~S. Panwar,
  S.~Rangan, and M.~Zorzi, ``Initial access in millimeter wave cellular
  systems,'' \emph{IEEE Transactions on Wireless Communications}, vol.~15,
  no.~12, pp. 7926--7940, Dec. 2016.

\bibitem{CoopDiv}
V.~K. Sakarellos, D.~Skraparlis, A.~D. Panagopoulos, and J.~D. Kanellopoulos,
  ``Cooperative diversity performance in millimeter wave radio systems,''
  \emph{IEEE Transactions on Communications}, vol.~60, no.~12, pp. 3641--3649,
  Dec. 2012.

\bibitem{RelayPhy1}
B.~Xie, Z.~Zhang, and R.~Q. Hu, ``Performance study on relay-assisted
  millimeter wave cellular networks,'' in \emph{IEEE 83rd Vehicular Technology
  Conference (VTC Spring)}, May 2016, pp. 1--5.

\bibitem{RelayPhy2}
S.~Biswas, S.~Vuppala, J.~Xue, and T.~Ratnarajah, ``On the performance of relay
  aided millimeter wave networks,'' \emph{IEEE Journal of Selected Topics in
  Signal Processing}, vol.~10, no.~3, pp. 576--588, Apr. 2016.

\bibitem{ConnRel}
X.~Lin and J.~G. Andrews, ``Connectivity of millimeter wave networks with
  multi-hop relaying,'' \emph{IEEE Wireless Communications Letters}, vol.~4,
  no.~2, pp. 209--212, April 2015.

\bibitem{CovRel}
K.~Belbase, Z.~Zhang, H.~Jiang, and C.~Tellambura, ``Coverage analysis of
  millimeter wave decode-and-forward networks with best relay selection,''
  \emph{IEEE Access}, vol.~6, pp. 22\,670--22\,683, 2018.

\bibitem{ProbD2D}
N.~Wei, X.~Lin, and Z.~Zhang, ``Optimal relay probing in millimeter-wave
  cellular systems with device-to-device relaying,'' \emph{IEEE Transactions on
  Vehicular Technology}, vol.~65, no.~12, pp. 10\,218--10\,222, Dec. 2016.

\bibitem{CovD2D}
S.~Wu, R.~Atat, N.~Mastronarde, and L.~Liu, ``Coverage analysis of d2d
  relay-assisted millimeter-wave cellular networks,'' in \emph{IEEE Wireless
  Communications and Networking Conference (WCNC)}, Mar. 2017, pp. 1--6.

\bibitem{RelayBlo}
J.~W. Sungoh~Kwon, ``Relay selection for mmwave communications,'' in \emph{the
  28th Annual IEEE International Symposium on Personal, Indoor and Mobile Radio
  Communications (IEEE PIMRC)}, Oct. 2017, pp. 1--5.

\bibitem{FairRelay}
Y.~Xu, H.~Shokri-Ghadikolaei, and C.~Fischione, ``Distributed association and
  relaying with fairness in millimeter wave networks,'' \emph{IEEE Transactions
  on Wireless Communications}, vol.~15, no.~12, pp. 7955--7970, Dec. 2016.

\bibitem{FallRelay}
R.~Congiu, H.~Shokri-Ghadikolaei, C.~Fischione, and F.~Santucci, ``On the
  relay-fallback tradeoff in millimeter wave wireless system,'' in \emph{IEEE
  Conference on Computer Communications Workshops (INFOCOM WKSHPS)}, Apr. 2016,
  pp. 622--627.

\bibitem{Hybrid}
S.~Sun, T.~S. Rappaport, R.~W. Heath, A.~Nix, and S.~Rangan, ``Mimo for
  millimeter-wave wireless communications: beamforming, spatial multiplexing,
  or both?'' \emph{IEEE Communications Magazine}, vol.~52, no.~12, pp.
  110--121, Dec. 2014.

\bibitem{Int_Fischione}
H.~Shokri-Ghadikolaei and C.~Fischione, ``The transitional behavior of
  interference in millimeter wave networks and its impact on medium access
  control,'' \emph{IEEE Transactions on Communications}, vol.~64, no.~2, pp.
  723--740, Feb. 2016.

\bibitem{Bai}
T.~Bai and R.~W. Heath, ``Coverage and rate analysis for millimeter-wave
  cellular networks,'' \emph{IEEE Transactions on Wireless Communications},
  vol.~14, no.~2, pp. 1100--1114, Feb. 2015.

\bibitem{loynes}
R.~M. Loynes, ``The stability of a queue with non-independent inter-arrival and
  service times,'' \emph{Mathematical Proceedings of the Cambridge
  Philosophical Society}, vol.~58, no.~3, pp. 497--520, 1962.

\bibitem{3GPP}
3GPP, ``Study on channel model for frequencies from 0.5 to 100 {GHz} (release
  14), 3gpp tr 38.901 v14.2.0,'' Tech. Rep., Sept. 2017.

\bibitem{BeamErr}
A.~Thornburg and R.~W. Heath, ``Ergodic capacity in mmwave ad hoc network with
  imperfect beam alignment,'' in \emph{{MILCOM} - IEEE Military Communications
  Conference}, Oct. 2015, pp. 1479--1484.

\bibitem{FDmmWave}
Z.~Xiao, P.~Xia, and X.~Xia, ``Full-duplex millimeter-wave communication,''
  \emph{IEEE Wireless Communications}, vol.~24, no.~6, pp. 136--143, Dec 2017.

\bibitem{PadComp}
M.~{Giordani}, M.~{Mezzavilla}, C.~N. {Barati}, S.~{Rangan}, and M.~{Zorzi},
  ``Comparative analysis of initial access techniques in {5G} mmwave cellular
  networks,'' in \emph{Annual Conference on Information Science and Systems
  ({CISS})}, March 2016, pp. 268--273.

\bibitem{PadMult}
M.~{Giordani}, M.~{Mezzavilla}, S.~{Rangan}, and M.~{Zorzi},
  ``Multi-connectivity in {5G} mmwave cellular networks,'' in
  \emph{Mediterranean Ad Hoc Networking Workshop (Med-Hoc-Net)}, June 2016, pp.
  1--7.

\bibitem{book}
F.~Gebali, \emph{Analysis of Computer and Communication Networks}.\hskip 1em
  plus 0.5em minus 0.4em\relax New York, NY, USA: Springer-Verlag, 2010.

\end{thebibliography}
\bibliographystyle{IEEEtran}

\end{document}